\documentclass[12pt,english,table]{article}
\usepackage{lmodern}

\usepackage[T1]{fontenc}
\usepackage[latin9]{inputenc}
\usepackage{float}
\usepackage{bm}
\usepackage{amsmath}
\usepackage{amsthm}
\usepackage{amssymb}
\usepackage{graphicx}
\usepackage{geometry}
\usepackage{rotfloat}
\usepackage{setspace}
\usepackage[authoryear]{natbib}
\makeatletter
\theoremstyle{plain}
\newtheorem{thm}{\protect\theoremname}
\theoremstyle{remark}
\newtheorem*{rem*}{\protect\remarkname}
\theoremstyle{plain}
\newtheorem{lem}{\protect\lemmaname}

\setcitestyle{round}
\usepackage{lscape}

\usepackage{xcolor}
\usepackage{rotating}

\usepackage{array}
\usepackage{tabularx}\usepackage{multirow}\usepackage{booktabs}
\usepackage{ragged2e}\newcolumntype{C}[1]{>{\centering\arraybackslash}p{#1}}
\usepackage{rotfloat}
\newcolumntype{J}[1]{>{\justify\arraybackslash}p{#1}}
\newcolumntype{R}[1]{>{\RaggedLeft\arraybackslash}p{#1}}
\newcolumntype{Q}[1]{>{\columncolor{Gray}\RaggedLeft\arraybackslash}p{#1}}
\newcolumntype{L}[1]{>{\RaggedRight\arraybackslash}p{#1}}
\newcolumntype{G}{@{\extracolsep{0.5cm}}l@{\extracolsep{0pt}}}%
\newcolumntype{P}[1]{>{\centering\arraybackslash}p{#1}}
\newcolumntype{Y}{>{\centering\arraybackslash}X}
\newcommand{\nhphantom}[1]{\sbox0{#1}\hspace{-\the\wd0}} 

\usepackage{longtable}

\AtBeginDocument{

}

\usepackage[colorlinks,
            linkcolor=blue,
            anchorcolor=blue,
            citecolor=blue]{hyperref}

\usepackage{babel}

\makeatletter
\renewcommand*{\fps@figure}{htb}
\makeatother

\usepackage{titlesec}
\usepackage{subcaption}
\usepackage{tikz}
\usepackage{circuitikz}
\usetikzlibrary{shapes, arrows, positioning}
\usepackage{colortbl}             
\usepackage{arydshln}             
\usepackage{multirow}             

\makeatother

\usepackage{babel}
\providecommand{\lemmaname}{Lemma}
\providecommand{\remarkname}{Remark}
\providecommand{\theoremname}{Theorem}

\begin{document}
\title{Dynamic Factor Correlation Model\thanks{The paper is based on results from the paper ``Convolution-$t$ distributions'',
which was presented at the CIREQ-CMP Econometrics Conference in Honor
of Eric Ghysels, Montreal, 2024.}{\normalsize\emph{\medskip{}
}}}
\author{\textbf{Chen Tong}$^{\ddagger}$$\quad$\textbf{ }and$\quad$\textbf{Peter
Reinhard Hansen}$^{\mathsection}$\thanks{Corresponding author: Peter Reinhard Hansen. Email: hansen@unc.edu.
Chen Tong acknowledges financial support from the Youth Fund of the
National Natural Science Foundation of China (72301227), and the Ministry
of Education of China, Humanities and Social Sciences Youth Fund (22YJC790117).}\bigskip{}
 \\
\\
 {\normalsize$^{\ddagger}$}{\normalsize\emph{Department of Finance,
School of Economics, Xiamen University}}\\
$^{\mathsection}${\normalsize\emph{Department of Economics, University
of North Carolina at Chapel Hill\medskip{}
 }}}
\date{{\normalsize\emph{\today}}}
\maketitle
\begin{abstract}
We introduce a new dynamic factor correlation model with a novel variation-free
parametrization of factor loadings. The model is applicable to high
dimensions and can accommodate time-varying correlations, heterogeneous
heavy-tailed distributions, and dependent idiosyncratic shocks, such
as those observed in returns on stocks in the same subindustry. We
apply the model to a ``small universe'' with 12 asset returns and
to a ``large universe'' with 323 asset returns. The former facilitates
a comprehensive empirical analysis and comparisons and the latter
demonstrates the flexibility and scalability of the model.

\bigskip{}
\end{abstract}
{\small\textit{{\noindent}Keywords:}}{\small{} Factor Structure, Multivariate
GARCH, Block Correlation Matrix, Heavy Tailed Distributions}{\small\par}

\noindent{\small\textit{{\noindent}JEL Classification:}}{\small{}
C32, C38, C55, C58}{\small\par}

\clearpage{}

\section{Introduction}

In this paper, we propose a flexible factor-correlation model that
can be scaled to high dimensions. The model relies on univariate volatility
models to define the time series of standardized returns, $Z_{1,t},,\ldots,Z_{n,t}$,
and factor variable, $F_{1,t},\ldots,F_{r,t}$, and a multivariate
GARCH model to capture the correlation structure among the factor
variables, $F_{t}$. Our main contributions are the dynamic modeling
of the factor loadings and the idiosyncratic correlation matrix with
various structures. A key theoretical contribution is a variation-free
parametrization of factor loading, which simplifies aspects of their
modeling. 

Univariate GARCH models have proven highly successful in capturing
heteroscedasticity in individual return series since their introduction
in \citet{Engle:1982} and the refinements in \citet{bollerslev:86}.
However, extending these models to the multivariate setting has been
far less straightforward. A key challenge lies in preserving the positive
definiteness of the covariance matrix in a natural and elegant manner.
The number of covariance/correlation increases with the square of
the dimension, and some structure is needed to make estimation feasible
with high dimensional systems. Many multivariate GARCH formulations
impose structure through parameter restrictions, that can be overly
restrictive in their functional form, limiting their flexibility in
capturing complex dependencies. This has led to a wide range of competing
models, each addressing particular aspects of these challenges without
providing a fully satisfactory generalization. Popular multivariate
GARCH models, include the BEKK model\footnote{BEKK is named after Baba, Engle, Kraft, and Kroner}
by \citet{EngleKroner:1995} and the Dynamic Conditional Correlation
model by \citet{Engle2002}, which both have restrictive dynamic structures,
especially in high dimensional setting.

Our paper relates to a very large body of literature including the
studies on factor models. Factors are widely used in finance, and
some use a \textit{latent} factor structure, such as those by \citet[2018, 2023]{OhPatton2017JBES}\nocite{OhPatton2018}\nocite{OhPatton2023},
\citet{CrealTsay2015}, and \citet{OpschoorLucasBarraVanDick:2021}
that  use copula (latent) factor structure to model the dynamics of
correlations, while assuming idiosyncratic shocks to be uncorrelated.
Another strand of literature leverages \textit{observed} factors,
with some models relying on realized measures from high-frequency
data to capture dynamic factor loadings. For instance, the Realized-Beta
GARCH model of \citet{HansenLundeVoev:2014} is a one-factor model
that utilizes realized correlations (between individual assets and
the market return). They also document a strong dependence between
idiosyncratic shocks, especially for stocks in the same sector. A
related model is the Factor HEAVY model by \citet{SheppardXu:2019},
which also rely on realized measure (realized betas) for updating
factor loadings. Other studies, such as \citet{Engle:2016} and \citet{Darolles:2018},
concentrate on modeling time-varying betas without accounting for
the idiosyncratic correlations.

The block correlation structure introduced by \citet{EngleKelly:2012}
is a valuable contribution to the literature on multivariate volatility
models, because it offers a parsimonious approach to modeling large-dimensional
covariance matrices while preserving interpretability. The structure
in \citet{EngleKelly:2012} does reduce the number of parameters,
but it does not inherently guarantee positive semi-definiteness (PSD)
and is cumbersome in situations with more than two blocks. This issue
was later resolved by \citet{ArchakovHansen:CanonicalBlockMatrix},
who developed a method to ensure positive definite block correlation
matrices for any number of blocks. This laid the foundation for the
multivariate Realized GARCH model proposed by \citet{ArchakovHansenLundeMRG},
which integrates realized measures of volatility within a coherent
multivariate GARCH framework. In related work, \citet{TongHansenArchakov:2024}
introduced a score-driven multivariate GARCH model based on convolution-$t$
distributions of \citet{HansenTong:2024}. This class of distributions
has greater flexibility for capturing complex nonlinear dependencies,
and can accommodate heterogeneous heavy-tailness and cluster structures
in tail dependencies, which were limitations of traditional specifications
with Gaussian and multivariate $t$-distributions. 

In this paper, we will also adopt convolution-$t$ distributions.
However, our dynamic model of the correlation structure is entirely
different from that in \citet{TongHansenArchakov:2024}. For instance,
we include observable factors into the modeling and specify dynamic
models for the corresponding factor loadings. Another key difference
is that \citet{TongHansenArchakov:2024} impose block structures on
the correlation matrix for returns, whereas we impose block structures
on the idiosyncratic correlation matrix, which enable us to adopt
sparse block correlation matrices. The structure of the idiosyncratic
correlation matrix plays a critical role in our modeling, and its
sparse nature facilitates implementation in very high dimensions.
Block correlation structures in idiosyncratic asset returns (related
to industry sectors) is well documented in the empirical studies,
including \citet{FanFurgerXiu:2016}, \citet{YacineXiu2017}, and
\citet{AndreouGagliardiniGhyselsRubin:2024}, \citet{Bodilsen:2024},
and \citet{HansenLundeVoev:2014}. Most of the earlier literature
involving dynamic factor models assume idiosyncratic asset returns
to be uncorrelated.

We largely rely on the existing literature to model the univariate
return series, as well as the dynamic correlation structure of factor
variables. The contributions of this paper is mainly relates to the
part of the model named the \emph{Core Correlation Model}. This part
has dynamic factor loadings and dynamic idiosyncratic correlation
matrix, and it has a structure that makes it scalable to high dimensions.
A key component in this part of the model is a novel variation-free
parametrization of the dynamic factor loadings, which is inspired
by the generalized Fisher transformation (GFT) of correlation matrices
by \citet{ArchakovHansen:Correlation}. We rely heavily on the score
driven framework by \citet{CrealKoopmanLucas:2011} to specify dynamic
models of factor loadings and various correlation matrices, and in
this context we utilizes Tikhonov regularized Moore-Penrose inverse
to ensure stable updating in the score-driven model. We draw heavily
on \citet{ArchakovHansen:CanonicalBlockMatrix} to formulate block
correlation matrices and sparse versions of these. Importantly, the
model is scalable to high dimensions, owing in part to a decoupled
estimation method of the core correlation model.

We apply the new model to 17 years of daily stock returns. A \emph{small
universe} with $n=12$ assets and a \emph{large universe} with $n=323$
assets. The small universe facilitates model comparisons under different
specifications and estimation methods, which is not possible in high
dimension, such as that of the large universe. As factor variables
we adopt the six cross-sectional factors, known as the Fama-French
five factors (FF5) and the momentum factor, and we include sector-specific
factors that are based on exchange traded funds (ETFs) for each of
the sectors. A sample correlation matrix motivates the use of subindustries
to define the block structure in the idiosyncratic correlation matrix. 

The empirical results are very encouraging. Applying the model to
the large universe with $n=323$ stocks and 63 subindustries poses
no obstacles with decoupled estimation. We find strong empirical support
for the specifications with convolution-$t$ distributions that outperform
the conventional multivariate $t$-distribution (and the Gaussian).
This is true for both the small universe and the large universe, and
holds in-sample as well as out-of-sample. The out-of-sample comparisons
favors a sparse block correlation structure, but also shows that aside
from correlations within subindustries there are often non-trivial
correlations between stocks in different subindustries, but predominantly
for stocks in the same sector. 

The paper is organized as follows. In Section \ref{sec:NewPaBlockCorr},
we introduce the new factor correlation model, which features a novel
variation-free parametrization of factor loadings and a (sparse) block
idiosyncratic correlation matrix. In Section \ref{sec:Distributions},
we present details about the convolution-$t$ distributions and the
particular variants we use in the empirical analysis. In Section \ref{sec:Joint-Parameterization},
we develop the score-driven dynamic models for two estimation methods
(joint and decoupled), and provide analytical expressions for the
score and information matrix across a range of convolution-$t$ distributions
and structures for the idiosyncratic correlation matrices. We present
the empirical analysis in Section \ref{sec:EmpiricalAnalysis} and
conclude in Section \ref{sec:Conclusion}. All proofs are provided
in the Appendix.

\section{The Factor Correlation Model\label{sec:NewPaBlockCorr}}

Let $\{\mathcal{F}_{t}\}$ be a filtration and $R_{t}$ is an $n$-dimensional
return vector adapted to $\mathcal{F}_{t}$. We denote the conditional
mean and the conditional covariance matrix by $\mu_{t}=\mathbb{E}(R_{t}|\mathcal{F}_{t-1})$
and $\Sigma_{t}=\mathrm{var}(R_{t}|\mathcal{F}_{t-1})$, respectively.
We use the diagonal elements of the latter, $\sigma_{it}^{2}=\mathrm{var}(R_{it}|\mathcal{F}_{t-1})$,
$i=1,\ldots,n$, to define the diagonal matrix of conditional volatilities,
$\Lambda_{\sigma_{t}}\equiv\mathrm{diag}(\sigma_{1t},\ldots,\sigma_{nt})$,
such that the conditional correlation matrix of $R_{t}$ is given
by $C_{t}=\Lambda_{\sigma_{t}}^{-1}\Sigma_{t}\Lambda_{\sigma_{t}}^{-1}$.

We adopt the spirit of the Dynamic Conditional Correlation (DCC) model
by \citet{Engle2002} that models the conditional variances and conditional
correlations separately. Each of the $n$ univariate return series
is modeled with a univariate GARCH model (we use EGARCH models in
empirical analysis). From the resulting conditional moments, $\mu_{it}$
and $\sigma_{it}$, we define the vector of standardized returns,
$Z_{t}=\Lambda_{\sigma_{t}}^{-1}(R_{t}-\mu_{t})\sim(0,C_{t})$. Similarly,
we define the standardized factor variables, $F_{jt}=\sigma_{f_{j},t}^{-1}(R_{f_{j},t}-\mu_{f_{j}})$
for $j=1,\ldots,r$, such that $F_{t}\sim(0,C_{F,t})$.

We have little to add to the large existing literature on univariate
GARCH models. We will therefore take the univariate GARCH models as
given, and treat $Z_{t}$ and $F_{t}$ as the observed data. Our focus
is on the modeling of the conditional correlation matrix, $C_{t}=\operatorname{var}_{t-1}(Z_{t})$
using the observed factor variables $F_{t}$. 

A central component of our model is a factor structure,
\begin{equation}
Z_{it}=\beta_{it}^{\prime}F_{t}+\omega_{it}e_{it},\quad e_{it}\sim(0,1)\label{eq:factorZonF}
\end{equation}
for $i=1,\ldots,n$ and $t=1,\ldots,T$, where $e_{it}$ is an idiosyncratic
shock and $\omega_{it}^{2}$ is the proportion of variance attributed
to the idiosyncratic component. To simplify the notation, we suppress
the dependence on $t$ in most of Sections \ref{sec:NewPaBlockCorr}
and \ref{sec:Distributions}, such that (\ref{eq:factorZonF}) is
expressed as $Z_{i}=\beta_{i}^{\prime}F+\omega_{i}e_{i}$. We will
reintroduce subscript-$t$ again once the dynamic model is introduced.

From the standardized factor variables, $F\sim\left(0,C_{F}\right)$,
we define the uncorrelated factor variables, $U=C_{F}^{-1/2}F\sim(0,I_{r})$,
where $C_{F}^{1/2}$ denotes the symmetric square-root of $C_{F}$.\footnote{An attractive feature of $U=C_{F}^{-1/2}F$ is that it maximizes the
average correlation between $U_{j}$ and $F_{j}$, $j=1,\ldots,r$,
which helps interpret the results based on $U$. There are other ways
to define uncorrelated factor variables from $F$, such as those based
on Cholesky decompositions. All choices are equivalent in terms of
the implied factor loadings on $F$, $\beta_{i}$, $i=1,\ldots,n$. } This enables us to rewrite (\ref{eq:factorZonF}) as
\begin{equation}
Z_{i}=\rho_{i}^{\prime}U+\omega_{i}e_{i},\quad e_{i}\sim\left(0,1\right)\label{eq:factorZonU}
\end{equation}
where the elements of vector $\rho_{i}=C_{F}^{1/2}\beta_{i}$ are
simply the correlation coefficients, as we have $\rho_{i}=[{\rm corr}(Z_{i},U_{1}),\ldots,{\rm corr}(Z_{i},U_{r})]^{\prime}$.
It now follows that $\omega_{i}=\sqrt{1-\rho_{i}^{\prime}\rho_{i}}$.
In matrix form the model can be expressed as
\begin{equation}
Z=\boldsymbol{\rho}^{\prime}U+\Lambda_{\omega}e,\quad e\sim\left(0,C_{e}\right)\label{eq:MultiFactorModel}
\end{equation}
where $\boldsymbol{\rho}=[\rho_{1},\ldots,\rho_{n}]\in\mathbb{R}^{r\times n}$,
$\Lambda_{\omega}={\rm diag}\left(\omega_{1},\ldots,\omega_{n}\right)\in\mathbb{R}^{n\times n}$,
and $e$ is the vector of idiosyncratic shocks. We will not require
$e$ to be cross-sectionally uncorrelated. However, we will introduce
parsimonious and sparse structures on $C_{e}$. It follows that the
conditional correlation matrix for returns is given by,
\[
C=\boldsymbol{\rho}^{\prime}\boldsymbol{\rho}+\Lambda_{\omega}C_{e}\Lambda_{\omega},
\]
which is a generalization of \citet{HansenLundeVoev:2014} who focus
on a single factor ($r=1$).

The number of factors, $r$, is typically small relative to the number
of assets $n$, which makes it straightforward formulate a dynamic
model for the conditional correlation matrix of the observed factors
$F$, $C_{F}$. We adopt the score-driven multivariate GARCH model
by \citet{TongHansenArchakov:2024} for this purpose. 

The conditional model of $Z$ given $U$ is the central component
of the proposed model, and we will refer to this as the \emph{Core
Correlation Model}. The key parameters in the core correlation model
are the factor loading parameters, $\boldsymbol{\rho}\in\mathbb{R}^{r\times n}$,
and the correlation matrix for the idiosyncratic shocks, $C_{e}$.
(The scaling matrix, $\Lambda_{\omega}$, is a function of $\boldsymbol{\rho}$).
The main obstacle to a dynamic model of factor loadings is the requirement:
$\rho_{i}^{\prime}\rho_{i}<1$ for $i=1,2,\ldots,n$. We resolve this
by introducing a novel and mathematically elegant reparametrization
of $\rho_{i}$, $i=1,\ldots,n$, which is inspired by the generalized
Fisher transformation of correlation matrices, see \citet{ArchakovHansen:Correlation}.
An overview of the model structure is illustrated in Figure \ref{fig:ModelStructure},
which has many additional details related to distributions and estimation
method that will be explained later in this paper.

\subsection{A Novel Parametrization of Factor Loadings\label{sec:Parametrization-of-Correlation}}

The correlation structure in the factor correlation model, (\ref{eq:MultiFactorModel}),
is fully characterized by $\rho_{1},\ldots,\rho_{n}$ and $C_{e}$,
because $\omega_{i}=\sqrt{1-\rho_{i}^{\prime}\rho_{i}}$. This parametrization
must satisfy $\rho_{i}^{\prime}\rho_{i}<1$ for all $i$. An alternative,
variation-free parametrization of vector $\rho_{i}$ is the following.
\begin{thm}
\label{thm:CorrTau}Let the correlation structure of $Z\in\mathbb{R}^{n}$
be given by (\ref{eq:factorZonU}), where $\rho_{i}^{\prime}\rho_{i}<1$
for all $i=1,\ldots,n$. Then 
\begin{equation}
\tau_{i}=\tfrac{\operatorname{artanh}(\sqrt{\rho_{i}^{\prime}\rho_{i}})}{\sqrt{\rho_{i}^{\prime}\rho_{i}}}\times\rho_{i}\in\mathbb{R}^{r},\qquad i=1,\ldots,n,\label{eq:Tau_rho}
\end{equation}
is a variation-free parametrization of $\rho_{i}$ with domain $\tau=(\tau_{1}^{\prime},\ldots,\tau_{n}^{\prime})^{\prime}\in\mathbb{R}^{rn}$. 

The inverse mapping and its Jacobian matrix are given by 
\[
\rho_{i}=\tfrac{\tanh(\sqrt{\tau_{i}^{\prime}\tau_{i}})}{\sqrt{\tau_{i}^{\prime}\tau_{i}}}\times\tau_{i}\qquad i=1,\ldots,n.
\]
and $J(\tau_{i})\equiv\tfrac{\partial\rho_{i}}{\partial\tau_{i}^{\prime}}=\sqrt{\tfrac{\rho_{i}^{\prime}\rho_{i}}{\tau_{i}^{\prime}\tau_{i}}}P_{\tau_{i}}^{\bot}+(1-\rho_{i}^{\prime}\rho_{i})P_{\tau_{i}}$,
respectively, where $P_{\tau_{i}}=\tau_{i}(\tau_{i}^{\prime}\tau_{i})^{-1}\tau_{i}^{\prime}$
and $P_{\tau_{i}}^{\bot}=I_{r}-P_{\tau_{i}}$ are orthogonal projection
matrices. 
\end{thm}
The unrestricted $\tau$-parametrization in Theorem \ref{thm:CorrTau}
is inspired by the new parametrization of correlation matrices in
\citet{ArchakovHansen:Correlation}. We apply the matrix logarithm
to
\begin{equation}
C_{i}^{\star}={\rm corr}\left(\begin{array}{c}
Z_{i}\\
U
\end{array}\right)=\left[\begin{array}{cc}
1 & \rho_{i}^{\prime}\\
\rho_{i} & I_{r}
\end{array}\right],\label{eq:NewC}
\end{equation}
which is given by 
\begin{align*}
\log C_{i}^{\star} & =\frac{1}{2}\left[\begin{array}{cc}
\log\left(1-\rho_{i}^{\prime}\rho_{i}\right) & \frac{1}{\sqrt{\rho_{i}^{\prime}\rho_{i}}}\ensuremath{\log\left(\frac{1+\sqrt{\rho_{i}^{\prime}\rho_{i}}}{1-\sqrt{\rho_{i}^{\prime}\rho_{i}}}\right)}\rho_{i}^{\prime}\\
\frac{1}{\sqrt{\rho_{i}^{\prime}\rho_{i}}}\ensuremath{\log\left(\tfrac{1+\sqrt{\rho_{i}^{\prime}\rho_{i}}}{1-\sqrt{\rho_{i}^{\prime}\rho_{i}}}\right)}\rho_{i} & \frac{1}{\rho_{i}^{\prime}\rho_{i}}\log\left(1-\rho_{i}^{\prime}\rho_{i}\right)\left(\rho_{i}\rho_{i}^{\prime}\right)
\end{array}\right],
\end{align*}
as shown in Appendix \ref{sec:NewParaCorr}. Here we use the convention
$\log C^{\star}=0$ if $\rho_{i}=0\in\mathbb{R}^{r}$. An interesting
observation is that $\operatorname{arctanh}(\sqrt{\rho_{i}^{\prime}\rho_{i}})=\frac{1}{2}\log\left(\tfrac{1+\sqrt{\rho_{i}^{\prime}\rho_{i}}}{1-\sqrt{\rho_{i}^{\prime}\rho_{i}}}\right)$
is the Fisher transformation of $\sqrt{\rho_{i}^{\prime}\rho_{i}}$
(the root of sum of squares factor loadings).

The $\tau$-parametrization in Theorem \ref{thm:CorrTau} is useful
for several reasons. First, it makes it easy to impose sparsity and
other structure on $C_{i}^{\star}$ because $\rho_{i}$ is proportional
to $\tau_{i}$, such that $\rho_{i,j}=0\Leftrightarrow\tau_{i,j}=0$
and $\rho_{i,j}=\rho_{i,j^{\prime}}\Leftrightarrow\tau_{i,j}=\tau_{i,j^{\prime}}$.
Moreover, we also have $\rho_{i}=\rho_{i^{\prime}}\Leftrightarrow\tau_{i}=\tau_{i^{\prime}}$
which shows that two assets have identical factor loadings if and
only if the corresponding $\tau$-vectors are identical. This makes
it easy to impose group structures on the factor loadings. That the
Jacobian, $J$, is symmetric and easy to compute will also be convenient
in the dynamic score-driven model of $\rho_{i}$. 
\begin{rem*}[Notation of Factor Parametrizations]
We use $\boldsymbol{\rho}=(\rho_{1},\ldots,\rho_{n})$ to denote
the $r\times n$ matrix of factor loadings and let $\rho=\operatorname{vec}(\boldsymbol{\rho})=(\rho_{1}^{\prime},\ldots,\rho_{n}^{\prime})^{\prime}\in\mathbb{R}^{rn}$
be the vector with stacked factor loadings. Similarly, we use $\tau=(\tau_{1}^{\prime},\ldots,\tau_{n}^{\prime})^{\prime}\in\mathbb{R}^{rn}$.
\end{rem*}

\subsection{Idiosyncratic Correlation Matrix\label{sec:ParaBlock}}

The $n\times n$ correlation matrix for the idiosyncratic component,
$C_{e}={\rm corr}(e)$, has $d=n(n-1)/2$ correlations. A dynamic
model for $C_{e}$ will therefore need some structure to be imposed,
unless $n$ is small. In this context, a simple and convenient structure
is to impose block structures on $C_{e}$, that can be easily combined
with sparsity assumptions. Dynamic block correlation matrices were
introduced by \citet{EngleKelly:2012}, and the canonical representation
of block matrices by \citet{ArchakovHansen:CanonicalBlockMatrix}
made it possible to apply this structure to matrices with more than
$2\times2$ blocks, and greatly simplified estimation and guaranteeing
positive definiteness.

Below we adopt notations from \citet{TongHansenArchakov:2024}, but
should emphasize that our model structure is very different from that
in \citet{TongHansenArchakov:2024}, because they do not incorporate
observable factors. They impose the block structures directly on $C={\rm corr}(Z)$,
whereas we impose the block structure on in idiosyncratic correlation
matrix, $C_{e}={\rm corr}(e)$.

\subsubsection{Block Correlation Matrices}

A block correlation matrix is defined by partitioning the variables
into $K$ groups, where the correlation between any two variables
depends solely on their group assignments. Let $n_{k}$ denote the
number of variables in the $k$-th group for $k=1,\ldots,K$, such
that $n=\sum_{k=1}^{K}n_{k}$. Define $\boldsymbol{n}=\left(n_{1},n_{2},\ldots,n_{K}\right)^{\prime}$
as the vector of group sizes. We assume that the variables are sorted
such that the first $n_{1}$ variables belong to the first group,
the next $n_{2}$ variables belong to the second group, and so on.

The block structure on $C_{e}$ can be expressed with
\[
\ensuremath{C_{e}=\left[\begin{array}{cccc}
C_{[1,1]} & C_{[1,2]} & \cdots & C_{[1,K]}\\
C_{[2,1]} & C_{[2,2]}\\
\vdots &  & \ddots\\
C_{[K,1]} &  &  & C_{[K,K]}
\end{array}\right]},
\]
where $C_{[k,l]}$ is an $n_{k}\times n_{l}$ matrix given by
\[
C_{[k,l]}=\left[\begin{array}{ccc}
\varrho_{kl} & \cdots & \varrho_{kl}\\
\vdots & \ddots & \vdots\\
\varrho_{kl} & \cdots & \varrho_{kl}
\end{array}\right],\text{ for }k\neq l\qquad\text{and }C_{[k,k]}=\left[\begin{array}{cccc}
1 & \varrho_{kk} & \cdots & \varrho_{kk}\\
\varrho_{kk} & 1 & \ddots\\
\vdots & \ddots & \ddots\\
\varrho_{kk} &  &  & 1
\end{array}\right].
\]
Here we have omitted the subscript-$e$ on the submatrices to simplify
the expositions. Within each block, there is (at most) a single correlation
coefficient, such that the block structure reduces the number of unique
correlations from $d=n\left(n-1\right)/2$ to at most $K\left(K+1\right)/2$.\footnote{This is based on the general case that the number of assets in each
group is at least two. When there are $\tilde{K}\leq K$ clusters
with only one asset, this number become $K\left(K+1\right)/2-\tilde{K}$.} This number does not increase with $n$, if $K$ is held constant,
and this makes it possible to scale the model to high dimensions.

We refer to this as the the block correlations structure or the Full
Block Correlation (FBC) structure. In our empirical analysis we partition
assets by subindustries (8 digit GICS codes) which is used to define
the blocks in $C_{e}$. Sectors (2 digit GICS codes) will later be
used to impose sparsity on $C_{e}$.

\subsubsection{Sparse Block Correlation Matrices\label{subsec:Sparse-Block-Correlation}}

We consider two particular types of sparse block correlation matrices
that are based on a multi-level partitioning of the variables. 

A Sparse Block Correlation (SBC) is based on a second, coarser partitioning
that is induce sparsity on the correlation matrix. We use GICS sectors
to define the coarser partitioning in our empirical analysis. The
Diagonal Block Correlation (DBC) structure imposes additional sparsity
by imposing $C_{[k,l]}=0$ for $k\neq l$. In our empirical analysis,
this will implies that idiosyncratic shocks for stocks in different
subindustries are uncorrelated. 

Examples of the four types of correlation structures are shown in
Figure \ref{fig:BlockStructures}.

\begin{figure}[H]
\centering
\begin{centering}
\thickspace{}\includegraphics[totalheight=0.4\textwidth]{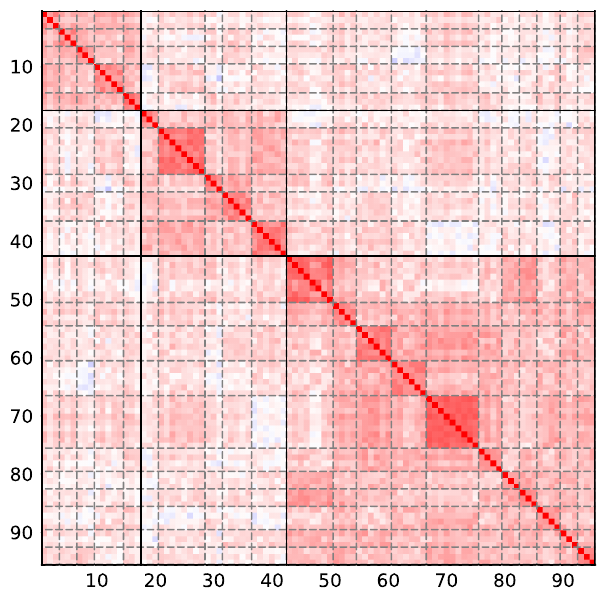}\includegraphics[totalheight=0.405\textwidth]{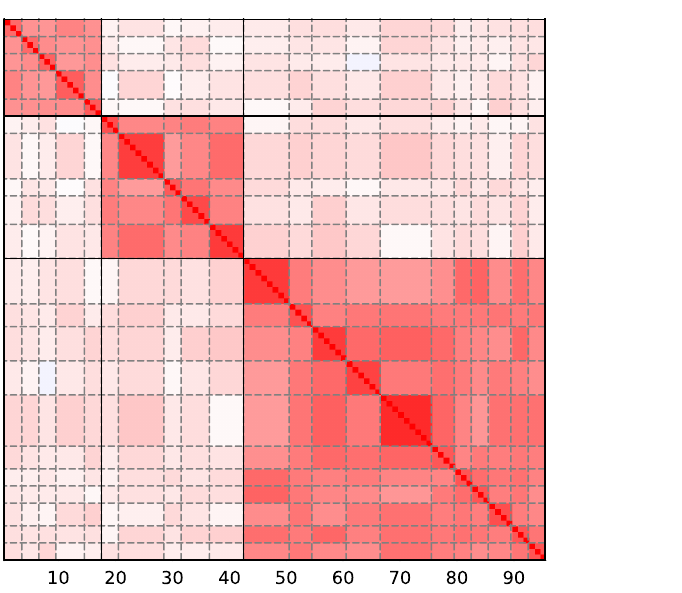}
\par\end{centering}
\begin{centering}
\includegraphics[totalheight=0.395\textwidth]{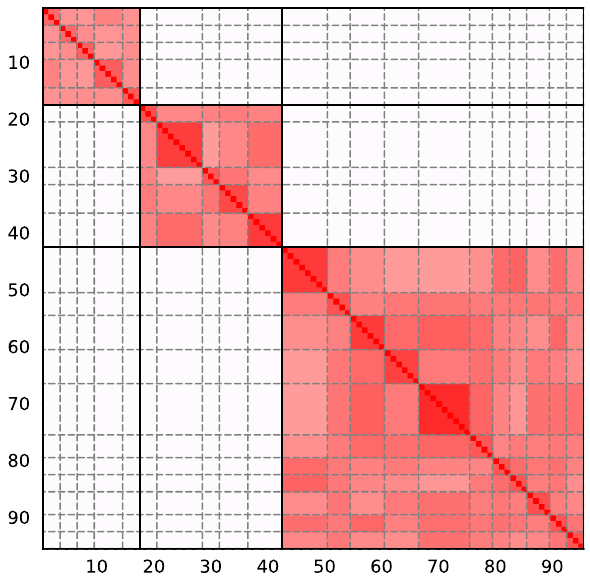}\includegraphics[totalheight=0.4\textwidth]{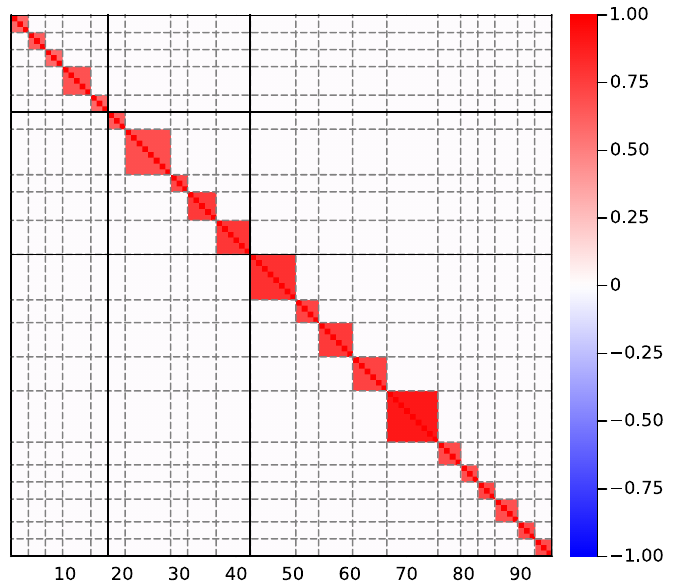}
\par\end{centering}
\caption{Examples of block correlation matrices. Upper-left: Unrestricted Correlation
matrix. Upper-right: Block Correlation matrix with blocks defined
by subindustries. Lower-left: Sparse Block Correlation matrix with
zero correlations between sectors. Lower-right: Diagonal Block Correlation
matrix.\label{fig:BlockStructures}}

\end{figure}

\subsection{Parametrizing the Idiosyncratic Correlation Matrix}

We parameterize the idiosyncratic correlation matrix, $C_{e}$, using
the generalized Fisher transformation by \citet{ArchakovHansen:Correlation},
\[
\gamma(C_{e})\equiv{\rm vecl}\left(\log C_{e}\right)\in\mathbb{R}^{n(n-1)/2}
\]
where $\log C_{e}$ is the matrix logarithm of $C_{e}$ and ${\rm vecl}(\cdot)$
vectorizes the elements in the lower triangle of $C_{e}$ (the elements
below the diagonal).\footnote{One can define the matrix logarithm of a nonsingular correlation matrix,
by $\log C_{e}=Q\log\Lambda Q^{\prime}$, where $C_{e}=Q\Lambda Q^{\prime}$
is the eigendecomposition of $C_{e}$.} The following example illustrates this parametrization for an $3\times3$
correlation matrix:{\small
\[
\gamma\equiv{\rm vecl}\left[\log\left(\begin{array}{ccc}
1.0 & \bullet & \bullet\\
0.7 & 1.0 & \bullet\\
0.4 & 0.6 & 1.0
\end{array}\right)\right]={\rm vecl}\left[\left(\begin{array}{ccc}
-.35 & \bullet & \bullet\\
.825 & -.53 & \bullet\\
.223 & .642 & -.24
\end{array}\right)\right]=\left(\begin{array}{c}
.825\\
.223\\
.642
\end{array}\right).
\]
}This parametrization defines a one-to-one mapping between $\mathbb{R}^{n(n-1)/2}$
and the set of positive definite correlation matrices, see \citet{ArchakovHansen:Correlation}.
The matrix logarithm preserves block structures, as illustrated with,{\small
\[
\ensuremath{\underbrace{\left[\begin{array}{cccc}
\cellcolor{black!5}1.0 & \cellcolor{black!5}0.8 & 0.4 & 0.4\\
\cellcolor{black!5}0.8 & \cellcolor{black!5}1.0 & 0.4 & 0.4\\
0.4 & 0.4 & \cellcolor{black!5}1.0 & \cellcolor{black!5}0.6\\
0.4 & 0.4 & \cellcolor{black!5}0.6 & \cellcolor{black!5}1.0
\end{array}\right]}_{=C_{e}}\quad\Longrightarrow\quad\underbrace{\left[\begin{array}{cccc}
\cellcolor{black!5}-.57 & \cellcolor{black!5}1.02 & .256 & .256\\
\cellcolor{black!5}1.02 & \cellcolor{black!5}-.57 & .256 & .256\\
.256 & .256 & \cellcolor{black!5}-.29 & \cellcolor{black!5}.628\\
.256 & .256 & \cellcolor{black!5}.628 & \cellcolor{black!5}-.29
\end{array}\right]}_{=\log C_{e}}}.
\]
}For a block matrix with $K$ blocks, $\gamma(C_{e})$ will have at
most $(K+1)K/2$ distinct elements, such that we can write $\gamma=B\eta$,
where $B$ is a known bit-matrix and $\eta$ is a subvector of $\gamma$.
In the example above we have,
\[
\gamma=B\eta,\quad B^{\prime}=\left(\begin{array}{cccccc}
1 & 0 & 0 & 0 & 0 & 0\\
0 & 1 & 1 & 1 & 1 & 0\\
0 & 0 & 0 & 0 & 0 & 1
\end{array}\right),\quad\eta=\left(\begin{array}{c}
1.02\\
.256\\
.628
\end{array}\right).
\]
Parametrizing the block correlation matrix, $C_{e}$, with $\eta$
does not impose additional superfluous restrictions, see \citet{TongHansen:2023}.
Thus, any non-singular block correlation matrix corresponds to a unique
$\eta$ vector, and any dynamic block correlation model can be expressed
as a dynamic model for $\eta$.

For the two sparse correlation structures, SBC and DBC, $\log C_{e}$
with have the same sparse structure, with zeroes outside the diagonal
blocks. The dimension of $\eta$ can therefore be reduced substantially.
For the DBC structure with $K$ diagonal blocks, we can use $\eta=(\eta_{1},\ldots,\eta_{K})^{\prime}$,
where
\begin{equation}
\eta_{k}=\tfrac{1}{n_{k}}\log\left(1+n_{k}\tfrac{\varrho_{kk}}{1-\varrho_{kk}}\right),\quad k=1,\ldots,K,\label{eq:gamma_k}
\end{equation}
see \citet[proposition 2]{ArchakovHansen:Correlation}. This follows
from $\eta_{k}=\log C_{[k,k]}$, where each diagonal block, $C_{[k,k]}$,
is an equicorrelation matrix. For later used we observe that inverse
transformation and its Jacobian are given by
\[
\varrho_{kk}=\tfrac{\exp\left(n_{k}\eta_{k}\right)-1}{\exp\left(n_{k}\eta_{k}\right)+n_{k}-1},\quad\text{and}\quad J_{k}=\tfrac{\partial\varrho_{kk}}{\partial\eta_{k}}=\tfrac{1}{\left(1-\varrho_{kk}\right)\left(1+\left(n_{k}-1\right)\varrho_{kk}\right)}.
\]
So, the analysis is greatly simplified with the DBC structure for
$C_{e}$, even in very high-dimensional settings.

For SBC correlation structure, we can apply the results for block
correlation matrices to each of the big diagonal blocks separately,
and combining the $\eta$s for each block to define $\eta$.

\subsubsection{Canonical Form of Idiosyncratic Correlation Matrix\label{subsec:CanBlock}}

To build a score-driven model for block correlation matrix, the conventional
way is to use
\begin{equation}
\frac{\partial\ell}{\partial\eta^{\prime}}=\frac{\partial\ell}{\partial\gamma^{\prime}}\frac{\partial\gamma}{\partial\eta^{\prime}}=\frac{\partial\ell}{\partial\gamma^{\prime}}B\label{eq:ParEqu}
\end{equation}
where $\gamma={\rm vecl}\left(\log C_{e}\right)\in\mathbb{R}^{n(n-1)/2}$
and $\ell$ is the log-likelihood function. Therefore, even if $C_{e}$
has a block structure with unique value $\eta$ with at most $K\left(K+1\right)/2$
elements, the conventional way through (\ref{eq:ParEqu}) also involves
the computation of $n\times n$ matrix, which doesn't take the advantages
of block structure of $C_{e}$. In fact, expect for diagonal block
correlation matrix which can be modeled through unrestricted $\eta_{k}$
for each $k=1,2\ldots,K$, the modeling of a general block correlation
matrix $C_{e}$ including the SBC case, both requires the following
canonical representation of block correlation matrix, such that the
related computation only involves the $K$-dimensional matrix. The
canonical representation of block correlation matrix resembles the
eigendecomposition of matrices, see \citet{ArchakovHansen:CanonicalBlockMatrix}.
For a block correlation matrix with block-sizes, $(n_{1},\ldots,n_{K})$,
we have
\begin{equation}
C_{e}=QDQ^{\prime},\quad D=\left[\begin{array}{cccc}
A & 0 & \cdots & 0\\
0 & \delta_{1}I_{n_{1}-1} & \ddots & \vdots\\
\vdots & \ddots & \ddots & 0\\
0 & \cdots & 0 & \delta_{K}I_{n_{K}-1}
\end{array}\right],\quad\delta_{k}=\frac{n_{k}-A_{kk}}{n_{k}-1}\label{eq:CanForm}
\end{equation}

\noindent where $A\in\mathbb{R}^{K\times K}$ with $A_{kk}=1+\left(n_{k}-1\right)\rho_{kk}$,
and $A_{kl}=\rho_{kl}\sqrt{n_{k}n_{l}}$ for $k\neq l$. Matrix $Q$
is a cluster-specific orthonormal matrix, $Q^{\prime}Q=QQ^{\prime}=I_{n}$,
which is solely determined by the block sizes, $(n_{1},\ldots,n_{K})$,
such that it does not depend on $C_{e}$. Computing the powers of
$C_{e}$, including the matrix inverse, logarithm and the exponential,
is greatly simplified as they only involve the calculations related
to $K\times K$ matrix $A$. From \citet[corollary 2]{ArchakovHansen:CanonicalBlockMatrix},
the unique values in $\gamma\left(C_{e}\right)$, i.e. the elements
in $\eta$, can be expressed as
\begin{equation}
\eta=L_{K}(\Lambda_{n}^{-1}\otimes\Lambda_{n}^{-1}){\rm vec}(W),\label{eq:eta_can}
\end{equation}
where $W=\log A-\log\Lambda_{\delta}$, where $\Lambda_{\delta}=\operatorname{diag}(\delta_{1},\ldots,\delta_{K})$
and $\Lambda_{n}=\operatorname{diag}(\sqrt{n_{1}},\ldots,\sqrt{n_{K}})$
are diagonal matrices. $L_{K}$ is the elimination matrix, that solves
$\mathrm{vech}(A)=L_{k}\mathrm{vec}(A)$, and $\otimes$ is the Kronecker
product. The canonical representation greatly facilitates the evaluation
of likelihood function and the computation of score. \citet{TongHansenArchakov:2024}
build a score driven model for a general block correlation matrix
by utilizing the canonical form (\ref{eq:CanForm}) and (\ref{eq:eta_can}).

\section{Distributions\label{sec:Distributions}}

The class of distributions we consider in our empirical analysis is
detailed in this section. We adopt the family of convolution-$t$
distributions of \citet{HansenTong:2024}. These can accommodate nonlinear
dependencies and heterogeneous marginal distributions. This class
of distributions nests the multivariate $t$-distributions (including
Gaussian distributions) as special cases. A convolution-$t$ distribution
has a relatively simple log-likelihood function, and we adopt this
distribution for both the factor variables, $F\sim\left(0,C_{F}\right)$,
and the idiosyncratic shocks, $e\sim\left(0,C_{e}\right)$. Under
the assumption that $e$ and $F$ are independent, the conditional
distribution of $Z$ given $F$ is also a convolution-$t$ distribution. 

\subsection{The Convolution-$t$ Distribution}

Let $X$ be a random vector with mean, $\mu=\mathbb{E}[X]$, and covariance
matrix, $\Sigma=\operatorname{var}(X)$. We use $X$ to represent
either $F$, $Z|U$ or $e$, and write
\[
X=\mu+\Xi V\sim\mathrm{CT}_{\boldsymbol{m},\boldsymbol{\nu}}^{{\rm std}}(\mu,\Xi),\quad{\rm where}\ \Xi\Xi^{\prime}=\Sigma,
\]
which follows the notation in \citet{HansenTong:2024}. A convolution-$t$
distribution is (aside from location and scale) a rotation of a random
vector, $V=(V_{1}^{\prime},\ldots,V_{G}^{\prime})^{\prime}\in\mathbb{R}^{n}$,
which consists of $G$ mutually independent multivariate $t$-distributions,
where $V_{g}\sim t_{\nu_{g}}^{\mathrm{std}}(0,I_{m_{g}})$ has dimension
$m_{g}$ and degrees of freedom, $\nu_{g}>2$, for $g=1,\ldots,G$
and with $n=\sum_{g=1}^{G}m_{g}$. Here $\boldsymbol{\nu}=(\nu_{1},\ldots,\nu_{G})^{\prime}$
is the vector with degrees of freedom and $\boldsymbol{m}=(m_{1},\ldots,m_{G})^{\prime}$
is the vector with the dimensions for the $G$ multivariate $t$-distributions.
The corresponding log-likelihood function is surprisingly simple.
Set $V=\Xi^{-1}\left(X-\mu\right)$, then
\begin{align}
\ell(X) & =-\log|\Xi|+\sum_{g=1}^{G}c_{g}-\tfrac{\nu_{g}+m_{g}}{2}\log\left(1+\tfrac{1}{\nu_{g}-2}V_{g}^{\prime}V_{g}\right),\label{eq:LogLGenC_StrucT}
\end{align}
where $c_{g}=c(\nu_{g},m_{g})=\log\left(\Gamma\left(\frac{\nu_{g}+m_{g}}{2}\right)/\Gamma\left(\frac{\nu_{g}}{2}\right)\right)-\frac{m_{g}}{2}\log[(\nu_{g}-2)\pi]$,
$g=1,\ldots,G$, are the normalizing constants in the multivariate
$t$-distribution. Note that we previously used a partitioning of
the variables, $n=\left(n_{1},\ldots,n_{K}\right)^{\prime}$, to form
a block correlation structure. The convolution-$t$ distribution involves
a second partitioning that defines the $G$ independent multivariate
$t$-distributions. This is a cluster structure for nonlinear dependences
in the underlying random innovations. The two cluster structures can
be different or identical. 

Next, we highlight five distributional properties of this model. First,
each element of $V_{g}\in\mathbb{R}^{m_{g}}$ follows the same marginal
$t$-distribution with $\nu_{g}$ degrees of freedom. This is not
necessarily true for the elements of $X$, as they are typically unique
convolutions of the $G$ underlying $t$-distributions. Second, the
marginal distributions of $X$ may be time-varying, as they depend
on $\Xi$, which can change over time in the model. Third, the elements
of $X$ have intricate dependencies, arising from the common $t$-distributions
they share. These dependencies induce tail correlations and can include
cluster-specific tail dependencies. Fourth, increasing the number
of $G$-clusters does not always improve the empirical fit. While
a larger $G$ increases the number of degrees-of-freedom parameters,
it also divides $V$ into a larger number of independent subvectors,
which removes the intrinsic dependence between elements of $V$ that
were previously part of the same multivariate $t$-distribution. Fifth,
this model nests the conventional multivariate $t$-distribution as
a special case when $G=1$, facilitating straightforward comparisons
with a natural benchmark model, such as the multivariate Gaussian
distribution.

\subsection{Three Special Convolution-$t$ Distributions}

The convolution-$t$ distributions define a broad class of distributions,
as $V$ can be partitioned in many ways. We will use three particular
types of convolution-$t$ distributions, one being the standard multivariate
$t$-distribution.

\subsubsection{Multivariate-$t$ Distribution (MT)}

The convolution-$t$ distributions nests the multivariate $t$-distribution
as a special case. The $n$-dimensional (standardized) multivariate
$t$-distribution is denoted, $X\sim t_{\nu}^{{\rm std}}(\mu,\Sigma)$,
where $\nu>2$ is the degrees of freedom. The corresponding log-likelihood
is given by
\begin{equation}
\ell(X)=c_{\nu,n}-\tfrac{1}{2}\log|\Sigma|-\tfrac{\nu+n}{2}\log\left(1+\tfrac{1}{\nu-2}\left(X-\mu\right){}^{\prime}\Sigma^{-1}\left(X-\mu\right)\right),\quad\nu>2.\label{eq:LogLstuT}
\end{equation}
As $\nu\rightarrow\infty$, the multivariate $t$-distribution approaches
the multivariate normal distribution, $N(\mu,\Sigma)$. The main advantage
of using the standardized $t$-distribution is that $\mathrm{var}(X)=\Sigma$.
If $X$ is standardized and $\Sigma=C$ has a block structure with
$K$ clusters, we can use the canonical representation in Section
\ref{subsec:CanBlock} of to obtain the following simplified expression,
\[
\ell(X)=c_{\nu,n}-\tfrac{1}{2}\log|A|-\tfrac{1}{2}\sum_{k=1}^{K}\left(n_{k}-1\right)\log\delta_{k}-\tfrac{\nu+n}{2}\log\left[1+\tfrac{1}{\nu-2}\left(Y_{0}^{\prime}A^{-1}Y_{0}+\sum_{k=1}^{K}\tfrac{1}{\delta_{k}}Y_{k}^{\prime}Y_{k}\right)\right].
\]
where $Y=Q^{\prime}\left(X-\mu\right)$. The multivariate $t$-distribution
has two potential drawbacks. First, all elements of a multivariate
$t$-distribution are dependent, because they share a common random
mixing variable. Second, all elements of $V$ are identically distributed,
because they are all $t$-distributed with the same degrees of freedom.
Both implications may be too restrictive in many applications, especially
if the dimension, $n$, is large.

\subsubsection{Cluster-$t$ Distribution (CT)}

The second special type of convolution-$t$ distribution is called
the Cluster-$t$ (CT) distribution. It has a cluster structure on
$V$, represented by $\boldsymbol{m}$. In the absence of a block
structure on $\Xi$, the log-likelihood function is computed by (\ref{eq:LogLGenC_StrucT}).
If $X$ is standardized, $\Sigma=C$ has a block structure with $\boldsymbol{n}=\boldsymbol{m}$
and $G=K$, and we set $\Xi=C^{1/2}$, then we have,
\[
V_{k}^{\prime}V_{k}=Y_{0}^{\prime}A^{-\tfrac{1}{2}}e_{k}e_{k}^{\prime}A^{-\tfrac{1}{2}}Y_{0}+\delta_{k}^{-1}Y_{k}^{\prime}Y_{k},\quad k=1,\ldots,K,
\]
where $Y=Q^{\prime}\left(X-\mu\right)=\left(Y_{0}^{\prime},Y_{1}^{\prime}\ldots,Y_{K}^{\prime}\right)^{\prime}$
and the log-likelihood function simplifies to
\begin{align*}
\ell(X) & =-\tfrac{1}{2}\log|A|+\sum_{k=1}^{K}c_{k}-\tfrac{1}{2}\left(n_{k}-1\right)\log\delta_{k}-\tfrac{\nu_{k}+n_{k}}{2}\log\left(1+\tfrac{1}{\nu_{k}-2}V_{k}^{\prime}V_{k}\right),
\end{align*}
where $c_{k}=c(\nu_{k},n_{k})$. The block structure simplifies implementation
of the score-driven model for this specification, and makes it possible
to apply in high dimensions.

\subsubsection{Convolution of Heterogeneous $t$-distributions (HT)}

The third special type of convolution-$t$ distributions has $G=n$.
So, the elements of $V$ are made up of $n$ independent univariate
$t$-distributions with degrees of freedom, $\nu_{i}$, $i=1,\cdots,n$.
This distribution can accommodate a high degree of heterogeneity in
the marginal properties of $X_{i}$, $i=1,\ldots,n$, which are different
convolutions of heterogeneous independent $t$-distributions. For
this reason, we refer to these as HT distributions, where H is short
for heterogeneous. The number of degrees of freedom increases from
$G$ to $n$, but the additional parameters do not guarantee a better
in-sample log-likelihood, because dependences between elements of
$V$ is eliminated. Without structure on $\Sigma$, the log-likelihood
function is simply computed by (\ref{eq:LogLGenC_StrucT}). As before,
if $X$ is standardized, $\Sigma=C$ has a block structure with group
assignments $\bm{n}$, and we set $\Xi=C^{1/2}$, then the log-likelihood
function simplifies to
\[
\ell(X)=c-\tfrac{1}{2}\log|A|-\tfrac{1}{2}\sum_{k=1}^{K}\left(n_{k}-1\right)\log\lambda_{k}-\sum_{k=1}^{K}\sum_{j=1}^{n_{k}}\tfrac{\nu_{k,j}+1}{2}\log\left(1+\tfrac{1}{\nu_{k,j}-2}V_{k,j}^{2}\right),
\]
where $c=\sum_{i=1}^{n}c(\nu_{i},1)$ and $V_{k,j}$ is the $j$-th
element of the vector $V_{k}$.

\section{Model Architecture and Components\label{sec:Joint-Parameterization}}

Individual returns and factor returns are model with univariate GARCH
models that define the standardized variables, $Z_{t}$ and $F_{t}$.
The dynamic model of $(Z_{t},F_{t})$ has several components. The
first is the dynamic model of $C_{F}=\mathrm{corr}(F_{t})$, which
is a score-driven model of $\gamma(C_{f})$, and this model delivers
the orthogonalized factor variables, $U_{t}$. The most innovative
component is the way we model dynamic factor loadings with the $\tau$-parametrization.
This is part of the core correlation model, which also include the
dynamic model of the idiosyncratic correlation matrix, $C_{e}$. The
two components of the core correlation model can be estimated jointly
or separately, where we refer to the latter as decoupled estimation.
An overview of the structure of the factor correlation model is illustrated
in Figure \ref{fig:ModelStructure}.

\subsection{Univariate Volatility Models for Returns Series}

The factor correlation model requires standardized returns, $Z_{i,t}$,
$i=1,\ldots,n$ and standardized factor variables, $F_{j,t}$, $j=1,\ldots,r$.
Each of these are obtained from suitable univariate models for the
conditional mean and variance of their corresponding return series. 

There is a wide range of choices for the univariate volatility models
that serve this purpose, see e.g. \citet{HansenLundeBeatGarch}. We
do not contribute to this modeling aspect. In our empirical analysis
we simply use univariate EGARCH model, see \citet{Nelson91}, to standardized
each returns series, from which we obtain $Z_{t}\in\mathbb{R}^{n}$
and $F_{t}\in\mathbb{R}^{r}$.

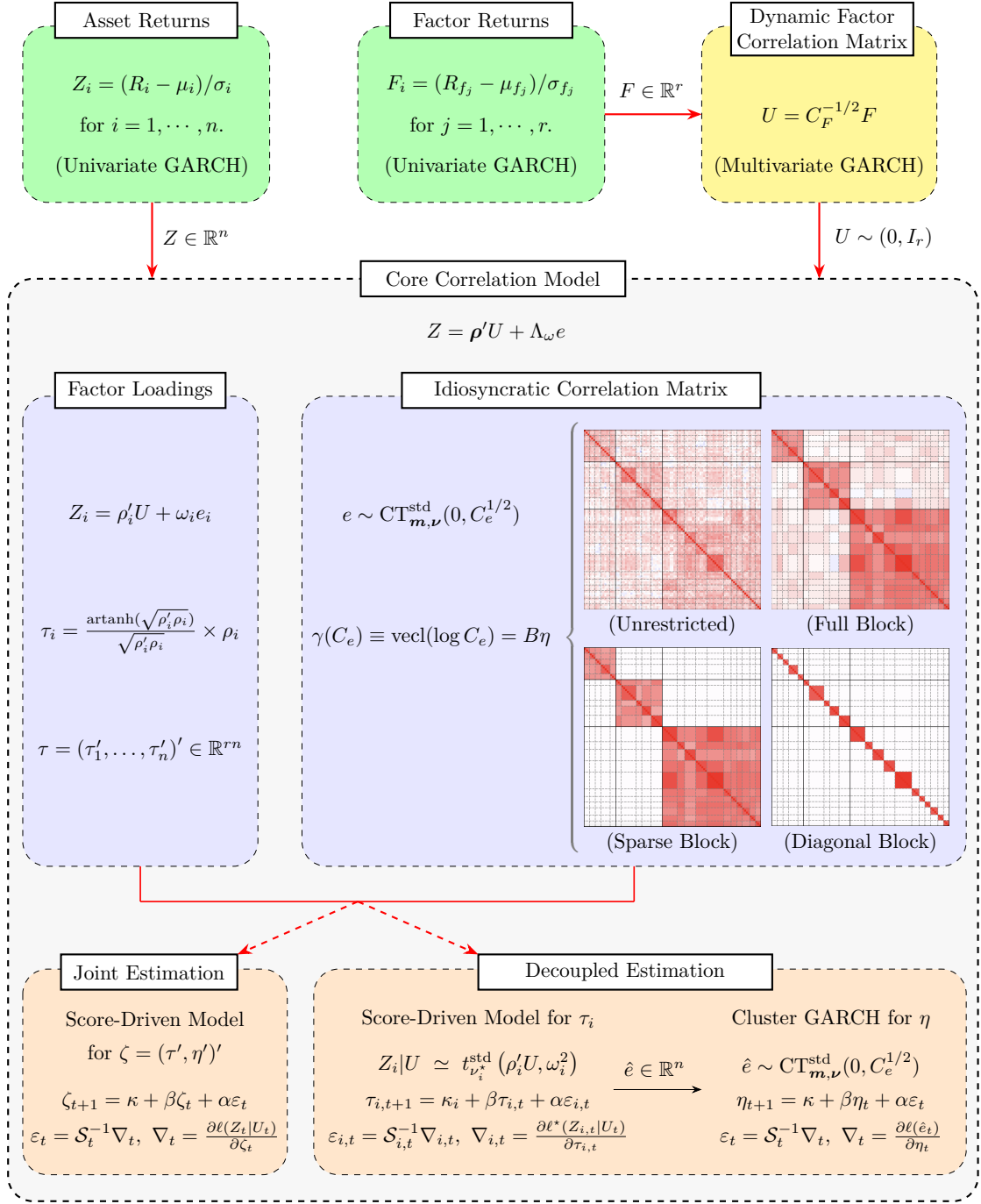
\begin{figure}[p]
\centering
\begin{centering}
\centering
\resizebox{0.9\textwidth}{!}{%
\begin{circuitikz}
\tikzstyle{every node}=[font=\footnotesize]
\draw [ rounded corners = 12,dashed, fill=green!30] (4.25,19-0.7) rectangle  node {} (8.25+0.2,16-0.7);
\draw [ fill=white ,line width=0.9pt] (4.8,19.4-0.7) rectangle  node {{Asset Returns}} (7.9,18.8-0.7);
\node  at (6.25+0.2,18-0.7) {$Z_i = (R_i - \mu_i)/{\sigma_i}$};
\node  at (6.25+0.2,17.3-0.7) {for $i=1,\cdots, n.$}; 
\node  at (6.25+0.2,16.6-0.7) {(Univariate GARCH)}; 
\draw [->, >=Stealth, color = red, line width=0.9pt] (6.25+0.2,16-0.7) -- (6.25+0.2,14);
\node  at (7+0.2,15-0.3) {$Z\in\mathbb{R}^{n}$};

\draw [ rounded corners = 12,dashed, fill=green!30] (9.25+0.7,19-0.7) rectangle  node {} (13.25+0.9,16-0.7);
\draw [ fill=white ,line width=0.9pt] (9.8+0.7,19.4-0.7) rectangle  node {{Factor Returns}} (12.7+0.9,18.8-0.7);
\node  at (11.25+0.8,18-0.7) {$F_i = (R_{f_j} - \mu_{f_j})/\sigma_{f_j}$};
\node  at (11.25+0.8,17.3-0.7) {for $j=1,\cdots, r.$}; 
\node  at (11.25+0.8,16.6-0.7) {(Univariate GARCH)};
\draw  [->, >=Stealth, color = red, line width=0.9pt] (13.25+0.9,17.5-0.7) -- (15.8,17.5-0.7); 
\node  at (14.05+0.9,17.9-0.7) {$F\in\mathbb{R}^{r}$};

\draw [rounded corners = 12,dashed, fill=yellow!50] (15.8,19-0.7) rectangle  node {} (19.8,16-0.7);
\node  at (17.8,17.5-0.7) {$U = C_F^{-1/2}F$};
\node  at (17.8,16.6-0.7) {(Multivariate GARCH)};
\draw [->, >=Stealth, color = red, line width=0.9pt] (17.8,16-0.7) -- (17.8,14);
\node  at (18.9,15-0.3) {$U \sim (0,I_r)$};
\draw [ fill=white ,line width=0.9pt] (16.3,19.4-0.7) rectangle  node {} (19.4,18.5-0.7);
\node  at (17.85,19.15-0.7) {Dynamic Factor};
\node  at (17.85,18.75-0.7) {Correlation Matrix};

\draw [rounded corners = 12, fill=black!3, dashed, line width=0.9pt] (4,14) rectangle  (20.5,-1.7);
\draw [fill=white ,line width=0.9pt] (10,14.3) rectangle  node {Core Correlation Model} (14.5,13.7);
\node  at (12.3,13.1) {$Z=\boldsymbol{\rho}^{\prime}U+\Lambda_{\omega}e$};

\draw [rounded corners = 12,dashed, fill=blue!10] (4.25,12) rectangle  (8.25,4);
\draw [ fill=white ,line width=0.9pt] (4.8,12.4) rectangle  node {Factor Loadings} (7.7,11.8);
\node  at (6.25,10) {$Z_i = \rho_i^{\prime}U+\omega_i e_i$};
\node  at (6.25,8) {$\tau_i=\tfrac{\operatorname{artanh}(\sqrt{\rho_{i}^{\prime}\rho_{i}})}{\sqrt{\rho_{i}^{\prime}\rho_{i}}}\times\rho_{i}$};
\node  at (6.25,6) {$\tau=\left(\tau_1^{\prime}, \ldots, \tau_n^{\prime}\right)^{\prime} \in \mathbb{R}^{r n}$};

\draw [rounded corners = 12,dashed, fill=blue!10] (9,12) rectangle  (18.8+1.5,4);
\draw [ fill=white ,line width=0.9pt] (10.7,12.4) rectangle  node {Idiosyncratic Correlation Matrix} (16.8,11.8);
\node (tikzmaker) at (15.3,10-0.1) {\includegraphics[width=3cm]{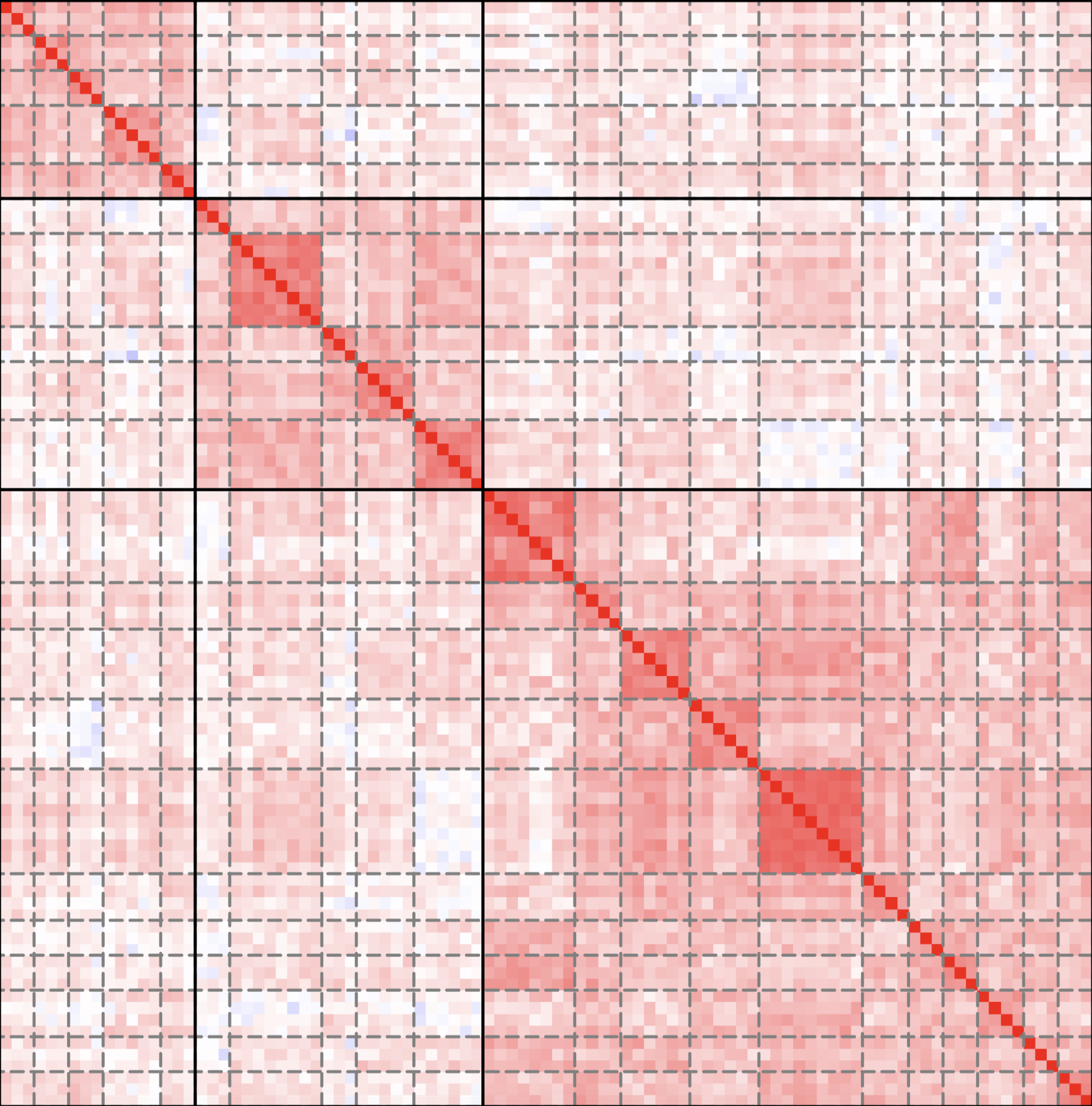}};
\node (tikzmaker) at (18.5,10-0.1) {\includegraphics[width=3cm]{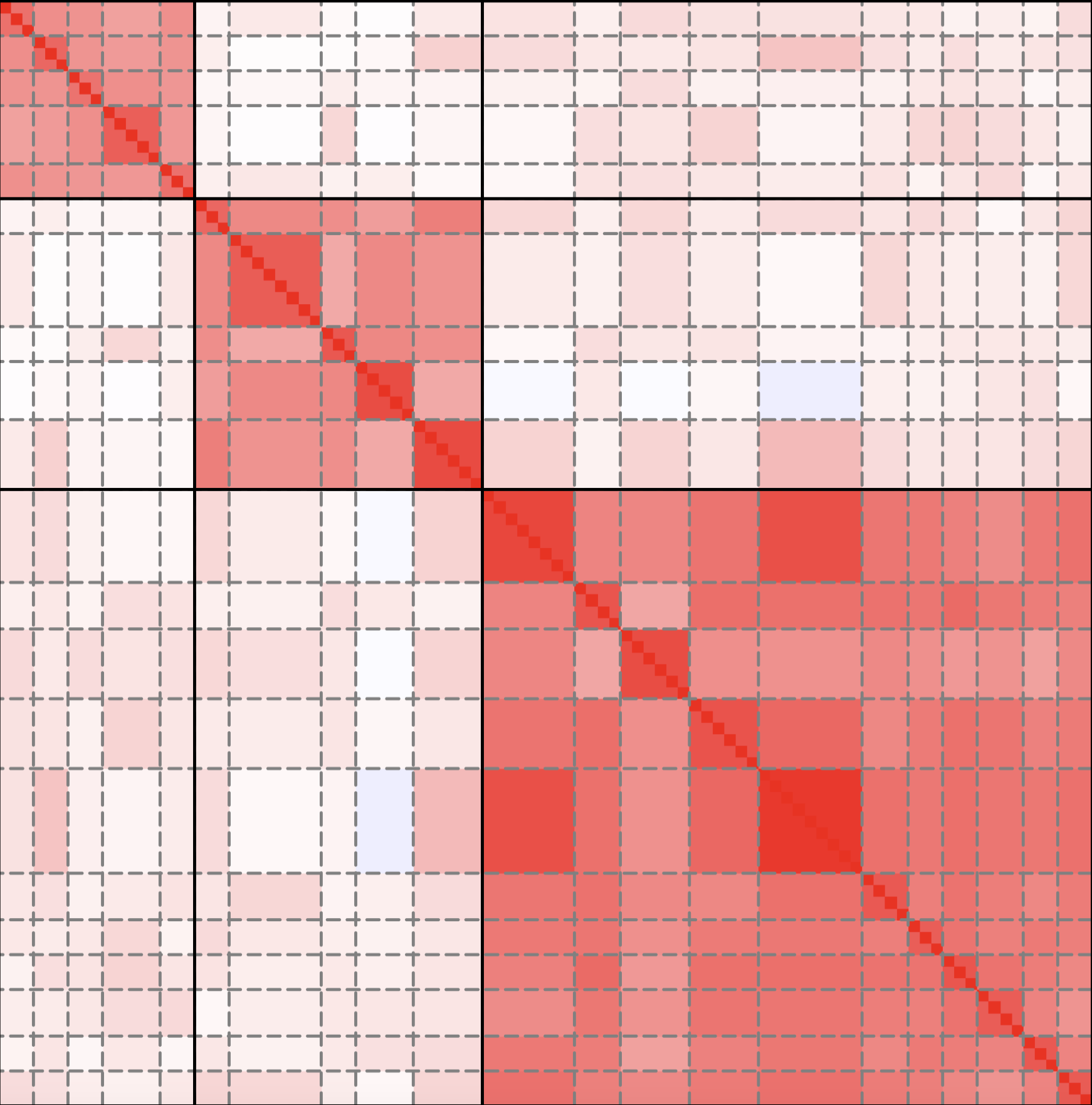}};
\node  at (15.3,8.2-0.1) {(Unrestricted)};
\node  at (18.5,8.2-0.1) { (Full Block)};
\node (tikzmaker) at (15.3,6.2) {\includegraphics[width=3cm]{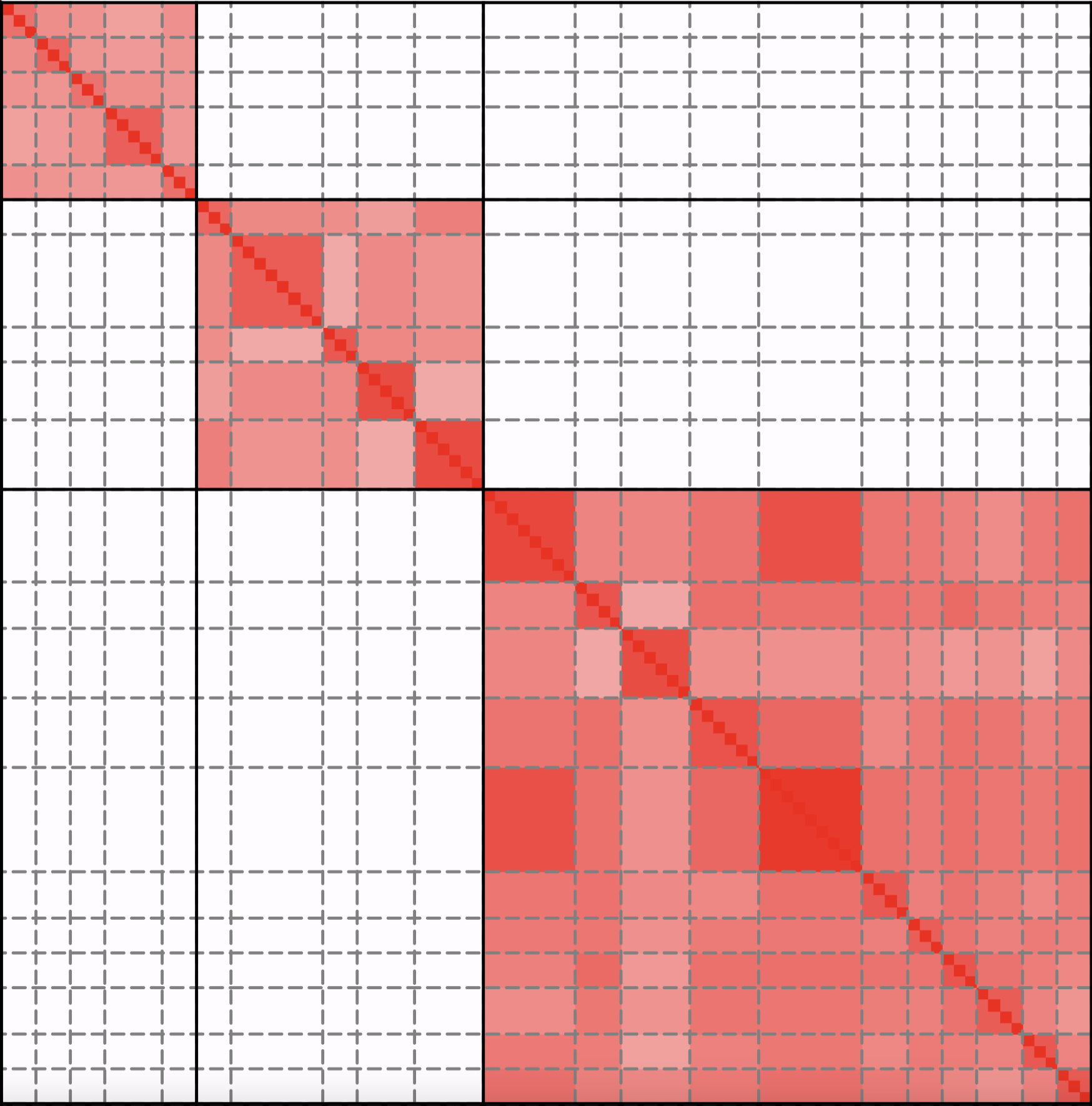}};
\node (tikzmaker) at (18.5,6.2) {\includegraphics[width=3cm]{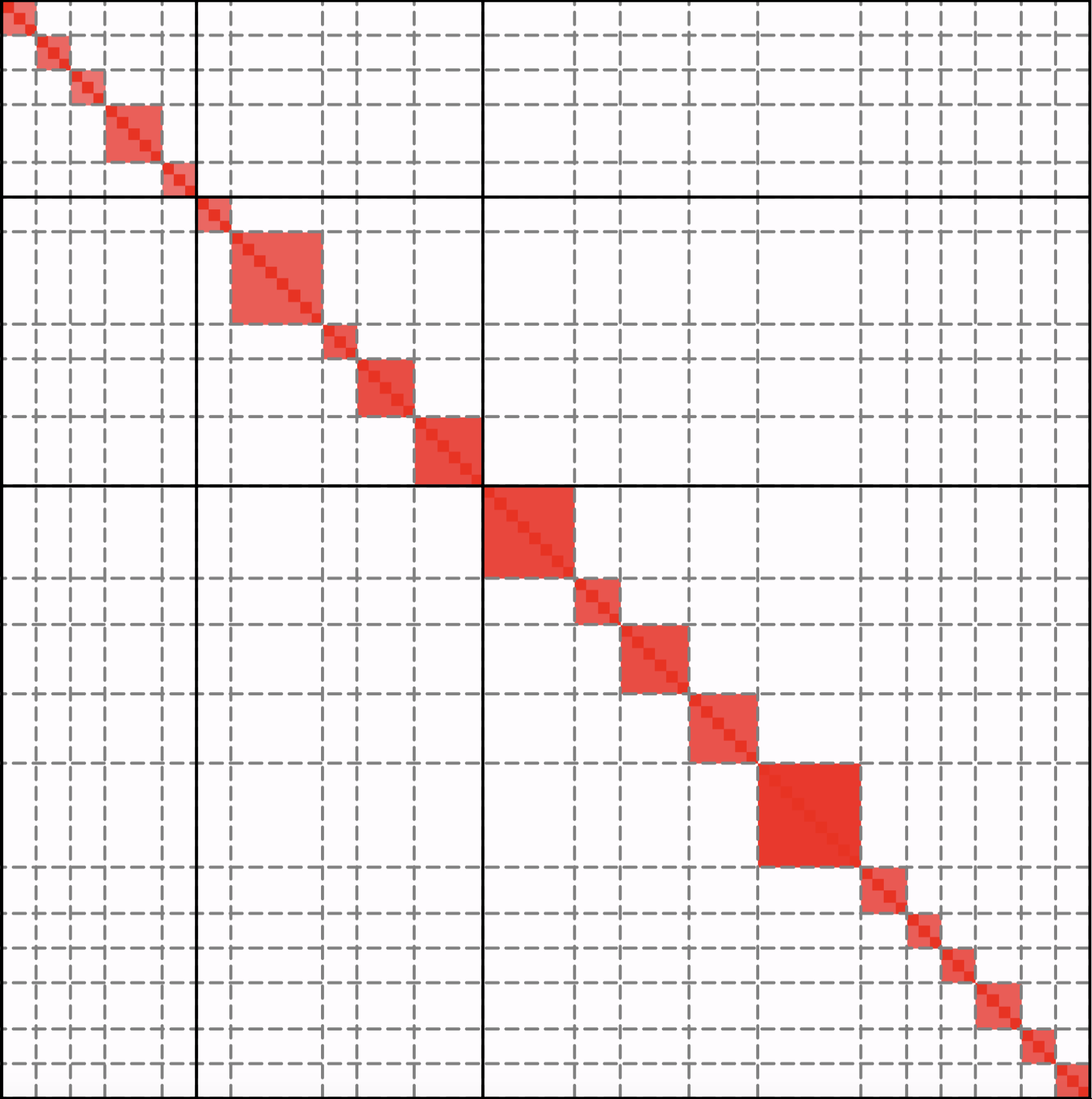}};
\node  at (15.3,4.4) {(Sparse Block)};
\node  at (18.5,4.4) {(Diagonal Block)};
\draw (13.6,7.9) node[rotate = 270] {{\color{gray}$\underbrace{\hspace{7.3cm}}$}};
\node  at (11.2,7.9) {$\gamma(C_e) \equiv \text{vecl}(\log C_e) =B\eta$};
\node  at (11.2,10) {$e \sim \mathrm{CT}_{\boldsymbol{m},\boldsymbol{\nu}}^{{\rm std}}(0,C_{e}^{1/2})$};

\draw [rounded corners = 12,dashed, fill=orange!20] (4.25,2.25) rectangle  (8.75,-1.25);
\node  at (6.5,1.4) {Score-Driven Model};
\node  at (6.5,0.8) {for $\zeta     =(\tau^{\prime},\eta^{\prime})^{\prime}$};
\node  at (6.5,0) {$\zeta_{t+1} =\kappa+\beta\zeta_{t}+\alpha\varepsilon_{t}$};
\node  at (6.5,-.6) {$\varepsilon_{t}=\mathcal{S}_{t}^{{-1}}\nabla_{{t}}, \ \nabla_{{t}}=\tfrac{\partial\ell(Z_{t}|U_{t})}{\partial\zeta_{t}} $};
\draw [ fill=white ,line width=0.9pt] (4.9,2.5) rectangle  node {Joint Estimation} (7.9,1.9);

\draw [rounded corners = 12,dashed, fill=orange!20] (9.2,2.25) rectangle  (20.3,-1.25);
\node  at (12,1.4) {Score-Driven Model for $\tau_i$};
\node  at (12,0.6) {$Z_{i}|U \ \simeq \ t_{\nu_{i}^{\star}}^{{\rm std}}\left(\rho_i^{\prime}U,\omega^2_i \right)$};
\node  at (12,0) {$\tau_{i,t+1} =\kappa_i+\beta\tau_{i,t}+\alpha\varepsilon_{i,t}$};
\node  at (12,-0.6) {$\varepsilon_{i,t}=\mathcal{S}_{i,t}^{-1}\nabla_{{i,t}}, \ \nabla_{{i,t}}=\tfrac{\partial\ell^{\star}(Z_{i,t}|U_{t})}{\partial\tau_{i,t}} $};
\draw [->, >=Stealth, color = black] (14.3,0.2) -- (15.8,0.2);
\node  at (15,0.55) {$\hat{e} \in \mathbb{R}^n$};

\node  at (18,1.4) {Cluster GARCH for $\eta$};
\node  at (18,0.6) {$\hat{e}\sim\mathrm{CT}_{\boldsymbol{m},\boldsymbol{\nu}}^{{\rm std}}(0,C_{e}^{1/2})$};
\node  at (18,0) {$\eta_{t+1}   =\kappa+\beta\eta_{t}+\alpha\varepsilon_{t}$};
\node  at (18,-0.6) {$\varepsilon_{t}=\mathcal{S}_{t}^{-1}\nabla_{{t}}, \ \nabla_{{t}}=\tfrac{\partial\ell(\hat{e}_{t})}{\partial\eta_{t}} $};
\draw [ fill=white ,line width=0.9pt] (12,2.5) rectangle  node {Decoupled Estimation} (17,1.9);

\draw  [color = red, line width=0.7pt] (6.25,4) -- (6.25,3.4);
\draw  [color = red, line width=0.7pt] (14.65,4) -- (14.65,3.4);
\draw  [color = red, line width=0.7pt] (6.25,3.4) -- (14.65,3.4);
\draw  [->, >=Stealth, color = red, dashed, line width=0.9pt] (9.95,3.4) -- (7.9,2.5);
\draw  [->, >=Stealth, color = red, dashed, line width=0.9pt] (9.95,3.4) -- (12,2.5);
\end{circuitikz}
}%
\par\end{centering}
\caption{Model architecture and components. The primary methodological contribution
in this paper is the core correlation model.\label{fig:ModelStructure}}
\end{figure}

\subsection{Model for Standardized Factor Variables}

We model the standardized factor variables, $F_{t}$, using the score-driven
model by \citet{TongHansenArchakov:2024}. Because $r$ is relatively
small, there is not need to impose structure on $C_{F_{t}}=\operatorname{corr}(F_{t})$.
The score-drive framework was introduced by \citet{CrealKoopmanLucas:2013},
where dynamic parameters are updated based on the score of the log-likelihood
function. Here we specify a score-model for the vector representation
of the correlation matrix, $\gamma_{t}^{F}=\operatorname{vecl}(\log C_{F_{t}})$,
specifically
\[
\gamma_{t+1}^{F}=\kappa^{F}+\beta^{F}\gamma_{t}^{F}+\alpha^{F}\varepsilon_{t}^{F},
\]
where $\kappa^{F}=(I_{r(r-1)/2}-\beta^{F})\mu^{F}$ with $\mu^{F}=\mathbb{E}[\gamma_{t}^{F}]$,
and $\alpha^{F}$ and $\beta^{F}$ are $r(r-1)/2\times r(r-1)/2$
matrices and $\ensuremath{\varepsilon_{t}^{F}=\mathcal{S}_{t}^{-1}\nabla_{\gamma_{t}^{F}}}$
where $\mathcal{S}_{t}$ is a scaling matrix and, $\nabla_{\gamma_{t}^{F}}=\partial\ell(F_{t})/\partial\gamma_{t}^{F}$,
is the score of the log-likelihood function. The distributional assumption
is that $F_{t}=C_{F,t}^{1/2}U_{t}$ with $U_{t}\sim\mathrm{CT}_{\boldsymbol{d},\boldsymbol{\nu}^{F}}^{{\rm std}}(0,I_{d})$,
and we consider three types of distributions for $U_{t}$: Gaussian,
Multivariate-$t$ (MT), and (HT) distributions.

\subsection{Dynamic Factor Loadings and Idiosyncratic Correlations}

The \textit{core correlation model} includes the dynamic models of
factor loadings and the idiosyncratic correlation matrix. The is the
central part of the proposed model, which describes the conditional
distribution of $Z_{t}$ given $U_{t}$, which involves dynamic factor
loadings and a dynamic idiosyncratic correlation matrix with various
structures. We adopt the convolution-$t$ distribution with $\ensuremath{\Xi}=C_{e}^{1/2}$
for this part of the model, such that
\begin{equation}
Z=\boldsymbol{\rho}^{\prime}U+\Lambda_{\omega}e,\quad e\sim\mathrm{CT}_{\boldsymbol{m},\boldsymbol{\nu}}^{{\rm std}}(0,C_{e}^{1/2}),\label{eq:condmodel}
\end{equation}
Here $C_{e}^{1/2}$ is the symmetric square-root of $C_{e}$, and
it follows that $Z|U\sim\mathrm{CT}_{\boldsymbol{m},\boldsymbol{\nu}}^{{\rm std}}(\boldsymbol{\rho}^{\prime}U,\Lambda_{\omega}C_{e}^{1/2})$.
We parametrize the factor loadings by $\tau$, whereas the idiosyncratic
correlation matrix is parametrized by $\eta$, using various types
of block structures, $\gamma(C_{e})=B\eta$. Since $\Lambda_{\omega}$
is a function of $\boldsymbol{\rho}$, the parameters dynamic parameters
are represented by the vector, $\zeta_{t}=(\tau_{t}^{\prime},\eta_{t}^{\prime})^{\prime}$,
which leads to a second score-driven model,
\begin{equation}
\zeta_{t+1}=\kappa+\beta\zeta_{t}+\alpha\varepsilon_{t},\label{eq:VAR1}
\end{equation}
where $\beta$ and $\alpha$ are coefficient matrices, $\kappa=(I_{p}-\beta)\mu_{\zeta}$
with $\mu_{\zeta}=\mathbb{E}[\zeta_{t}]$, and $\varepsilon_{t}$
is defined by the score in period $t$.\footnote{It is straightforward to include additional lagged values of $\zeta_{t}$,
such that (\ref{eq:VAR1}) has a higher-order VAR(p) structure, and
adding $q$ lagged values of $\varepsilon_{t}$, would generalize
(\ref{eq:VAR1}) to a VARMA(p,q) model, we do not pursue these extensions
in this paper.} In the empirical analysis we let $\beta$ and $\alpha$ be diagonal
matrices to keep the model relatively parsimonious.

A key aspect of a score-driven model is the innovation, $\varepsilon_{t}=\mathcal{S}_{t}^{-1}\nabla_{\zeta_{t}}$,
where $\nabla_{\zeta_{t}}=\partial\ell(Z_{t}|U_{t})/\partial\zeta_{t}$
is the score of the predictive log-likelihood function and $\mathcal{S}^{-1}$
is a suitable scaling matrix. This model structure is very intuitive,
because (\ref{eq:VAR1}) continuously updates $\zeta_{t}$ in the
direction dictated by the first-order conditions of the log-likelihood. 

To simplify the expositions, we suppress subscript-$t$ in the rest
of this Section. For the convolution-$t$ distribution we need the
following result that follows from \citet[theorems 5 and 6]{TongHansenArchakov:2024}.
\begin{lem}
Suppose that $X\sim\mathrm{CT}_{\boldsymbol{m},\boldsymbol{\nu}}^{{\rm std}}(\mu,\Xi)$
and define $V=(V_{1},\ldots,V_{G})=\Xi^{-1}\left(X-\mu\right)$, where
$V_{g}\in\mathbb{R}^{n\times m_{g}}$ with $\sum_{g=1}^{G}m_{g}=n$.
The partial derivatives (scores) with respect to $\mu$ and ${\rm vec}\left(\Xi\right)$
are given by 
\[
\nabla_{\mu}=\ensuremath{\sum_{g=1}^{G}W_{g}\Xi^{\prime-1}P_{g}V_{g}},\qquad\nabla_{\Xi}=\sum_{g=1}^{G}W_{g}{\rm vec}\left(\Xi^{\prime-1}P_{g}V_{g}V^{\prime}\right)-{\rm vec}\left(\Xi^{\prime-1}\right),
\]
respectively, where $W_{g}=\frac{\nu_{g}+m_{g}}{\nu_{g}-2+V_{g}^{\prime}V_{g}}$.
The corresponding terms in the information matrix are
\begin{eqnarray*}
\mathcal{I}_{\mu} & = & \sum_{g=1}^{G}\tfrac{\left(\nu_{g}+m_{g}\right)\nu_{g}}{\left(\nu_{g}+m_{g}+2\right)\left(\nu_{g}-2\right)}\Xi^{\prime-1}P_{g}P_{g}^{\prime}\Xi^{-1},\\
\mathcal{I}_{\Xi} & = & \left(I_{n}\otimes\Xi^{\prime-1}\right)\ensuremath{\left(K_{n}+\Upsilon_{G}\right)\left(I_{n}\otimes\Xi^{-1}\right),}
\end{eqnarray*}
and $\mathcal{I}_{\mu\Xi}=0$, where $K_{n}$ is the commutation matrix,
$P_{g}$ \textup{is a $n\times m_{g}$ matrix from identity matrix
$I_{n}=\left(P_{1},\ldots,P_{G}\right)$, and $\Upsilon_{G}$ is related
to $\bm{m}$ and $\bm{\nu}$, see} \citet{TongHansenArchakov:2024}.
\end{lem}
Note that $W_{g}$ mitigates the influence that an outlier in $V_{g}$
has on the scores. For latter use, we define
\begin{equation}
\xi\equiv\left(\begin{array}{c}
{\rm vec}\left(\mu,\omega\right)^{\prime}\\
\eta
\end{array}\right),\quad\text{ and}\quad\Pi\equiv\frac{\partial\xi}{\partial\zeta^{\prime}}=\left(\begin{array}{cc}
M & \bm{0}\\
\bm{0} & I
\end{array}\right),\quad\text{where}\ M=\frac{\partial{\rm vec}\left(\mu,\omega\right)^{\prime}}{\partial\tau^{\prime}},\label{eq:DefineXi}
\end{equation}
where $\mu=\boldsymbol{\rho}^{\prime}U$ and $\omega=\left(\omega_{1},\omega_{2},\ldots,\omega_{n}\right)^{\prime}$
with $\omega_{i}=\sqrt{1-\rho_{i}^{\prime}\rho_{i}}$ for $i=1,\ldots,n$.
Next, the following Theorem provides expressions for key terms in
the score-driven model.
\begin{thm}[Key terms in Score-Driven Correlation Model]
\label{thm:JointScore}Suppose that $e\sim\mathrm{CT}_{\boldsymbol{m},\boldsymbol{\nu}}^{{\rm std}}(0,C_{e}^{1/2})$
such that $Z|U\sim\mathrm{CT}_{\boldsymbol{m},\boldsymbol{\nu}}^{{\rm std}}(\mu,\Xi)$
with $\mu=\boldsymbol{\rho}^{\prime}U$ and $\Xi=\Lambda_{\omega}C_{e}^{1/2}$,
then the score vector and information matrix with respect to the dynamic
parameter, $\zeta$, are given by
\[
\nabla_{\zeta}=\Pi^{\prime}\nabla_{\xi},\qquad\mathcal{I}_{\zeta}=\Pi^{\prime}\mathcal{I}_{\xi}\Pi,
\]
where
\[
\nabla_{\xi}=\Theta^{\prime}\left(\begin{array}{c}
\begin{array}{c}
\nabla_{\mu}\\
\nabla_{\Xi}
\end{array}\end{array}\right),\quad\mathcal{I}_{\xi}=\Theta^{\prime}\left(\begin{array}{cc}
\mathcal{I}_{\mu} & 0\\
0 & \mathcal{I}_{\Xi}
\end{array}\right)\Theta.
\]
with $\theta=\left[\mu^{\prime},{\rm vec}\left(\Xi\right)^{\prime}\right]^{\prime}$,
we have
\begin{align*}
\Theta & =\frac{\partial\theta}{\partial\xi^{\prime}}=\left[\begin{array}{ccc}
I_{n} & \bm{0} & \bm{0}\\
\bm{0} & \frac{\partial{\rm vec}\left(\Xi\right)}{\partial\omega^{\prime}} & \frac{\partial{\rm vec}\left(\Xi\right)}{\partial\eta^{\prime}}
\end{array}\right]\left[\begin{array}{cc}
K_{2n} & \bm{0}\\
\bm{0} & I_{n(n-1)/2}
\end{array}\right],
\end{align*}
where $K_{2n}$ is the communication matrix and
\[
\tfrac{\partial{\rm vec}\left(\Xi\right)}{\partial\omega^{\prime}}=(C_{e}^{1/2}\otimes I_{n})E_{d}^{\prime},\qquad\tfrac{\partial{\rm vec}\left(\Xi\right)}{\partial\eta^{\prime}}=(I_{n}\otimes\Lambda_{\omega})(C_{e}^{1/2}\oplus I_{n})^{-1}\tfrac{\partial{\rm vec}\left(C_{e}\right)}{\partial\gamma^{\prime}}B.
\]
where $E_{d}$ is the elimination matrix such that ${\rm diag}\left(S\right)=E_{d}{\rm vec}\left(S\right)$
for any square matrix S. The expression for $\partial{\rm vec}(C_{e})/\partial\gamma^{\prime}$
is in Appendix \ref{sec:Proof-of-Theorem1}. The matrix $M$ in (\ref{eq:DefineXi})
is given by
\[
M=K_{n2}\left[\begin{array}{c}
\left(U^{\prime}\otimes I_{n}\right)\sum_{i=1}^{n}\left(I_{d}\otimes P_{i}\right)J_{i}\left(P_{i}^{\prime}\otimes I_{d}\right)\\
-\frac{1}{2}\text{\ensuremath{\Lambda_{\omega}^{-1}}}E_{d}\ensuremath{\left(I_{n^{2}}+K_{n}\right)\left(\boldsymbol{\rho}^{\prime}\otimes I_{n}\right)}\sum_{i=1}^{n}\left(I_{d}\otimes P_{i}\right)J_{i}\left(P_{i}^{\prime}\otimes I_{d}\right)
\end{array}\right],
\]
where $P_{i}$ is the $i$-th column of identity matrix $I_{n}$,
and $J_{i}$ is the Jacobian matrix $\partial\rho_{i}/\partial\tau_{i}$.
\end{thm}
\begin{proof}
See Appendix \ref{sec:Proof-of-Theorem1}.
\end{proof}
A common choice for the scaling matrix is the Fisher information $S_{t}=\mathcal{I}_{\zeta_{t}}=\mathbb{E}\left(\nabla_{\zeta_{t}}\nabla_{\zeta_{t}}^{\prime}\right)$,
see \citet{CrealKoopmanLucas:2013}. However, in the present context,
$\mathcal{I}_{\zeta_{t}}$ will not be invertible if $r\geq2$, which
stems from the conditional mean, $\rho_{i}^{\prime}U_{}$, being univariate,
while $\rho_{i}$ is $r$-dimensional. 

A natural starting point is to use the Moore-Penrose inverse, $\mathcal{I}_{\zeta}^{+}$,
of $\mathcal{I}_{\zeta}$. Unfortunately, this choice turns out to
be numerically unstable, as we explain in the next subsection. 

\subsubsection{Moore-Penrose Inverse of Information Matrix}

Note that the parameter vector $\xi$ defined in (\ref{eq:DefineXi})
is a minimal parameterization for the convolution-$t$ distribution,
and the information matrix for $\zeta$ can be expressed as $\mathcal{I}_{\zeta}=\Pi^{\prime}\mathcal{I}_{\xi}\Pi$
with its Moore-Penrose inverse given by
\[
\mathcal{I}_{\zeta}^{+}=\Pi^{+}\mathcal{I}_{\xi}^{-1}\Pi^{+\prime},\quad{\rm with}\quad\Pi^{+}=\left(\begin{array}{cc}
M^{+} & \bm{0}\\
\bm{0} & I
\end{array}\right),
\]
such that the scaled innovation $\varepsilon$ under Moore-Penrose
inverse, denote $\varepsilon_{0}$, is given by
\begin{equation}
\varepsilon_{0}=\mathcal{I}_{\zeta}^{+}\nabla_{\zeta}=\Pi^{+}\mathcal{I}_{\xi}^{-1}\nabla_{\xi}\label{eq:e0}
\end{equation}
where we use $\nabla_{\zeta}=\Pi^{\prime}\nabla_{\xi}$ and the fact
that $\Pi^{+\prime}\Pi^{\prime}=I$. The matrix $M^{+}$ in $\Pi^{+}$
is given by
\[
M^{+}={\rm diag}\left(M_{1}^{+},M_{2}^{+},\ldots,M_{n}^{+}\right),\quad M_{i}^{+}=\left(\begin{array}{c}
\frac{\partial\mu_{i}}{\partial\tau_{i}^{\prime}}\\
\frac{\partial\omega_{i}}{\partial\tau_{i}^{\prime}}
\end{array}\right)^{+},\quad i=1,2,\ldots,n.
\]
with the following expression for $M_{i}^{+}$
\begin{align}
M_{i}^{+} & =\frac{1}{U^{\prime}U\rho_{i}^{\prime}\rho_{i}-\left(U^{\prime}\rho_{i}\right)^{2}}J_{i}^{-1}\left[\begin{array}{c}
\rho_{i}^{\prime}\rho_{i}U^{\prime}-\rho_{i}^{\prime}U\rho_{i}^{\prime}\\
\omega_{i}\left(\rho_{i}^{\prime}UU^{\prime}-U^{\prime}U\rho_{i}^{\prime}\right)
\end{array}\right]^{\prime}.\label{eq:PinvMi}
\end{align}
This expression clarifies the reason behind the numerical instability
for $\mathcal{I}_{\zeta}^{+}$. The random vector, $U$, can be proportional
to (or nearly proportional to) $\rho_{i}$, (which has non-trivial
probability). As the angle between $U$ and $\rho_{i}$ vanishes,
the denominator vanishes to zero at a faster rate than the numerator,
resulting in (arbitrarily) large elements of $M_{i}^{+}$. To avoid
this issue, we adopt the Tikhonov regularized Moore-Penrose inverse
as our scaling matrix. The Tikhonov regularization involves a penalty
term, $\lambda_{i}\geq0$, for each assets, $i=1,\ldots,n$. 

\subsubsection{Tikhonov Regularized Moore-Penrose Inverse}

Let $\bm{\lambda}=\left(\lambda_{1},\ldots,\lambda_{n}\right)^{\prime}$
be a vector of non-negative penalty terms and define
\begin{equation}
M_{i,\lambda_{i}}^{+}=M_{i}^{\prime}\left(M_{i}M_{i}^{\prime}+\lambda_{i}I_{2}\right)^{+},\quad{\rm where}\quad M_{i,\lambda_{i}}^{+}={\rm argmin}_{S}\ \left\Vert SM_{i}-I_{r}\right\Vert _{2}+\lambda_{i}\left\Vert S\right\Vert _{2}.\label{eq:PinvMlambda}
\end{equation}
This is the Tikhonov regularized Moore-Penrose inverse where the boundary
case, $\lambda_{i}=0$, corresponds to the usual Moore-Penrose inverse
in (\ref{eq:PinvMi}). The scaled innovation in (\ref{eq:e0}) now
can be adapted to
\begin{equation}
\varepsilon_{\bm{\lambda}}=\Pi_{\bm{\lambda}}^{+}\mathcal{I}_{\xi}^{-1}\nabla_{\xi},\label{eq:e_Lambda}
\end{equation}
where the Tikhonov regularized matrix $\Pi_{\bm{\lambda}}^{+}$ is
given by
\[
\Pi_{\bm{\lambda}}^{+}=\left(\begin{array}{cc}
M_{\bm{\lambda}}^{+} & \bm{0}\\
\bm{0} & I
\end{array}\right),\quad\text{with}\quad M_{\bm{\lambda}}^{+}={\rm diag}\left(M_{1,\lambda_{1}}^{+},M_{2,\lambda_{2}}^{+},\ldots,M_{n,\lambda_{n}}^{+}\right).
\]
Note that when $\bm{\lambda}=0$, the scaled innovation $\varepsilon_{\bm{\lambda}}$
will return to $\varepsilon_{0}$ in (\ref{eq:e0}).\footnote{One could also consider a regularization based on a common penalty
parameter, $\lambda_{i}=\lambda^{*}$ for all $i=1,2\ldots,n.$ We
do not explore this in this paper.}

In our empirical analysis, we compare three definitions of $\varepsilon$,
given by $\varepsilon_{0}=\Pi^{+}\mathcal{I}_{\xi}^{-1}\nabla_{\xi}$
(Moore-Penrose), $\varepsilon_{\bm{\lambda}}=\Pi_{\bm{\lambda}}^{+}\mathcal{I}_{\xi}^{-1}\nabla_{\xi}$
(Tikhonov), and the very simple choice, $\varepsilon^{*}=\nabla_{\zeta}$
corresponding to $S=I$ with the identity as the scaling matrix. In
specifications using $\varepsilon_{\bm{\lambda}}$, the vector of
penalty parameters, $\bm{\lambda}$, is estimated and the empirical
results clearly favors this choice for scaling matrix. 

\subsection{Decoupled Estimation}

When $n$ is large it becomes necessary to simplify the estimation
problem further. This can be achieved by decoupling the dynamic structure
for factor loadings from that of the idiosyncratic correlation matrix. 

We propose an estimation strategy where the $n$ score-driven models
for $\rho_{i}$, $i=1,\ldots,n$ are estimated separately. These low-dimensional
score-driven models yield the idiosyncratic residuals, $\hat{e}_{i,t}$,
$i=1,\ldots,n$, that are subsequently used to estimate a score-driven
model for the idiosyncratic correlation matrix $C_{e}$.

\subsubsection{Decoupled Model for $i$-th Factor Loadings\label{subsec:DynamicLoadings}}

To construct a score-driven model for the dynamic factor loading $\rho_{i}$
using information from $Z_{i}$, we leverage the fact that the marginal
distribution of a convolution-$t$ distribution can be well approximated
by a univariate Student's $t$ distribution, see \citet{Patil:1965},
\citet{Alcaraz2023}, and references therein. Here we will approximating
the marginal distribution with the $t$-distribution that minimizes
Kullback-Leibler divergence. 

Therefore, we assume that
\[
Z_{i}=\rho_{i}^{\prime}U+\omega_{i}e_{i},\quad e_{i}\simeq t_{\nu_{i}^{\star}}^{{\rm std}}\left(0,1\right)
\]
It is important to note that the approximating $t$-distribution,
$t_{\nu_{i}^{\star}}^{\text{std }}(0,1)$, is only used to facilitate
the construction of a score-driven model for $\rho_{i}$. Each of
the score-driven models are estimated by maximizing the quasi log-likelihood
function, defined by the approximating $t$-distribution, $t_{\nu_{i}^{\star}}^{\text{std }}\left(\mu_{i},\omega_{i}\right)$,
where $\mu_{i}=\rho_{i}^{\prime}U$ and $\omega_{i}=\sqrt{1-\rho_{i}^{\prime}\rho_{i}}$.
The corresponding quasi log-likelihood function is given by
\begin{align*}
\ell^{\star}(Z_{i}|U) & =c_{\nu_{i}^{\star}}-\log\left(\omega_{i}\right)-\tfrac{\nu_{i}^{\star}+1}{2}\log\left(1+\tfrac{e_{i}^{2}}{\nu_{i}^{\star}-2}\right),
\end{align*}
where $c_{\nu_{i}^{\star}}$ is a normalizing constant. From (\ref{eq:e_Lambda}),
Theorem \ref{thm:SeqDynamicLoadings} shows that the scaled score
used to update the dynamics of $\tau\left(\rho_{i}\right)$ is given
by
\[
\varepsilon_{\lambda_{i}}=M_{i,\lambda_{i}}^{+}\mathcal{I}_{\xi_{i}}^{-1}\nabla_{\xi_{i}},\quad{\rm where}\quad\nabla_{\xi_{i}}=\tfrac{\partial\ell\left(Z_{i}|U\right)}{\xi_{i}},\quad\xi_{i}=\left(\mu_{i},\omega_{i}\right)^{\prime},
\]
and $M_{i,\lambda_{i}}^{+}$ is defined in (\ref{eq:PinvMlambda}).
Note that in this modeling strategy, the marginal densities of $Z_{i}$
are approximated by a univariate Student's $t$-distribution with
degrees of freedom $\nu_{i}^{\star}$. The parameter $\nu_{i}^{\star}$
is estimated using a score-driven approach, which can be interpreted
as minimizing the Kullback-Leibler (KL) divergence between the true
marginal density function and the approximating density function.
\begin{thm}
\label{thm:SeqDynamicLoadings}Suppose that $e_{i}\sim t_{\nu_{i}^{\star}}^{{\rm std}}\left(0,1\right)$
such that $Z_{i}|U\sim t_{\nu_{i}^{\star}}^{{\rm std}}\left(\mu_{i},\omega_{i}^{2}\right)$
with $\mu_{i}=\rho_{i}^{\prime}U$ and $\omega_{i}=\sqrt{1-\rho_{i}^{\prime}\rho_{i}}$,
then the score vector and information matrix with respect to the dynamic
parameter, $\tau\left(\rho_{i}\right)$, are given by
\begin{align*}
\nabla_{\tau_{i}} & =J_{i}\left[\tfrac{W_{i}e_{i}}{\omega_{i}}U+\tfrac{1-W_{i}e_{i}^{2}}{\omega_{i}^{2}}\rho_{i}\right]\\
\mathcal{I}_{\tau_{i}} & =J_{i}\left[\tfrac{\left(\nu_{i}^{\star}+1\right)\nu_{i}^{\star}}{\left(\nu_{i}^{\star}+3\right)\left(\nu_{i}^{\star}-2\right)}\tfrac{UU^{\prime}}{\omega_{i}^{2}}+\tfrac{2\nu_{i}^{\star}}{\nu_{i}^{\star}+3}\tfrac{1-\omega_{i}^{2}}{\omega_{i}^{4}}\right]J_{i}
\end{align*}
where $W_{i}=\frac{\nu_{i}^{\star}+1}{\nu_{i}^{\star}-2+e_{i}^{2}}$
and $J_{i}$ is the Jacobian matrix $\partial\rho_{i}/\partial\tau_{i}$.
With $\xi_{i}\equiv\left(\mu_{i},\omega_{i}\right)^{\prime}$, we
have
\[
\nabla_{\xi_{i}}=\frac{1}{\omega_{i}}\left(\begin{array}{c}
\begin{array}{c}
W_{i}e_{i}\\
W_{i}e_{i}^{2}-1
\end{array}\end{array}\right),\quad\mathcal{I}_{\xi_{i}}=\frac{1}{\omega_{i}^{2}}\left[\begin{array}{cc}
\tfrac{\left(\nu_{i}^{\star}+1\right)\nu_{i}^{\star}}{\left(\nu_{i}^{\star}+3\right)\left(\nu_{i}^{\star}-2\right)} & 0\\
0 & \frac{2\nu_{i}^{\star}}{\nu_{i}^{\star}+3}
\end{array}\right],
\]
such that the scaled score for updating the dynamics of $\tau_{i}$
is given by $\varepsilon_{\lambda_{i}}=M_{i,\lambda_{i}}^{+}\mathcal{I}_{\xi_{i}}^{-1}\nabla_{\xi_{i}}$,
where $M_{i,\lambda_{i}}^{+}$ is defined in (\ref{eq:PinvMlambda}).
\end{thm}

\subsubsection{Score-driven Model for Idiosyncratic Correlation Matrix $C_{e}$\label{subsec:Second-Step}}

The final stage of the decoupled estimation method is defined by a
score-driven model for the idiosyncratic correlation matrix $C_{e}$
using the estimated residuals $\hat{e}=\left(e_{1},e_{2},\ldots,e_{n}\right)^{\prime}$
from the first stage of decoupled estimation. Our distributional specifications
are variants of the convolution $t$-distribution $e\sim\mathrm{CT}_{\boldsymbol{m},\boldsymbol{\nu}}^{{\rm std}}(0,C_{e}^{1/2})$,
and we consider four structure of $C_{e}$. (1) The most flexible
structure is an Unrestricted Correlation matrix, which is heavily
parameterized. A dynamic model for the elements in an unrestricted
$C_{e}$ is therefore only practical when the number of assets $n$
is small. We estimate the model with $n=12$ in our empirical analysis,
but it will be difficult to increase the dimension much beyond 12.
For larger dimensions we will need to impose structure on $C_{e}$.
(2) The Block Correlation matrix for $C_{e}$ offers a way to reduce
the number of free parameters, we will use subindustries to define
the block structure. When the number of blocks is large (e.g. $K=63$
which is used in some of our empirical analyses) additional structure
is needed, and we consider two sparse variants that achieve this:
(3) The Sparse Block Correlation matrix, which has idiosyncratic correlations
for stocks in different sectors to be zero, and (4) Diagonal Block
Correlation matrix for $C_{e}$, where the only nontrivial correlations
are between stocks in the same subindustry. 

The two sparse variants of $C_{e}$ greatly simplify the estimation
of the dynamic model for $C_{e}$, because it can be decomposed in
to estimation involving subvectors of $e$ that are specific to sectors
or subindustries. For instance, with the Diagonal Block structure
we have $C_{e}={\rm diag}\left(C_{[1,1]},C_{[2,2]},\ldots,C_{[K,K]}\right)$
where $C_{[k,k]}$ is an equicorrelation matrix with coefficient $\varrho_{kk}$.
The unique element in $\log(C_{[k,k]})$, $\eta_{k}\in\mathbb{R}$,
is given from (\ref{eq:gamma_k}), and it can be modeled as an unrestricted
parameter. Let $e_{[k]}$ denote the corresponding subvector of $e$
such that $\mathrm{var}(e_{[k]})=C_{[k,k]}$. These subvectors are
independent for the Gaussian, HT, and CT distributions (provide that
the cluster structure in CT is also based on subindustries such that
$e_{(k)}\sim t_{n_{k},\nu_{k}}^{{\rm std}}(0,C_{[k,k]})$). This independence
allows us to estimate the dynamics of $\eta_{k}$ separately by using
information from $e_{[k]}$ alone. The scaled score for updating $\eta_{k}$
is then given by
\[
\varepsilon_{[k]}=\mathcal{I}_{\eta_{k}}^{-1}\nabla_{\eta_{k}},\quad\text{where}\quad\nabla_{\eta_{k}}=\tfrac{\partial\ell\left(e_{[k]}\right)}{\eta_{k}}.
\]
The formulas for the score $\nabla_{\eta_{k}}$ and the information
matrix $\mathcal{I}_{\eta_{k}}$ are provided in Theorem \ref{thm:TquiCorr}.
\begin{thm}[Multivariate-$t$ with Equicorrelation Matrix $C$]
\label{thm:TquiCorr}Suppose that $X\sim t_{n,\nu}^{{\rm std}}(0,C)$,
where $C$ is an equicorrelation matrix with correlation coefficient
$\varrho$. Then the score and information matrix with respect to
dynamic parameter, $\eta=\frac{1}{n}\log\left(1+n\frac{\varrho}{1-\varrho}\right)$,
are given by
\begin{align*}
\nabla_{\eta} & =-\tfrac{J}{2}\left[\tfrac{\left(n-1\right)}{1+\left(n-1\right)\varrho}-\tfrac{\left(n-1\right)}{1-\varrho}+W\left(\tfrac{X^{\prime}X}{\left(1-\varrho\right)^{2}}-\tfrac{\left[1+\left(n-1\right)\rho^{2}\right]X^{\prime}\iota_{n}\iota_{n}^{\prime}X}{\left(1-\varrho\right)^{2}\left[1+(n-1)\varrho\right]^{2}}\right)\right],\\
\mathcal{I}_{\eta} & =J^{2}\left[\tfrac{1}{4}\tfrac{\left(3\phi-1\right)\left(n-1\right)^{2}}{\left(1+\left(n-1\right)\varrho\right)^{2}}+\tfrac{\phi}{2}\tfrac{\left(n-1\right)}{\left(1-\varrho\right)^{2}}+\tfrac{1-\phi}{4}\tfrac{\left(1-\left(n+1\right)\varrho\right)\left(n-1\right)^{2}}{\left(1+\left(n-1\right)\varrho\right)\left(1-\varrho\right)^{2}}\right],
\end{align*}
where $J_{}=\frac{1}{\left(1-\varrho\right)\left(1+\left(n-1\right)\varrho\right)}$
and $\iota_{n}$ is a n-dimensional vector of ones, and
\[
\phi=\tfrac{\nu+n}{\nu+n+2},\qquad W=\tfrac{\nu+n}{\nu-2+Q},\qquad Q=\tfrac{X^{\prime}X}{1-\varrho}-\tfrac{\varrho X^{\prime}\iota_{n}\iota_{n}^{\prime}X}{\left(1-\varrho\right)\left[1+(n-1)\varrho\right]}.
\]
\end{thm}
\noindent When $e$ follows the HT distribution, the subvectors $e_{[k]}$
are mutually independent and follow a HT distribution with $\boldsymbol{\nu}_{k}$
degrees of freedom, where $\boldsymbol{\nu}_{k}$ is the subvector
of $\boldsymbol{\nu}$ that corresponds to the $k$-th group. In this
case, we can still estimate the dynamics of each of the $\eta_{k}$-vectors
separately, solely using the information provided by $e_{[k]}$. The
formulas for $\nabla_{\eta_{k}}$ and $\mathcal{I}_{\eta_{k}}$ are
provided in Theorem \ref{thm:HTquiCorr}.
\begin{thm}[HT distribution with Equicorrelation Matrix $C$]
\label{thm:HTquiCorr}Suppose that $X\sim\mathrm{CT}_{\boldsymbol{n},\boldsymbol{\nu}}^{{\rm std}}(0,C^{1/2})$
, where $C$ is an equicorrelation matrix with correlation coefficient
$\varrho$ and $\bm{n}=\iota_{n}$. Then the score vector and information
matrix with respect to the dynamic parameters, $\eta=\frac{1}{n}\log\left(1+n\frac{\varrho}{1-\varrho}\right)$,
are given by
\begin{align*}
\nabla_{\eta} & =-\tfrac{J}{2}\ensuremath{\left[\tfrac{\left(n-1\right)}{1+\left(n-1\right)\varrho}-\tfrac{\left(n-1\right)}{1-\varrho}\right]}-\tfrac{J}{2}\sum_{i=1}^{n}W_{i}V_{i}\left[\tfrac{\left(nX_{i}-\iota_{n}^{\prime}X\right)}{n\left(1-\varrho\right)^{3/2}}-\tfrac{\left(n-1\right)\iota_{n}^{\prime}X}{n\left(1+(n-1)\varrho\right)^{3/2}}\right]\\
\mathcal{I}_{\eta} & =J^{2}\left[\tfrac{1}{4}\tfrac{\left(n^{2}+\mathcal{A}_{1}\right)\left(n-1\right)^{2}}{n^{2}\left(1+\left(n-1\right)\varrho\right)^{2}}+\tfrac{1}{4}\tfrac{\left(n-1\right)\mathcal{A}_{2}}{n^{2}\left(1-\varrho\right)^{2}}+\tfrac{1}{2}\tfrac{\left(n-1\right)^{2}\mathcal{A}_{3}}{n^{2}\left(1+\left(n-1\right)\varrho\right)\left(1-\varrho\right)}\right]
\end{align*}
where $J$ and $\iota_{n}$ are defined in Theorem \ref{thm:TquiCorr},
and 
\[
W_{i}=\tfrac{\nu_{i}+1}{\nu_{i}-2+V_{i}^{2}},\quad V_{i}=\tfrac{nX_{i}-\iota_{n}^{\prime}X}{n\sqrt{1-\varrho}}+\tfrac{\iota_{n}^{\prime}X}{n\sqrt{1+(n-1)\varrho}}
\]
\[
\mathcal{A}_{1}=3\iota_{n}^{\prime}\phi+\left(n-1\right)\iota_{n}^{\prime}\psi-2n,\hspace{0.75em}\mathcal{A}_{2}=\left(3\iota_{n}^{\prime}\phi-n\right)\left(n-1\right)+\iota_{n}^{\prime}\psi+n,\hspace*{0.75em}\mathcal{A}_{3}=\iota_{n}^{\prime}\psi+2n-3\iota_{n}^{\prime}\phi
\]
where $\phi$ and $\psi$ are both vectors with $\phi_{i}=\frac{\nu_{i}+1}{\nu_{i}+3}$,
$\psi_{i}=\frac{\phi_{i}\nu_{i}}{\nu_{i}-2}$ for $i=1,2,\ldots,n$.
\end{thm}
Under the sparse block structure where blocks are defined by subindustries,
we can express $C_{e}={\rm diag}\left(C_{s_{1}},C_{s_{2}},\ldots,C_{s_{N}}\right)$
where $C_{s_{j}}$ is an block correlation matrix of the $j$-th sector
with group size $s_{j}$, for $j=1,\ldots,N$. From the formula (\ref{eq:eta_can}),
the unique element in $\log(C_{s_{j}})$ is $\eta_{j}\in\mathbb{R}^{s_{j}\left(s_{j}+1\right)/2}$,
which can be modeled in an unrestricted way. Let $e_{\{j\}}$ denote
the subvector of $e$ that corresponds to the $j$-th sector, such
that $e_{\{j\}}\sim(0,C_{s_{j}})$. When combined with either the
Gaussian, CT, or HT distribution, $e_{\{j\}}$ will be mutually independent.
This allows us to estimate the dynamics of $\eta_{j}$ separately,
using only the information from $e_{\{j\}}$. The scaled score for
updating $\eta_{j}$ is given by
\[
\varepsilon_{\{j\}}=\mathcal{I}_{\eta_{j}}^{-1}\nabla_{\eta_{j}},\quad\text{where}\quad\nabla_{\eta_{j}}=\tfrac{\partial\ell\left(e_{\{j\}}\right)}{\eta_{j}}.
\]
When $e$ following CT or HT distribution, subvectors $e_{\{j\}}$
will also have a CT or HT distribution, $e_{\{j\}}\sim\text{CT}_{\boldsymbol{n}_{j},\boldsymbol{\nu}_{j}}^{{\rm std}}(0,C_{s_{1}}^{1/2})$,
with dimensions and degrees of freedom inherited from the parents
distribution.

Note that this factorization is not feasible with the MT distribution,
as its subvectors are not mutually independent. As a result, a score-driven
model with the MT distribution cannot be estimated in this way, making
it practical only in low-dimensional settings.

\section{Empirical Analysis\label{sec:EmpiricalAnalysis}}

\subsection{Data}

Our empirical analysis is based on daily close-to-close returns of
constituents of the S\&P 500 index and return series for each of the
factor variables.\footnote{The data were obtained from the CRSP database of WRDS.}
Our sample period is from January 2007 to December 2023, with a total
of $T=4,278$ trading days.

\subsubsection{Individual Returns series}

The initial set of stocks consists of all stocks that were constituents
of the S\&P 500 index at some point during the sample period. We excluded
stocks using the following exclusion criteria: (E1) stocks for which
returns were not available over the full sample period; (E2) stocks
from the Real Estate and Communication Services sectors, as is standard
in the existing literature;\footnote{The Real Estate sector was introduced as the 11th sector in 2016,
and the Communication Services sector was introduces as an expanded
and rebranded version of the Telecommunication Services sector. One
consequence of this is that  the corresponding sector-specific SPDR
ETFs, (XLRE) and (XLC), are not available in the full sample period.} (E3) stocks in sub-industries with fewer than 3 stocks (after applying
E1 and E2). This resulted in a balanced sample with 323 stocks from
63 subindustries in 9 sectors. These are listed in Table \ref{tab:SeqEstLarge1}
with ticker symbols under the 63 sub-industries in the Supplemental
Material. We use the Global Industry Classification Standard (GICS)
to sort stocks by their eight-digit GICS code in ascending order,
and use sub-industries to define block structures in the idiosyncratic
correlation matrix.

Before we analyze the full set of $n=323$ stocks, which we refer
to as the \emph{Large Universe}, we will analyze a smaller subset
with $n=12$ stocks, which we label the \emph{Small Universe}. The
Small Universe provides a framework where it is possible to make comparisons
with existing models. The simpler model also enables us to compare
joint estimation of all model parameters with the two-stage method
we use for the large universe.

\subsubsection{Factor Variables}

We use as many as $r=15$ factors. These include those in the Fama-French
five-factor model (FF5)\footnote{FF5 represents the following portfolios: Market (MKT), small-minus-big
(SMB), high-minus-low (HML), robust-minus-weak (RMW), and conservative-minus-aggressive
(CMA).} and the momentum factor (UMD) by \citet{Carhart:1997}. These six
cross-sectional factors were obtained from the \citet{french_datalibrary}
Data Library. We also include nine sector-specific factors using returns
for SPDR ETFs: Energy (XLE), Materials (XLB), Industrials (XLI), Consumer
Discretionary (XLY), Consumer Staples (XLP), Health Care (XLV), Financials
(XLF), Information Technology (XLK), and Utilities (XLU), as in \citet{FanFurgerXiu:2016},
\citet{YacineXiu2017}, \citet{DaiLuXiu:2019}, and \citet{Bodilsen:2024}. 

\subsection{Standardized Asset Returns and Factor Variables}

All model-specifications are based on the standardized variables,
$Z_{t}$ and $F_{t}$, that are obtained from $n+r$ univariate EGARCH
models. These take the form
\begin{align}
R_{i,t} & =a_{0,i}+a_{1,i}R_{i,t-1}+\sigma_{i,t}Z_{i,t},\quad Z_{i,t}\sim(0,1),\nonumber \\
\log\sigma_{i,t+1} & =b_{0,i}+b_{1,i}\log\sigma_{i,t}+b_{2,i}Z_{i,t}+b_{3,i}|Z_{i,t}|,\label{eq:EGARCH}
\end{align}
for $i=1,\ldots,n$, which has an AR(1) structure for the conditional
mean. The same univariate model is estimated for each of the $r$
factor return series. Estimation results for the $n+r$ EGARCH models
are reported in the Supplemental Material, see Table \ref{tab:EGARCHest}.

\subsection{Small Universe }

We begin by analyzing the Small Universe with $n=12$ stocks and $r=8$
factors. These stocks belong to four sub-industries (with three stocks
in each) within the Health Care and Information Technology sectors
(two sub-industries from each). The $r=8$ factors are given by FF5,
UMD, and the ETFs, XLV and XLK, that represent the two sector factors,
Health Care and Information Technology, respectively. 

The lower triangle of Table \ref{tab:EmpiricalCorrelations} reports
the full-sample unconditional correlation matrix and the upper triangle
reports the idiosyncratic correlations. The block structures, as defined
by sub-industries, are highlighted with shaded regions. The idiosyncratic
correlations are based on the residuals obtained by regressing $Z_{i,t}$
on $F_{t}$ and constant, $i=1,\ldots,n$.
\begin{table}
\caption{Unconditional and Idiosyncratic correlations (Small Universe)}

\begin{centering}
\begin{footnotesize} 
\begin{tabularx}{\textwidth}{p{1cm}YYYYYYYYYYYYYYY}
\toprule
\midrule
         & \multicolumn{3}{c}{Health Care } & \multicolumn{3}{c}{Managed} & \multicolumn{3}{c}{Semiconductor} & \multicolumn{3}{c}{\multirow{2}[1]{*}{Semiconductors}} \\
          & \multicolumn{3}{c}{Distributors} & \multicolumn{3}{c}{Health Care} & \multicolumn{3}{c}{Materials \& Equipment} & \multicolumn{3}{c}{} \\
\\[-0.2cm]      
          & ABC   & CAH   & MCK   & HUM   & UNH   & WLP   & AMAT  & KLAC  & LRCX  & ADI   & MCHP  & TXN \\
    \midrule
\\[-0.2cm]    
    ABC   & \cellcolor{black!8}1.000 & \cellcolor{black!8}0.472 & \cellcolor{black!8}0.520 & 0.044 & 0.058 & 0.080 & \cellcolor{black!8}0.013 & \cellcolor{black!8}0.021 & \cellcolor{black!8}0.015 & 0.033 & 0.001 & 0.012 \\
    CAH   & \cellcolor{black!8}0.640 & \cellcolor{black!8}1.000 & \cellcolor{black!8}0.497 & 0.025 & 0.034 & 0.062 & \cellcolor{black!8}0.043 & \cellcolor{black!8}0.044 & \cellcolor{black!8}0.040 & 0.038 & 0.013 & 0.006 \\
    MCK   & \cellcolor{black!8}0.669 & \cellcolor{black!8}0.655 & \cellcolor{black!8}1.000 & 0.063 & 0.068 & 0.074 & \cellcolor{black!8}0.012 & \cellcolor{black!8}0.001 & \cellcolor{black!8}0.012 & -0.01 & -0.02 & -0.01 \\
\\[-0.2cm] 
    HUM   & 0.311 & 0.300 & 0.322 & \cellcolor{black!8}1.000 & \cellcolor{black!8}0.499 & \cellcolor{black!8}0.482 & 0.011 & -0.01 & 0.005 & \cellcolor{black!8}-0.01 & \cellcolor{black!8}-0.02 & \cellcolor{black!8}-0.02 \\
    UNH   & 0.383 & 0.370 & 0.388 & \cellcolor{black!8}0.651 & \cellcolor{black!8}1.000 & \cellcolor{black!8}0.575 & 0.033 & 0.005 & 0.019 & \cellcolor{black!8}0.017 & \cellcolor{black!8}0.011 & \cellcolor{black!8}0.007 \\
    WLP   & 0.384 & 0.376 & 0.379 & \cellcolor{black!8}0.635 & \cellcolor{black!8}0.732 & \cellcolor{black!8}1.000 & 0.038 & 0.026 & 0.030 & \cellcolor{black!8}0.003 & \cellcolor{black!8}0.018 & \cellcolor{black!8}0.001 \\
\\[-0.2cm] 
    AMAT  & \cellcolor{black!8}0.260 & \cellcolor{black!8}0.301 & \cellcolor{black!8}0.264 & 0.215 & 0.285 & 0.276 & \cellcolor{black!8}1.000 & \cellcolor{black!8}0.539 & \cellcolor{black!8}0.613 & 0.351 & 0.369 & 0.329 \\
    KLAC  & \cellcolor{black!8}0.262 & \cellcolor{black!8}0.296 & \cellcolor{black!8}0.254 & 0.202 & 0.265 & 0.266 & \cellcolor{black!8}0.753 & \cellcolor{black!8}1.000 & \cellcolor{black!8}0.628 & 0.349 & 0.346 & 0.336 \\
    LRCX  & \cellcolor{black!8}0.254 & \cellcolor{black!8}0.290 & \cellcolor{black!8}0.256 & 0.205 & 0.269 & 0.264 & \cellcolor{black!8}0.790 & \cellcolor{black!8}0.794 & \cellcolor{black!8}1.000 & 0.354 & 0.358 & 0.344 \\
\\[-0.2cm] 
    ADI   & 0.291 & 0.317 & 0.271 & \cellcolor{black!8}0.221 & \cellcolor{black!8}0.297 & \cellcolor{black!8}0.277 & 0.663 & 0.653 & 0.651 & \cellcolor{black!8}1.000 & \cellcolor{black!8}0.520 & \cellcolor{black!8}0.540 \\
    MCHP  & 0.259 & 0.291 & 0.252 & \cellcolor{black!8}0.200 & \cellcolor{black!8}0.275 & \cellcolor{black!8}0.268 & 0.674 & 0.654 & 0.655 & \cellcolor{black!8}0.754 & \cellcolor{black!8}1.000 & \cellcolor{black!8}0.474 \\
    TXN   & 0.283 & 0.302 & 0.275 & \cellcolor{black!8}0.213 & \cellcolor{black!8}0.295 & \cellcolor{black!8}0.279 & 0.660 & 0.656 & 0.654 & \cellcolor{black!8}0.769 & \cellcolor{black!8}0.737 & \cellcolor{black!8}1.000 \\
\\[-0.2cm] 
\midrule
\bottomrule
\end{tabularx}
\end{footnotesize}
\par\end{centering}
{\small Note: Unconditional correlations for the 12 assets are reported
below the diagonal, whereas the (crude) idiosyncratic correlations,
based on the residuals obtained by regressing each asset return on
the eight factor returns. \label{tab:EmpiricalCorrelations}}{\small\par}
\end{table}

Unconditional correlations are quite similar within each block of
the correlation matrix. Two assets from the same subindustry are more
correlated than stocks from different subindustries, albeit subindustries
within the same sector are more correlated that subindustries from
different sectors.

Correlations are greatly reduced by controlling for the eight factors.
In fact, the idiosyncratic correlations involving stocks from different
sectors are very close to zero. The average between-sectors idiosyncratic
correlation is 0.012. The factor variables also explain much of the
within-sector correlations, but to a lesser extend. In the Health
Care sector, between-subindustries idiosyncratic correlations average
about 0.056, whereas they average about 0.348 in the Information Technology
sector. While this correlations are based on static factor loadings,
the results do suggest that correlations can be non-negligible between
subindustries in the same sector. The idiosyncratic correlations within
subindustries are even larger and range from 47.2\% to 62.8\%. Thus,
the eight factors, including the two sector factors, are clearly unable
to explains the correlation structure between stocks in the same subindustry,
at least not with static factor loadings.

The structure of the empirical idiosyncratic correlation matrix in
Table \ref{tab:EmpiricalCorrelations} is the motivation for the sparse
block correlation structures we introduced in Section \ref{subsec:Sparse-Block-Correlation}.
The results in Table \ref{tab:EmpiricalCorrelations} are based on
a simplified structure with static factor loadings, and the dynamic
factor correlation model may therefore reduce idiosyncratic correlations
further. Next, we present the empirical results for the dynamic correlations
model for the eight factor variables.

\begin{table}
\caption{Model for Dynamic Factor Variables (Small Universe)}

\begin{centering}
\begin{footnotesize}
\begin{tabularx}{\textwidth}{p{1cm}YYYYYYYYYYYYYYY}
\toprule 
\midrule
    \multicolumn{9}{c}{Panel A: Unconditional Correlation Matrix $C_F$} \\
\\[-0.2cm]     
          & MKT   & SMB   & HML   & RMW   & CMA   & UMD   & XLV   & XLK \\
    \midrule
\\[-0.2cm]    
    MKT   & 1.000 & 0.301 & 0.097 & -0.295 & -0.140 & -0.057 & 0.766 & 0.891 \\
    SMB   & 0.301 & 1.000 & 0.113 & -0.279 & 0.019 & -0.129 & 0.136 & 0.165 \\
    HML   & 0.097 & 0.113 & 1.000 & -0.064 & 0.500 & -0.307 & -0.051 & -0.096 \\
    RMW   & -0.295 & -0.279 & -0.064 & 1.000 & 0.093 & 0.014 & -0.221 & -0.228 \\
    CMA   & -0.140 & 0.019 & 0.500 & 0.093 & 1.000 & -0.050 & -0.101 & -0.268 \\
    UMD   & -0.057 & -0.129 & -0.307 & 0.014 & -0.050 & 1.000 & 0.001 & 0.032 \\
    XLV   & 0.766 & 0.136 & -0.051 & -0.221 & -0.101 & 0.001 & 1.000 & 0.645 \\
    XLK   & 0.891 & 0.165 & -0.096 & -0.228 & -0.268 & 0.032 & 0.645 & 1.000 \\
\\[-0.2cm] 
    \midrule
    \multicolumn{9}{c}{Panel B: Estimation Results on Dynamic Correlations Models of $C_F$} \\
\\[-0.2cm]   
          & $\bar{\mu}^F$ & $\bar{\beta}^F$ & $\bar{\alpha}^F$ & $\bar{\nu}^F$  & ${\nu}^F$(range) & $p$     & -$\ell (F)$   & BIC \\
    \midrule
\\[-0.2cm]    
    Gauss & 0.086 & 0.985 & 0.031 & $\infty$ &       & 84    & 36967 & 74636 \\
    MT & 0.088 & 0.986 & 0.032 & 10.30 &       & 85    & \textbf{36124} & \textbf{72959} \\
    HT & 0.086 & 0.982 & 0.034 & 8.084 & [6.259,11.43] & 92    & 36254 & 73277 \\
\\[-0.2cm]    
\midrule
\bottomrule
\end{tabularx}
\end{footnotesize}
\par\end{centering}
{\small Note: Panel A reports the sample correlation matrix for the
eight factors (FF5+}UMD{\small +XLV+XLK) and Panel B presents estimation
results for score-driven models based on thee distributional specifications.
The vector of transformed correlations, $\gamma(C_{F})$, has 28 elements
and the average estimates of their score-model parameters, $\mu^{F}$,
$\alpha^{F}$, $\beta^{F}$, and $\nu^{F}$. The HT distribution has
8 degrees of freedom parameters and we report their range. The number
of free parameters in each model is denoted by $p$, and we report
the negative value of the log-likelihood function $\ell(F)$ for each
model and the corresponding Bayesian Information Criterion (BIC),
which is defined by $\mathrm{BIC}=-2\ell+p\log T$. The largest log-likelihood
and smallest BIC values are highlighted in bold font.\label{tab:FactorEstimation}}{\small\par}
\end{table}

\subsubsection{Correlation Matrix for Factor Variables}

Table \ref{tab:FactorEstimation} reports the estimation results for
the dynamic model for the eight factor variables in $F_{t}$. The
correlation matrix in Panel A has correlations ranging from -30.7\%
to as much as 89.1\%. The sector returns are highly correlated with
market returns (and with each other). 

Panel B of Table \ref{tab:FactorEstimation} presents the estimation
results for the dynamic conditional correlation model for $F_{t}$,
i.e. $C_{F,t}$ using the score-driven model for $\gamma_{t}^{F}=\operatorname{vecl}\left(\log C_{F,t}\right)$,
by \citet{TongHansenArchakov:2024}. The maximized log-likelihood
function and the corresponding BIC show that the heavy-tailed specifications,
MT and HT, outperform the Gaussian specification, with MT (multivariate-$t$)
having the best performance. That HT is inferior to MT, despite having
more parameters, suggests that the nonlinear dependencies that a multivariate
distribution implies is important for the factor variables. Consequently,
we adopt the MT specification for $F_{t}$ and use the corresponding
$C_{F,t}$ to define the orthogonal factor variables, $U_{t}=C_{F,t}^{-1/2}F_{t}$,
that are used in the subsequent analysis.
\begin{table}
\caption{Estimation Results for Factor Loadings (CT with Full Block $C_{e}$)}

\begin{centering}
\begin{footnotesize}
\begin{tabularx}{\textwidth}{p{0.8cm}YYYYYYYYYYYYYYY}
\toprule
\midrule
         & \multicolumn{3}{c}{Health Care } & \multicolumn{3}{c}{Managed} & \multicolumn{3}{c}{Semiconductor} & \multicolumn{3}{c}{\multirow{2}[1]{*}{Semiconductors}} \\
          & \multicolumn{3}{c}{Distributors} & \multicolumn{3}{c}{Health Care} & \multicolumn{3}{c}{Materials \& Equipment} & \multicolumn{3}{c}{} \\
\\[-0.2cm]      
          & ABC   & CAH   & MCK   & HUM   & UNH   & WLP   & AMAT  & KLAC  & LRCX  & ADI   & MCHP  & TXN \\ 
    \midrule
\\[-0.2cm]    
    MKT   & \cellcolor{black!8}0.266 & \cellcolor{black!8}0.289 & \cellcolor{black!8}0.274 & 0.201 & 0.259 & 0.240 & \cellcolor{black!8}0.365 & \cellcolor{black!8}0.362 & \cellcolor{black!8}0.361 & 0.382 & 0.390 & 0.393 \\
    SMB   & \cellcolor{black!8}0.067 & \cellcolor{black!8}0.076 & \cellcolor{black!8}0.048 & 0.038 & 0.008 & 0.015 & \cellcolor{black!8}0.114 & \cellcolor{black!8}0.120 & \cellcolor{black!8}0.122 & 0.102 & 0.135 & 0.087 \\
    HML   & \cellcolor{black!8}0.005 & \cellcolor{black!8}0.044 & \cellcolor{black!8}0.028 & 0.038 & 0.043 & 0.070 & \cellcolor{black!8}-0.01 & \cellcolor{black!8}-0.02 & \cellcolor{black!8}-0.01 & -0.01 & -0.01 & -0.02 \\
    RMW   & \cellcolor{black!8}0.007 & \cellcolor{black!8}-0.02 & \cellcolor{black!8}-0.01 & 0.000 & 0.013 & -0.01 & \cellcolor{black!8}-0.08 & \cellcolor{black!8}-0.07 & \cellcolor{black!8}-0.08 & -0.09 & -0.08 & -0.06 \\
    CMA   & \cellcolor{black!8}0.010 & \cellcolor{black!8}0.008 & \cellcolor{black!8}-0.01 & -0.03 & -0.03 & 0.005 & \cellcolor{black!8}-0.06 & \cellcolor{black!8}-0.05 & \cellcolor{black!8}-0.05 & -0.04 & -0.03 & -0.01 \\
    UMD   & \cellcolor{black!8}0.013 & \cellcolor{black!8}-0.03 & \cellcolor{black!8}0.013 & 0.001 & 0.030 & 0.001 & \cellcolor{black!8}-0.02 & \cellcolor{black!8}0.001 & \cellcolor{black!8}-0.01 & -0.02 & -0.01 & -0.01 \\
    XLV   & \cellcolor{black!8}0.426 & \cellcolor{black!8}0.429 & \cellcolor{black!8}0.416 & 0.400 & 0.482 & 0.452 & \cellcolor{black!8}0.189 & \cellcolor{black!8}0.192 & \cellcolor{black!8}0.184 & 0.228 & 0.202 & 0.224 \\
    XLK   & \cellcolor{black!8}0.144 & \cellcolor{black!8}0.176 & \cellcolor{black!8}0.160 & 0.120 & 0.156 & 0.143 & \cellcolor{black!8}0.459 & \cellcolor{black!8}0.455 & \cellcolor{black!8}0.437 & 0.465 & 0.468 & 0.487 \\
\\[-0.2cm] 
\midrule
\bottomrule
\end{tabularx}
\end{footnotesize}
\par\end{centering}
{\small Note: The averaged factor loadings for the model with Cluster-$t$
distribution (CT) under Full Block $C_{e}$ structure from Table \ref{tab:EstSmall}.
\label{tab:Factorloadings}}{\small\par}
\end{table}

\subsubsection{Factor Loadings and Idiosyncratic Correlations (joint estimation)}

Next we turn to the central component of the model, which is the dynamic
model of factor loadings and idiosyncratic correlations. We can estimated
this components of the model in two ways: Either jointly (simultaneously)
or decoupled. In this section we use joint estimation. These estimation
results are presented in Table \ref{tab:EstSmall}, where we report
the average values of $\mu$, $\alpha$, $\beta$, $\nu$, $\log\lambda$,
as well as the individual estimates of $\nu_{i}$ and $\log\lambda_{i}$.
Several interesting observations emerge from Table \ref{tab:EstSmall}.
\begin{enumerate}
\item The heavy-tailed distributions offer substantial improvements over
the Gaussian specification, as was the case in the model for $F_{t}.$
But, here we find that the convolution-$t$ distributions outperform
the multivariate-$t$ distribution with $\ell(Z|U)$ increases as
much as 1,400 units. 
\item The estimates of $\alpha$ (the updating-parameter in the score model)
are larger for the heavy-tailed specifications than the Gaussian specification.
We attribute this to the $W$-variables that mitigates the impact
of extreme values (outliers). This is a well-known feature of score-driven
models, see e.g. \citet{Harvey2013}. 
\item The CT specification has a cluster structure with fewer degrees of
freedom and more nonlinear dependencies that the HT specification.
While the the estimated degrees of freedom in all CT specifications
are quite similar to the average of the corresponding coefficients
in the HT specifications, the largest log-likelihood (and best BIC)
is achieved by the CT specification with the sparse block correlation
structure. The additional degrees of freedom in the HT distribution
cannot make up for the lack of nonlinear dependencies that the CT
distribution can accommodate. 
\item A comparison of specifications with different structures for the idiosyncratic
correlation matrix, reveals that the diagonal structure (DBC) is too
restrictive. This was also suggested by the preliminary empirical
results in Table \ref{tab:EmpiricalCorrelations}, which were based
on static factor loadings.
\item The penalty parameters, $\lambda_{1},\ldots,\lambda_{n}$ are all
estimated to be large, which demonstrates a need to regularize the
scaling matrix in the score model. The corresponding shocks are denoted
$\varepsilon_{\lambda}.$ The bottom of Table \ref{tab:EstSmall}
include comparisons with the score models use the Moore-Penrose inverse
as scaling matrix (with shocks denoted $\varepsilon_{0}$ because
$\lambda=0$) and the simplistic $S=I$ scaling matrix, with shocks
denoted $\varepsilon^{\ast}$. Using regularized model has the largest
log-likelihood, $\ell\left(Z|U\right)$, in all cases. On average
the regularized score driven model has a log-likelihood that is about
349 unit larger than the unregularized variant, $\varepsilon_{0}$,
and 375 units larger than the simplistic unscaled score model.
\end{enumerate}
\begin{sidewaystable}
\caption{Factor Loadings and Idiosyncratic Correlations (joint estimation)}

\begin{centering}
\begin{scriptsize} 
\begin{tabularx}{\textwidth}{p{1.5cm}YYYYp{-0.1cm}YYYYp{-0.4cm}YYYYp{-0.1cm}YYYYp{-0.1cm}YYYYYYYY}
\toprule
\midrule
          & \multicolumn{4}{c}{Unrestricted $C_e$}   &       & \multicolumn{4}{c}{Full Block $C_e$}  &       & \multicolumn{4}{c}{Sparse Block $C_e$} &       & \multicolumn{4}{c}{Diagonal Block $C_e$} \\
\cmidrule{2-5}\cmidrule{7-10}\cmidrule{12-15}\cmidrule{17-20}          & {Gauss} & MT    & CT & HT &       & {Gauss} & MT   & CT & HT &       & {Gauss} & MT    & CT & HT &       & {Gauss} & MT     & CT & HT \\
\midrule
\\[-0.2cm]  
    $\bar{\mu}$ & 0.128 & 0.131 & 0.129 & 0.130 &       & 0.146 & 0.145 & 0.146 & 0.146 &       & 0.150 & 0.150 & 0.151 & 0.150 &       & 0.154 & 0.151 & 0.152 & 0.151 \\
\\[-0.3cm]  
    $\bar{\beta}$ & 0.949 & 0.942 & 0.951 & 0.950 &       & 0.924 & 0.963 & 0.937 & 0.959 &       & 0.948 & 0.936 & 0.938 & 0.952 &       & 0.946 & 0.953 & 0.939 & 0.961 \\
\\[-0.3cm] 
    $\bar{\alpha}$ & 0.010 & 0.016 & 0.016 & 0.018 &       & 0.014 & 0.020 & 0.028 & 0.027 &       & 0.015 & 0.025 & 0.029 & 0.028 &       & 0.016 & 0.024 & 0.026 & 0.027 \\
\\[-0.2cm]
    $\bar{\nu}$    & $\infty$ & 4.874 & 3.764 & 3.449 &       & $\infty$ & 4.841 & 3.751 & 3.437 &       & $\infty$ & 4.850 & 3.755 & 3.441 &       & $\infty$ & 4.835 & 3.815 & 3.480 \\
\\[-0.2cm]
    ${\nu}_1$    &       &       &       & 3.675 &       &       &       &       & 3.624 &       &       &       &       & 3.632 &       &       &       &       & 3.661 \\
    ${\nu}_2$    &       &       & 3.640 & 3.194 &       &       &       & 3.633 & 3.231 &       &       &       & 3.641 & 3.237 &       &       &       & 3.666 & 3.251 \\
    ${\nu}_3$    &       &       &       & 3.312 &       &       &       &       & 3.279 &       &       &       &       & 3.281 &       &       &       &       & 3.296 \\
    ${\nu}_4$    &       &       &       & 2.998 &       &       &       &       & 3.079 &       &       &       &       & 3.079 &       &       &       &       & 3.084 \\
    ${\nu}_5$    &       &       & 3.710 & 3.867 &       &       &       & 3.719 & 3.673 &       &       &       & 3.722 & 3.667 &       &       &       & 3.729 & 3.672 \\
    ${\nu}_6$    &       &       &       & 3.570 &       &       &       &       & 3.549 &       &       &       &       & 3.544 &       &       &       &       & 3.548 \\
    ${\nu}_7$    &       &       &       & 3.587 &       &       &       &       & 3.643 &       &       &       &       & 3.660 &       &       &       &       & 3.732 \\
    ${\nu}_8$    &       &       & 3.782 & 3.113 &       &       &       & 3.733 & 3.289 &       &       &       & 3.735 & 3.304 &       &       &       & 3.801 & 3.361 \\
    ${\nu}_9$    &       &       &       & 3.707 &       &       &       &       & 3.510 &       &       &       &       & 3.513 &       &       &       &       & 3.488 \\
    ${\nu}_{10}$   &       &       &       & 3.411 &       &       &       &       & 3.304 &       &       &       &       & 3.305 &       &       &       &       & 3.369 \\
    ${\nu}_{11}$   &       &       & 3.925 & 3.240 &       &       &       & 3.918 & 3.295 &       &       &       & 3.923 & 3.298 &       &       &       & 4.064 & 3.466 \\
    ${\nu}_{12}$   &       &       &       & 3.715 &       &       &       &       & 3.768 &       &       &       &       & 3.771 &       &       &       &       & 3.833 \\
          &       &       &       &       &       &       &       &       &       &       &       &       &       &       &       &       &       &       &  \\
    $\bar{\log\lambda}$ & 5.744 & 3.121 & 2.937 & 3.035 &       & 5.639 & 3.555 & 3.431 & 3.344 &       & 5.913 & 3.685 & 3.234 & 4.501 &       & 6.228 & 3.827 & 4.576 & 3.974 \\
\\[-0.2cm]
    $\log \lambda_1$ & 3.469 & 3.652 & 4.716 & 5.398 &       & 3.442 & 4.237 & 6.259 & 7.017 &       & 3.384 & 4.112 & 5.356 & 13.11 &       & 3.511 & 7.229 & 4.927 & 10.63 \\
    $\log \lambda_2$ & 9.574 & 4.165 & 3.753 & 6.445 &       & 10.37 & 8.421 & 9.987 & 7.854 &       & 11.23 & 9.323 & 7.170 & 11.49 &       & 13.48 & 8.017 & 10.85 & 8.133 \\
    $\log \lambda_3$ & 8.384 & 3.149 & 3.433 & 3.523 &       & 9.023 & 3.483 & 3.855 & 3.720 &       & 9.803 & 3.454 & 3.921 & 3.713 &       & 12.34 & 3.634 & 4.102 & 3.887 \\
    $\log \lambda_4$ & 9.817 & 6.419 & 8.099 & 9.049 &       & 10.27 & 9.782 & 8.841 & 9.371 &       & 11.17 & 9.896 & 9.745 & 13.26 &       & 10.66 & 8.556 & 19.02 & 10.06 \\
    $\log \lambda_5$ & 2.578 & 1.723 & 1.390 & 0.950 &       & 2.782 & 1.824 & 1.379 & 1.134 &       & 2.683 & 1.820 & 1.381 & 1.132 &       & 2.692 & 1.776 & 1.408 & 1.254 \\
    $\log \lambda_6$ & 4.958 & 4.482 & 3.369 & 3.223 &       & 5.157 & 4.220 & 3.503 & 3.400 &       & 5.133 & 4.172 & 3.561 & 3.444 &       & 4.897 & 4.274 & 3.627 & 3.682 \\
    $\log \lambda_7$ & 6.138 & 3.391 & 2.402 & 1.949 &       & 3.711 & 2.598 & 1.862 & 1.906 &       & 3.686 & 2.737 & 1.947 & 1.958 &       & 3.495 & 2.624 & 2.276 & 2.009 \\
    $\log \lambda_8$ & 9.702 & 3.134 & 2.666 & 0.660 &       & 10.62 & 1.568 & 0.636 & 0.375 &       & 11.50 & 2.000 & 0.727 & 0.418 &       & 12.53 & 1.815 & 1.009 & 0.735 \\
    $\log \lambda_9$ & 3.759 & 1.741 & 1.409 & 0.791 &       & 1.915 & 1.396 & 0.994 & 0.858 &       & 1.931 & 1.551 & 1.120 & 0.918 &       & 2.440 & 1.728 & 2.281 & 1.289 \\
    $\log \lambda_{10}$ & 2.558 & 1.056 & 0.868 & 1.429 &       & 1.882 & 1.019 & 0.820 & 1.572 &       & 1.875 & 1.080 & 0.837 & 1.610 &       & 2.055 & 1.567 & 1.672 & 2.476 \\
    $\log \lambda_{11}$ & 4.375 & 2.748 & 2.406 & 2.120 &       & 4.130 & 3.042 & 2.501 & 2.180 &       & 4.127 & 2.907 & 2.495 & 2.188 &       & 3.736 & 3.148 & 2.818 & 2.540 \\
    $\log \lambda_{12}$ & 3.614 & 1.798 & 0.732 & 0.879 &       & 4.355 & 1.072 & 0.535 & 0.741 &       & 4.423 & 1.165 & 0.544 & 0.759 &       & 2.882 & 1.554 & 0.908 & 0.987 \\
          &       &       &       &       &       &       &       &       &       &       &       &       &       &       &       &       &       &       &  \\
    $p$     & 498   & 499   & 502   & 510   &       & 330   & 331   & 334   & 342   &       & 318   & 319   & 322   & 330   &       & 312   & 313   & 316   & 324 \\
\\[-0.3cm]
    -$\ell (Z|U)$   & 50987 & 45891 & 44498 & \textbf{44455} &       & 51155 & 46059 & \textbf{44639} & 44657 &       & 51166 & 46073 & \textbf{44659} & 44671 &       & 51860 & 46827 & \textbf{45459} & 45547 \\
\\[-0.3cm]
    BIC   & 106138 & 95954 & 93193 & \textbf{93174} &       & 105069 & 94886 & \textbf{92071} & 92174 &       & 104991 & 94813 & \textbf{92010} & 92101 &       & 106329 & 96271 & \textbf{93560} & 93803 \\
\\[0cm]  
    \multicolumn{20}{l}{Differences in $\ell (Z|U)$ relative to $\varepsilon_{\lambda}=\Pi_{\lambda}^{+}\mathcal{I}_{\xi}^{-1}\nabla_{\xi}$ } \\
\\[-0.2cm]     
    $\varepsilon^*=\nabla_{\zeta}$ & -280   & -330   & -356   & -367   &       & -294   & -331   & -358   & -359   &       & -304   & -341   & -371   & -377   &       & -447   & -435   & -515   & -541 \\
    $\varepsilon_0=\mathcal{I}_{\zeta}^{+}\nabla_{\zeta}$  & -371   & -286   & -293   & -302   &       & -340   & -287   & -328   & -298   &       & -376   & -292   & -325   & -309   &       & -494   & -370   & -437   & -468 \\
\\[-0.2cm]    
\midrule
\bottomrule
\end{tabularx}
\end{scriptsize}
\par\end{centering}
{\small Note: Joint estimation results for the dynamic models of factor
loadings and the idiosyncratic correlation matrix. We report the average
values of $\mu$, $\alpha$, $\beta$, $\nu$, $\log\lambda$, as
well as the individual estimates of $\nu_{i}$ and $\log\lambda_{i}$.
The number of model parameters is denoted by $p$, and we also report
the negative log-likelihood function, $-\ell(Z|U)$, and the Bayesian
Information Criterion (BIC). At the bottom, we compare the differences
in $\ell(Z|U)$ with two alternative specifications for the scaled
score, relative to the regularized version proposed in this paper.
The largest log-likelihood and smallest BIC values for each structure
of $C_{e}$ are highlighted in bold font.\label{tab:EstSmall}}{\small\par}
\end{sidewaystable}

\begin{figure}[t]
\centering
\centering{}\includegraphics[width=1\textwidth]{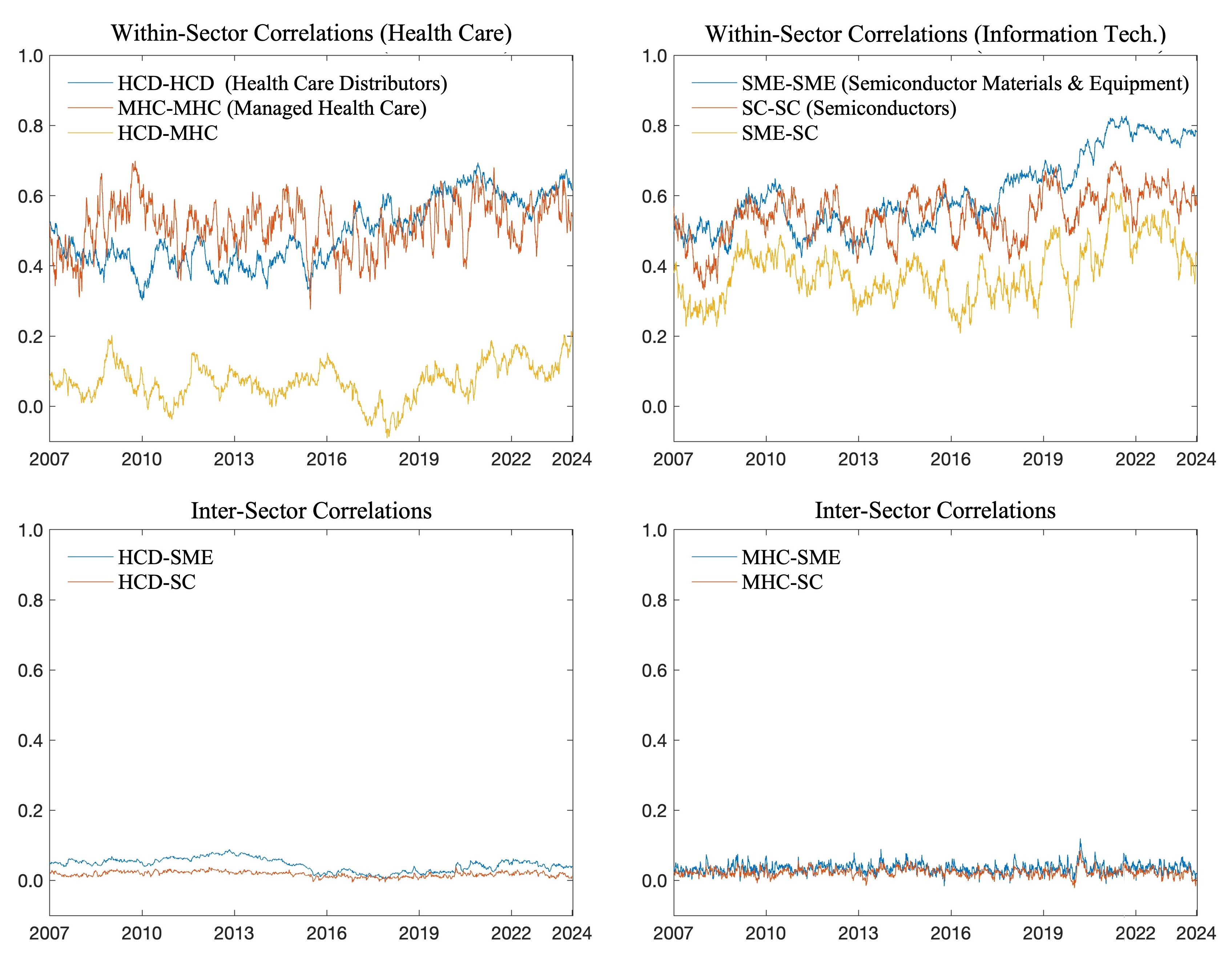}\caption{{\small Time series of idiosyncratic block correlations among different
sub-industries in small universe. The idiosyncratic correlations are
based on the CT distribution with Full Block $C_{e}$ in Table \ref{tab:EstSmall}.\label{fig:BlockCorrSmall}}}
\end{figure}

Figure \ref{fig:BlockCorrSmall} illustrates the time series of block
correlations for the specification estimated with Cluster-$t$ distribution
(CT) under Full Block $C_{e}$. The upper panels reveal sizable correlations
for stocks in the same sector, especially for stocks in the same subindustry.
The average correlation between the two Health Care subindustries,
HCD and MHC, is just 0.056, but there is substantial time variation
in this correlation that exceeds 0.20 on several occasions. This underscores
the importance of permitting non-trivial idiosyncratic correlations
to be time-varying. The lower panels in Figure \ref{fig:BlockCorrSmall}
present the estimated paths for correlations between stocks in different
sectors. These are also close to zero and remain fairly constant over
time. These observations support a sparse block correlation structure
for $C_{e}$, which imposes idiosyncratic shocks from different sectors
to be be zero. This correlation structure, combined with the CT distributions
has the smallest BIC.
\begin{figure}[t]
\centering
\centering{}\includegraphics[width=1\textwidth]{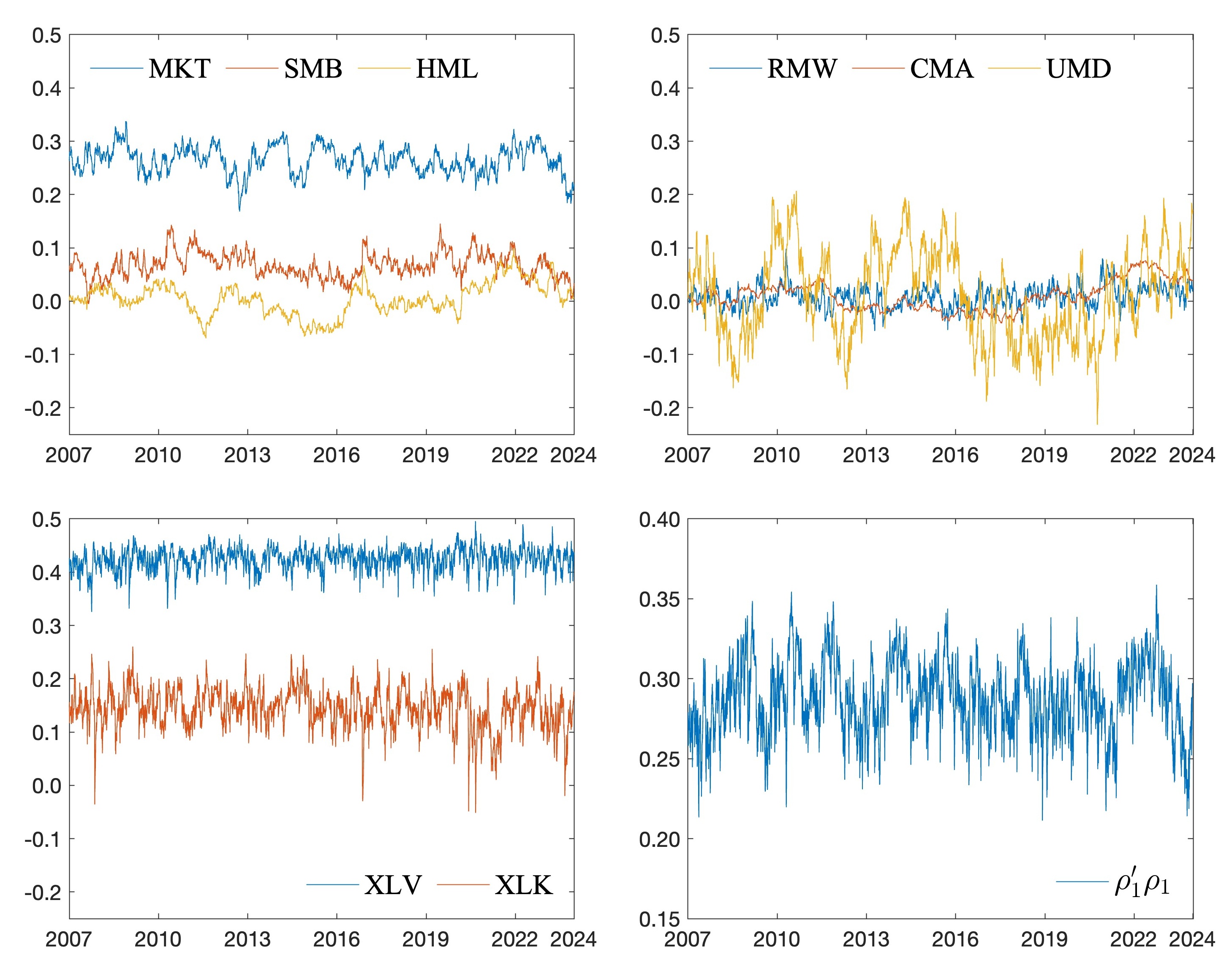}\caption{{\small This figure displays the time series of Factor Loadings of
Company ABC on eight factors (FF5+}UMD{\small +XLV+XLK). The final
subfigure plots the dynamics of the explanatory power of the factors,
measured by $\rho_{1}^{\prime}\rho_{1}$. The factor loadings are
derived from the joint estimation results of the model with Full Block
$C_{e}$ under CT distribution in Table \ref{tab:EstSmall}.\label{fig:FactorLoadingsABC}}}
\end{figure}

Table \ref{tab:Factorloadings} presents the average factor loadings
as estimated with the CT specification a block correlation structure
for $C_{e}$. These are the average values of $\rho_{i,t}$ which
related to $U_{t}$. The lower panel in Table \ref{tab:Factorloadings}
presents the corresponding average loadings on the standardized factor
variables, $F_{t}$. The factor loadings for assets in the same sector
tend to be quite similar, which is also reflected in the average factor
loadings. The most significant factors are, not surprisingly, the
market factor (MKT) and the two sector-factors (XLV and XLK). The
size factor (SMB) is also important for the IT sector, but has small
average impact on stocks in the Health Care sector. 

Figure \ref{fig:FactorLoadingsABC} displays the time series of factor
loadings for Cencora Inc. (which has Ticker ABC). Its factor loadings,
$\rho_{1,t}$, are estimated with the CT specification and the non-sparse
block structure for $C_{e}$. The upper panels are the factor loadings
for the six cross-sectional factors (FF5+UMD) and the lower left panel
presents the factor loadings for the two sector-factors. The lower-right
panel displays the fraction of the variance that is explained by factors,
as defined by $\rho_{1,t}^{\prime}\rho_{1,t}$. The factor loadings
are time-varying and exhibit significant heterogeneity across different
factors. From Table \ref{tab:Factorloadings}, although the average
factor loading for HML is only 0.005, it dynamically ranges from -0.069
to 0.099 over time. Similarly, the factor loading on UMD (the momentum
factor) displays the largest variation, despite having a mean of just
0.013. This stock belongs to the Health Care sector, and we also observe
that its factor loading on XLV (Health care sector) is large and relatively
stable over time. Additionally, there is a notable correlation with
the Information Technology ETF (XLK), reflecting inter-sector dependencies.
The last subfigure reveals that the averaged explanatory power of
the factors is 0.288 and is also dynamic over time.

\subsubsection{Decoupled Estimation}

In this section, we estimate core component of our model with the
decoupled estimation method as described in Section \ref{subsec:DynamicLoadings}.
The approach to estimation can be scaled to large dimensions, because
the dynamic models for factor loadings are estimated separately for
each asset and separately from the dynamic model for $C_{e}$. The
dynamic models for the vector of factor loadings, $\rho_{i}$, is
estimated by maximizing $\ell^{\ast}\left(Z_{i}|U\right)$, $i=1,\ldots,n$,
from which we obtain the residuals, $e_{t}\in\mathbb{R}^{n}$, $t=1,\ldots,T$.
These are subsequently used to estimated the dynamic model for the
idiosyncratic correlation matrix, $C_{e}$, which is done by maximizing
$\ell\left(e\right)$.

The results for the $n$ factor loading models are reported in Table
\ref{tab:SeqFirstSmall}. The factor loading models are estimated
by an approximating $t$-distribution, and the estimated degrees of
freedom, $\nu_{i}^{\star}$, is small for all twelve assets, indicating
that they all exhibit heavy tails. 

While decoupled estimation is less efficient than joint estimation,
it does yield quite similar paths for the estimated factor loadings.
The empirical correlation between the two sets of factor loadings,
denoted ${\rm corr}(\rho)$, is reported in the last row of Table
\ref{tab:SeqFirstSmall}. The specification used for the joint estimation
is that based on the HT distribution and an unrestricted $C_{e}$.
\begin{table}[t]
\caption{Factor Loadings (decoupled estimation)}

\begin{centering}
\begin{footnotesize}
\begin{tabularx}{\textwidth}{p{1.3cm}YYYYYYYYYYYYYYY}
\toprule
\midrule
         & \multicolumn{3}{c}{Health Care } & \multicolumn{3}{c}{Managed} & \multicolumn{3}{c}{Semiconductor} & \multicolumn{3}{c}{\multirow{2}[1]{*}{Semiconductors}} \\
          & \multicolumn{3}{c}{Distributors} & \multicolumn{3}{c}{Health Care} & \multicolumn{3}{c}{Materials \& Equip.} & \multicolumn{3}{c}{} \\
\\[-0.2cm]      
          & ABC   & CAH   & MCK   & HUM   & UNH   & WLP   & AMAT  & KLAC  & LRCX  & ADI   & MCHP  & TXN \\
    \midrule 
\\[-0.2cm]    
    $\bar{\mu}$    & 0.140 & 0.130 & 0.124 & 0.096 & 0.137 & 0.125 & 0.136 & 0.134 & 0.131 & 0.149 & 0.157 & 0.158 \\
    $\bar{\beta}$  & 0.979 & 0.976 & 0.990 & 0.953 & 0.986 & 0.966 & 0.974 & 0.960 & 0.979 & 0.955 & 0.978 & 0.985 \\
    $\bar{\alpha}$   & 0.019 & 0.018 & 0.012 & 0.016 & 0.028 & 0.027 & 0.028 & 0.068 & 0.033 & 0.045 & 0.033 & 0.052 \\
    $\nu^{\star}$     & 3.702 & 3.332 & 3.368 & 3.145 & 3.728 & 3.766 & 3.978 & 3.468 & 3.832 & 3.644 & 3.623 & 4.001 \\
\\[-0.2cm]
$\log \lambda$ & 3.339 & 5.454 & 5.075 & 5.900 & 2.053 & 2.974 & 2.976 & 1.334 & 2.616 & 2.600 & 2.981 & 1.481 \\
\\[-0.2cm]
    $p$     & 26    & 26    & 26    & 26    & 26    & 26    & 26    & 26    & 26    & 26    & 26    & 26 \\
    -$\ell (Z_i|U)$   & 4784  & 4505  & 4683  & 4720  & 4637  & 4707  & 4294  & 4257  & 4439  & 4178  & 4064  & 4071 \\
\\[-0.2cm]
    $\text{corr}  (\rho)$ & 0.989 & 0.992 & 0.992 & 0.992 & 0.993 & 0.987 & 0.989 & 0.988 & 0.983 & 0.989 & 0.991 & 0.990 \\
\\[-0.2cm] 
\midrule
\bottomrule
\end{tabularx}
\end{footnotesize}
\par\end{centering}
{\small Note: Estimation results for the $n$ factor loading models
estimated with decoupled estimation, by maximizing $\ell\left(Z_{i}|U\right)$
for $i=1,\ldots,12$, using the best-approximating $t$-distribution.
We report the average values of $\mu$, $\alpha$, $\beta$, as well
as the individual estimates of $\nu_{i}^{*}$ and $\log\lambda_{i}$.
The number of model parameters is represented by $p$, and we report
the negative log-likelihood function, $\ell\left(Z_{i}|U\right)$.
The correlation between the factor loadings estimated with decoupled
estimation and those obtained from joint estimation (using an unrestricted
$C_{e}$ with the HT distribution) are reported in the last row.\label{tab:SeqFirstSmall}}{\small\par}
\end{table}

Table \ref{tab:SeqSecondSmall} presents the estimation results for
the dynamic models of $C_{e}$, estimated with decoupled estimation
by maximizing $\ell(e)$. We have estimated this model with four types
of distributions and four different correlation structures for $C_{e}$.
Each of these models have the form of a Cluster-GARCH model, see \citet{TongHansenArchakov:2024},
and can be estimated as such. For some of the specifications with
sparse $C_{e}$, it is possible to estimated lower-dimensional models
for each sector/subindustry separated as discussed in Section \ref{subsec:Second-Step}. 

For instance with the Sparse-Block $C_{e}$ we can estimate separate
Cluster-GARCH models for the two sectors, with the exception of the
case where the multivariate $t$-distribution is used (because it
implies non-linear dependencies). Similarly with a diagonal block
$C_{e}$, some of the specification can be simplified to (four) dynamic
models with the structure detailed in Theorem \ref{thm:TquiCorr}
and Theorem \ref{thm:HTquiCorr}. The estimated degrees or freedom
are very similar to those we obtained with joint estimation. 

The main difference between joint and decoupled estimation can be
seed from the conditional log-likelihood of $Z$ given $F$, which
is given by $\ell(Z|U)=-\log|\Lambda_{\rho}|+\ell(e)$. While joint
estimation maximizes this quantity, decoupled estimation only does
so indirectly. This can be seen from the values of $\ell(Z|U)$ which
are presented at the bottom of Table \ref{tab:SeqSecondSmall}. This
log-likelihood is about 200 units smaller with decoupled estimation,
which is to be expected because there are 300 fewer parameters being
estimated when decoupled estimation is used. The last row of Table
\ref{tab:SeqSecondSmall} reports the correlation between the estimated
elements of $C=\boldsymbol{\rho}^{\prime}\boldsymbol{\rho}+\Lambda_{\omega}C_{e}\Lambda_{\omega}$,
with joint estimation versus decouple estimation, and once again we
do find very high correlations between the two estimation methods.

\begin{sidewaystable}
\caption{Decoupled Estimation for Dynamic Idiosyncratic Correlation Matrix}

\begin{centering}
\begin{scriptsize}
\begin{tabularx}{\textwidth}{p{1.6cm}YYYYp{-0.1cm}YYYYp{-0.1cm}YYYYp{-0.1cm}YYYYp{-0.1cm}YYYYYYYY}
\toprule 
\midrule
          & \multicolumn{4}{c}{Unrestricted $C_e$}   &       & \multicolumn{4}{c}{Full Block $C_e$}  &       & \multicolumn{4}{c}{Sparse Block $C_e$} &       & \multicolumn{4}{c}{Diagonal Block $C_e$} \\
\cmidrule{2-5}\cmidrule{7-10}\cmidrule{12-15}\cmidrule{17-20}          & {Gauss} & MT    & CT & HT &       & {Gauss} & MT   & CT & HT &       & {Gauss} & MT    & CT & HT &       & {Gauss} & MT     & CT & HT \\
\midrule
\\[-0.2cm]  
    $\bar{\mu}$    & 0.112 & 0.118 & 0.114 & 0.110 &       & 0.174 & 0.207 & 0.210 & 0.476 &       & 0.285 & 0.331 & 0.345 & 0.476 &       & 0.374 & 0.461 & 0.476 & 0.476 \\
          &       &       &       &       &       &       &       &       &       &       &       &       &       &       &       &       &       &       &  \\
    $\bar{\beta}$  & 0.900 & 0.911 & 0.920 & 0.921 &       & 0.850 & 0.947 & 0.991 & 0.992 &       & 0.955 & 0.989 & 0.989 & 0.990 &       & 0.954 & 0.989 & 0.988 & 0.990 \\
          &       &       &       &       &       &       &       &       &       &       &       &       &       &       &       &       &       &       &  \\
    $\bar{\alpha}$ & 0.006 & 0.007 & 0.008 & 0.009 &       & 0.019 & 0.013 & 0.011 & 0.011 &       & 0.031 & 0.017 & 0.018 & 0.018 &       & 0.035 & 0.021 & 0.022 & 0.021 \\
          &       &       &       &       &       &       &       &       &       &       &       &       &       &       &       &       &       &       &  \\
    $\bar{\nu}$    & $\infty$  & 4.994 & 3.867 & 3.541 &       & $\infty$  & 4.972 & 3.862 & 3.535 &       & $\infty$  & 4.974 & 3.865 & 3.537 &       & $\infty$  & 4.899 & 3.887 & 3.549 \\
          &       &       &       &       &       &       &       &       &       &       &       &       &       &       &       &       &       &       &  \\
    $\nu_1$    &       &       &       & 3.717 &       &       &       &       & 3.697 &       &       &       &       & 3.707 &       &       &       &       & 3.716 \\
    $\nu_2$    &       &       & 3.672 & 3.192 &       &       &       & 3.655 & 3.230 &       &       &       & 3.662 & 3.235 &       &       &       & 3.672 & 3.236 \\
    $\nu_3$    &       &       &       & 3.347 &       &       &       &       & 3.311 &       &       &       &       & 3.312 &       &       &       &       & 3.317 \\
    $\nu_4$    &       &       &       & 3.041 &       &       &       &       & 3.117 &       &       &       &       & 3.116 &       &       &       &       & 3.119 \\
    $\nu_5$    &       &       & 3.813 & 3.865 &       &       &       & 3.831 & 3.716 &       &       &       & 3.832 & 3.717 &       &       &       & 3.826 & 3.710 \\
    $\nu_6$    &       &       &       & 3.671 &       &       &       &       & 3.651 &       &       &       &       & 3.654 &       &       &       &       & 3.644 \\
    $\nu_7$    &       &       &       & 3.823 &       &       &       &       & 3.896 &       &       &       &       & 3.900 &       &       &       &       & 3.954 \\
    $\nu_8$    &       &       & 3.926 & 3.317 &       &       &       & 3.904 & 3.406 &       &       &       & 3.905 & 3.407 &       &       &       & 3.938 & 3.429 \\
    $\nu_9$    &       &       &       & 3.782 &       &       &       &       & 3.661 &       &       &       &       & 3.652 &       &       &       &       & 3.622 \\
    $\nu_{10}$   &       &       &       & 3.520 &       &       &       &       & 3.426 &       &       &       &       & 3.424 &       &       &       &       & 3.439 \\
    $\nu_{11}$   &       &       & 4.055 & 3.361 &       &       &       & 4.060 & 3.404 &       &       &       & 4.061 & 3.409 &       &       &       & 4.112 & 3.488 \\
    $\nu_{12}$   &       &       &       & 3.855 &       &       &       &       & 3.906 &       &       &       &       & 3.908 &       &       &       &       & 3.916 \\
          &       &       &       &       &       &       &       &       &       &       &       &       &       &       &       &       &       &       &  \\
    $p$     & 198   & 199   & 202   & 210   &       & 30    & 31    & 34    & 42    &       & 18    & 19    & 22    & 30    &       & 12    & 13    & 16    & 24 \\
\\[-0.3cm]
    -$\ell(e)$   & 64586 & 58771 & 57367 & \textbf{57289} &       & 64701 & 58955 & 57538 & \textbf{57508} &       & 64709 & 58968 & 57551 & \textbf{57520} &       & 65383 & 59656 & 58336 & \textbf{58334} \\
\\[-0.3cm]
    BIC   & 130828 & 119206 & 116423 & \textbf{116334} &       & 129653 & 118169 & \textbf{115360} & 115367 &       & 129569 & 118095 & \textbf{115286} & 115291 &       & 130866 & 119421 & \textbf{116806} & 116869 \\
          &       &       &       &       &       &       &       &       &       &       &       &       &       &       &       &       &       &       &  \\
 -$\ell(Z|U)_{\text{Seq}}$ & 51113 & 46154 & 44750 & 44672 &       & 51228 & 46338 & 44921 & 44891 &       & 51236 & 46351 & 44934 & 44903 &       & 51910 & 47039 & 45719 & 45717 \\
\\[-0.2cm]
 -$\ell(Z|U)_{\text{Joint}}$ & 50987 & 45891 & 44498 & 44455 &       & 51155 & 46059 & 44639 & 44657 &       & 51166 & 46073 & 44659 & 44671 &       & 51860 & 46827 & 45459 & 45547 \\
\\[-0.2cm]  
    $\text{corr} (C)$ & 0.991 & 0.992 & 0.994 & 0.994 &       & 0.986 & 0.991 & 0.994 & 0.994 &       & 0.987 & 0.992 & 0.994 & 0.995 &       & 0.988 & 0.990 & 0.993 & 0.994 \\
\midrule
\bottomrule
\end{tabularx}
\end{scriptsize}
\par\end{centering}
{\small Note: This table presents the second-step estimation results
for the dynamic idiosyncratic correlation matrix $C_{e}$, obtained
by maximizing $\ell(e)$, where $e$ is filtered from the first-step
estimation reported in Table \ref{tab:SeqFirstSmall}. We report the
average values of$\mu$, $\alpha$, $\beta$, and $\nu$, as well
as the individual estimates of $\nu_{i}$. Additionally, we include
the number of model parameters $p$, the negative log-likelihood function
$\ell(e)$, and the Bayesian Information Criterion (BIC). The largest
log-likelihood and smallest BIC values for each structure of $C_{e}$
are highlighted in bold font. At the bottom of the table, we compute
the log-likelihood values for returns, $\ell(Z|U)=-\log\left|\Lambda_{\omega}\right|+\ell(e)$,
under joint and decoupled estimations, respectively, where $\Lambda_{\omega}$
in the decoupled estimation is taken from the first-step estimation
in Table \ref{tab:SeqFirstSmall}. Furthermore, we provide the correlation
between the estimated elements of $C=\boldsymbol{\rho}^{\prime}\boldsymbol{\rho}+\Lambda_{\omega}C_{e}\Lambda_{\omega}$,
with joint estimation versus decouple estimation. \label{tab:SeqSecondSmall}}{\small\par}
\end{sidewaystable}

\subsection{Out-of-Sample Results (Small Universe)}

We next compare the estimated model specifications in terms of their
out-of-sample (OOS) performance. We take the $Z_{t}$ and $U_{t}$
variables obtained from the full sample, and focus on the performance
of the core model. We estimate each of the specification for the core
correlation model with joint estimation and decoupled estimation,
using 12 years of data (2007-2018). The estimated models are then
compared out-of-sample using data from the years 2019 to 2023. 

We report the out-of-sample results in Table \ref{tab:OOSsmall}.
The out-of-sample log-likelihood function, $\ell(Z|U)$, is reported
for 32 model specifications, based on 4 different distributions, 4
different structures for the idiosyncratic correlation matrix, $C_{e}$,
and the two estimation methods (joint and decoupled). 

There are several interesting observations to made from this out-of-sample
comparison. First, the convolution-$t$ distribution with a cluster
structure, CT, is always the best distributional specification, regardless
of the correlation structure and estimation method. This differs from
the in-sample results where the HT specification (with an unrestricted
$C_{e}$) had the largest log-likelihood. Second, the Sparse Block
structure for $C_{e}$ tend to have the best out-of-sample performance
followed by the (non-sparse) block correlation structure. The most
restrictive structure, diagonal block correlation, has the worst out-of-sample
performance, which is consistent with the in-sample results. Third,
decoupled estimation is inferior to joint estimation, despite having
much fewer parameters to estimate. The best out-of-sample log-likelihood
is obtained with joint estimation of the model with the CT distribution
and a sparse block correlation matrix. Overall, the out-of-sample
results are in line with BIC statistics reported in Table \ref{tab:EstSmall}.
\begin{table}[t]
\caption{Out-of-sample Results in Small Universe}

\begin{centering}
\begin{footnotesize}
\begin{tabularx}{\textwidth}{p{3.5cm}YYYYYYYYYY}
\toprule 
\midrule
          & Gauss & MT    & CT    & HT \\
    \midrule
          \\[-0.2cm]          
          \multicolumn{5}{l}{\it Unrestricted $C_e$} \\
          \\[-0.3cm]           
    Joint   & -14040   & -12592 & -12288 & -12303 \\
    Decoupled   & -14049 & -12613 & -12291 & -12312 \\
          \\[-0.2cm]          
          \multicolumn{5}{l}{\it Full Block $C_e$} \\
          \\[-0.3cm] 
    Joint   & -14003   & -12557 & -12214 & -12242 \\
    Decoupled   & -14058 & -12599 & -12271 & -12292 \\
          \\[-0.2cm]          
          \multicolumn{5}{l}{\it Sparse Block $C_e$} \\
          \\[-0.3cm] 
    Joint   &  -14014   & -12551 & \textbf{-12202} & -12230 \\
    Decoupled   & -14055 & -12598 & -12241 & -12256 \\
          \\[-0.2cm]          
          \multicolumn{5}{l}{\it Diagonal Block  $C_e$} \\
           \\[-0.3cm] 
    Joint   & -14289 & -12900 & -12602 & -12603 \\
    Decoupled   & -14299 & -12854 & -12545 & -12582 \\
    \\[-0.2cm]  
\midrule
\bottomrule
\end{tabularx}
\end{footnotesize}

\par\end{centering}
{\small Note: The out-of-sample log-likelihood $\ell(Z|U)$ is reported
for different distributional specifications, different structures
on the idiosyncratic correlation matrix, $C_{e}$, and the two estimation
methods (joint and decoupled). The in-sample (estimation) period spans
the years 2007 to 2018 and the out-of-sample (evaluation) period spans
the years 2019 to 2023. The largest out-of-sample log-likelihood is
highlighted in bold font.\label{tab:OOSsmall}}{\small\par}
\end{table}

\subsection{Analysis of Large Universe}

The large universe has $n=323$ stocks that are distributed over 9
sectors and 63 subindustries. We use subindustries to define the block
structure for the idiosyncratic correlation matrix, $C_{e}$, as we
did in the small universe. The factor variables include the six cross-sectional
factors (FF5+UMD) and the nine sector factors, such that this analysis
is based on $r=15$ factors.

The large number of stocks necessitates the use of decoupled estimation,
as it is not practical to estimated the $232\times15$ dynamic factor
loadings simultaneously in conjunction with a dynamic model for $C_{e}\in\mathbb{R}^{323\times323}$.
We do not attempt to estimate a model for an unrestricted $C_{e}$,
because it has 52,003 distinct correlations. The block structure with
$K=63$ reduce this number to 2,016, which is also unpractically large.
Instead we focus on the two sparse correlation structure for $C_{e}$:
The sparse block correlation matrix, with 296 distinct correlations,
and the diagonal block correlation matrix with 63 distinct correlations. 

We do not consider the multivariate $t$-distribution (MT) in the
large universe for two reasons. One is that it was found to be inferior
to CT and HT in the Small Universe. More importantly, estimation with
the MT distribution is far more involved than those for the three
other distributions, because MT entails non-linear dependencies across
subindustries. With decoupled estimation, the Gaussian, CT, and HT
specifications, have simplified estimation of the dynamic idiosyncratic
correlation matrix, because estimation can be done for each sector
(or subindustry) separately. The reason is that the Gaussian and HT
distributions are composed of independent elements, and the CT has
a cluster structure that is aligned with subindustries.\footnote{The computationally most complex structure is that for the sector
``Industrials'' which has 11 subindustries, which is manageable.} 

Figure \ref{fig:BlockCorrLargeUniverse} presents the sample block
correlation matrix, $C_{e}$, based on the residuals from the $323$
dynamic factor loading models. The correlation matrix is estimated
using the method of moments by \citet{ArchakovHansen:CanonicalBlockMatrix}
with a $63\times63$ block structure based on subindustries. Solid
lines are used to indicate the nine sectors and dashed lines are used
to indicate the GICS Industries of which 42 are represented in the
Large Universe. The names of the nine sectors are give along with
their two-digit GICS code. For better visualization, we have truncated
correlations smaller than 0.05 in absolute value to zero. There are
many nontrivial correlations between subindustries, which renders
the diagonal block correlation structure invalid. However, most non-trivial
correlations are found between subindustries within the same sector,
which is consistent with the sparse block correlation structure, which
can also accommodate the heterogeneity in the correlations within
and between subindustries. Most of large correlations between subindustries
are for subindustries in the same industries (six-digit GICS codes).
Two notable exceptions to this is the Material subindustry, ``Fertilizers
\& Agricultural Chemicals'', which is correlated with several Energy
subindustries, and the ``Consumer Staples Merchandise Retail'' subindustry,
which is correlated with subindustries in the Consumer Discretionary
sector, primarily those in the ``Consumer Discretionary Distribution
\& Retail'' industry group (2550).

The results in Figure \ref{fig:BlockCorrLargeUniverse} was based
on a static block correlation matrix. We present the results for the
dynamic model for $C_{e}$ in Table \ref{tab:SeqEstLarge1} for each
of the nine sectors. For each sector, we report the average estimates
of $\mu$, $\alpha$, $\beta$, and $\nu$, as well as the range of
the estimates for $\nu_{i}$ (the degrees of freedom parameters).
We also report the number of model parameters, $p$, in each sector
model, the negative log-likelihood function, $-\ell(e_{\{j\}})$,
and the corresponding Bayesian Information Criterion (BIC). The largest
log-likelihood and the smallest BIC are highlighted in bold font for
each sector, and the best model specification is. The best model specification
is in all cases provided by the HT distribution with the Sparse Block
Correlation structure. The latter confirms the need to account for
correlations between subindustries in the same sector, as was indicated
by the static sample correlations in Figure \ref{fig:BlockCorrLargeUniverse}.
We note that the estimate of $\bar{\mu}$ is smaller in the specifications
with Sparse Block structure, relative to Diagonal Block structure,
which is consistent with within-subindustry correlations tend to be
much larger than between-subindustry correlations.

\begin{table}[tbh]
\caption{Estimation Results (Large Universe)}

\begin{centering}
\begin{footnotesize}
\begin{tabularx}{\textwidth}{p{1.5cm}YYYYYYYYYYYYYYY}
\toprule 
\midrule
          & $\bar{\mu}$    & $\bar{\beta}$  & $\bar{\alpha}$ & $\bar{\nu}$    & ${\nu}_{\min}$  & ${\nu}_{\max}$  & $p$     & -$\ell(e)$   & BIC \\
    \midrule
    \multicolumn{10}{c}{\textbf{Energy Sector (10)}} \\
    \multicolumn{10}{l}{\textit{Sparse Block $C_e$}} \\
    Gauss & 0.095 & 0.845 & 0.018 &       &       &       & 45    & 98096 & 196569 \\
    CT    & 0.103 & 0.883 & 0.016 & 5.360 & 4.387 & 6.232 & 50    & 93065 & 186548 \\
    HT   & 0.103 & 0.880 & 0.015 & 4.741 & 3.511 & 6.915 & 62    & \textbf{92892} & \textbf{186302} \\
    \multicolumn{10}{l}{\textit{Diagonal Block $C_e$}} \\
    Gauss & 0.222 & 0.994 & 0.012 &       &       &       & 15    & 98792 & 197709 \\
    CT    & 0.255 & 0.993 & 0.016 & 5.368 & 4.383 & 6.196 & 20    & 93821 & 187809 \\
    HT   & 0.249 & 0.994 & 0.014 & 4.746 & 3.559 & 6.997 & 30    & {93674} & {187599} \\   
    \midrule
    \multicolumn{10}{c}{\textbf{Materials Sector (15)}} \\    
    \multicolumn{10}{l}{\textit{Sparse Block $C_e$}} \\
    Gauss & 0.082 & 0.800 & 0.025 &       &       &       & 45    & 146905 & 294186 \\
    CT    & 0.082 & 0.807 & 0.023 & 4.406 & 3.792 & 5.136 & 50    & 134784 & 269985 \\
    HT   & 0.082 & 0.855 & 0.021 & 3.704 & 2.899 & 5.131 & 70    & \textbf{133168} & \textbf{266921} \\
    \multicolumn{10}{l}{\textit{Diagonal Block $C_e$}} \\
    Gauss & 0.163 & 0.939 & 0.019 &       &       &       & 15    & 147415 & 294955 \\
    CT    & 0.208 & 0.958 & 0.026 & 4.416 & 3.823 & 5.148 & 20    & 135351 & 270868 \\
    HT   & 0.196 & 0.975 & 0.020 & 3.705 & 2.892 & 5.171 & 30    & {133787} & {267825} \\    
    \midrule
    \multicolumn{10}{c}{\textbf{Industrials Sector (20)}} \\   
    \multicolumn{10}{l}{\textit{Sparse Block $C_e$}} \\
    Gauss & 0.044 & 0.834 & 0.014 &       &       &       & 198   & 310906 & 623467 \\
    CT    & 0.048 & 0.860 & 0.013 & 4.211 & 3.474 & 5.010 & 209   & 282381 & 566509 \\
    HT   & 0.048 & 0.833 & 0.016 & 3.612 & 2.832 & 4.397 & 251   & \textbf{279157} & \textbf{560412} \\
    \multicolumn{10}{l}{\textit{Diagonal Block $C_e$}} \\
    Gauss & 0.184 & 0.839 & 0.029 &       &       &       & 33    & 313656 & 627588 \\
    CT    & 0.220 & 0.956 & 0.028 & 4.240 & 3.499 & 5.051 & 44    & 285484 & 571336 \\
    HT   & 0.210 & 0.889 & 0.035 & 3.614 & 2.829 & 4.495 & 121   & {282542} & {566096} \\   
    \midrule
    \multicolumn{10}{c}{\textbf{Consumer Discretionary Sector (25)}} \\  
    \multicolumn{10}{l}{\textit{Sparse Block $C_e$}} \\
    Gauss & 0.048 & 0.848 & 0.013 &       &       &       & 165   & 268484 & 538348 \\
    CT    & 0.054 & 0.881 & 0.011 & 3.939 & 3.145 & 4.939 & 175   & 238769 & 479001 \\
    HT   & 0.054 & 0.858 & 0.013 & 3.474 & 2.514 & 6.319 & 212   & \textbf{236154} & \textbf{474080} \\
    \multicolumn{10}{l}{\textit{Diagonal Block $C_e$}} \\
    Gauss & 0.177 & 0.906 & 0.021 &       &       &       & 30    & 270545 & 541341 \\
    CT    & 0.222 & 0.968 & 0.021 & 3.968 & 3.155 & 4.935 & 40    & 241168 & 482671 \\
     HT   & 0.209 & 0.962 & 0.029 & 3.500 & 2.507 & 6.311 & 90    & {238843} & {478439} \\       
    \midrule
    \multicolumn{10}{c}{\textbf{Consumer Staples Sector (30)}} \\  
    \multicolumn{10}{l}{\textit{Sparse Block $C_e$}} \\
    Gauss & 0.096 & 0.770 & 0.029 &       &       &       & 18    & 93581 & 187313 \\
    CT    & 0.102 & 0.902 & 0.026 & 4.054 & 3.889 & 4.220 & 21    & 84125 & 168426 \\
    HT   & 0.102 & 0.909 & 0.026 & 3.367 & 3.090 & 4.035 & 34    & \textbf{83085} & \textbf{166455} \\
    \multicolumn{10}{l}{\textit{Diagonal Block $C_e$}} \\
    Gauss & 0.183 & 0.875 & 0.030 &       &       &       & 9     & 93710 & 187495 \\
    CT    & 0.196 & 0.898 & 0.029 & 4.045 & 3.884 & 4.193 & 12    & 84225 & 168550 \\
    HT    & 0.185 & 0.918 & 0.025 & 3.355 & 3.092 & 3.922 & 21    & {83225} & {166626} \\           
 \midrule    
    \bottomrule
\end{tabularx}
\end{footnotesize}

\par\end{centering}
{\small Note: Table continued on next page.\label{tab:SeqEstLarge1}}{\small\par}
\end{table}

\addtocounter{table}{-1} 

\begin{table}[tbh]
\caption{(cont.)}

\begin{centering}
\begin{footnotesize}
\begin{tabularx}{\textwidth}{p{1.5cm}YYYYYYYYYYYYYYY}
\toprule
\midrule 
          & $\bar{\mu}$    & $\bar{\beta}$  & $\bar{\alpha}$ & $\bar{\nu}$    & ${\nu}_{\min}$  & ${\nu}_{\max}$  & $p$     & -$\ell(e)$   & BIC \\
    \midrule
    \multicolumn{10}{c}{\textbf{Health Care Sector (35)}} \\
    \multicolumn{10}{l}{\textit{Sparse Block $C_e$}} \\
    Gauss & 0.038 & 0.866 & 0.015 &       &       &       & 108   & 275826 & 552554 \\
    CT    & 0.036 & 0.914 & 0.013 & 3.964 & 3.400 & 4.420 & 116   & 246629 & 494228 \\
    HT    & 0.036 & 0.921 & 0.011 & 3.318 & 2.511 & 4.264 & 155   & \textbf{242199} & \textbf{485695} \\
    \multicolumn{10}{l}{\textit{Diagonal Block $C_e$}} \\
    Gauss & 0.139 & 0.910 & 0.024 &       &       &       & 24    & 277265 & 554730 \\
    CT    & 0.137 & 0.966 & 0.021 & 3.990 & 3.458 & 4.443 & 32    & 248323 & 496913 \\
    HT    & 0.129 & 0.957 & 0.021 & 3.340 & 2.518 & 4.212 & 88    & {244037} & {488811} \\
    \midrule
    \multicolumn{10}{c}{\textbf{Financial Sector (40)}} \\
    \multicolumn{10}{l}{\textit{Sparse Block $C_e$}} \\
    Gauss & 0.042 & 0.937 & 0.012 &       &       &       & 165   & 316545 & 634470 \\
    CT    & 0.045 & 0.922 & 0.016 & 4.536 & 3.671 & 5.272 & 175   & 289209 & 579880 \\
    HT    & 0.045 & 0.919 & 0.019 & 3.935 & 3.014 & 6.127 & 220   & \textbf{287231} & \textbf{576301} \\
    \multicolumn{10}{l}{\textit{Diagonal Block $C_e$}} \\
    Gauss & 0.166 & 0.939 & 0.033 &       &       &       & 30    & 320172 & 640594 \\
    CT    & 0.184 & 0.968 & 0.036 & 4.498 & 3.693 & 5.063 & 40    & 293115 & 586564 \\
    HT    & 0.171 & 0.957 & 0.041 & 3.842 & 3.005 & 5.589 & 90    & {291676} & {584105} \\
    \midrule
    \multicolumn{10}{c}{\textbf{Information Technology Sector (45)}} \\
    \multicolumn{10}{l}{\textit{Sparse Block $C_e$}} \\
    Gauss & 0.041 & 0.830 & 0.014 &       &       &       & 135   & 269981 & 541091 \\
    CT    & 0.043 & 0.868 & 0.011 & 3.810 & 3.122 & 4.942 & 144   & 238939 & 479083 \\
    HT    & 0.043 & 0.899 & 0.011 & 3.282 & 2.658 & 4.981 & 181   & \textbf{233593} & \textbf{468699} \\
    \multicolumn{10}{l}{\textit{Diagonal Block $C_e$}} \\
    Gauss & 0.119 & 0.949 & 0.012 &       &       &       & 27    & 272230 & 544687 \\
    CT    & 0.146 & 0.979 & 0.015 & 3.842 & 3.186 & 4.953 & 36    & 241795 & 483892 \\
    HT    & 0.137 & 0.968 & 0.017 & 3.288 & 2.655 & 4.702 & 72    & {236975} & {474552} \\
    \midrule
    \multicolumn{10}{c}{\textbf{Utilities Sector (55)}} \\
    \multicolumn{10}{l}{\textit{Sparse Block $C_e$}} \\
    Gauss & 0.084 & 0.921 & 0.036 &       &       &       & 9     & 100548 & 201171 \\
    CT    & 0.091 & 0.975 & 0.035 & 6.969 & 6.958 & 6.979 & 11    & 96396 & 192884 \\
    HT    & 0.091 & 0.945 & 0.053 & 5.370 & 4.293 & 6.761 & 26    & \textbf{96080} & \textbf{192377} \\
    \multicolumn{10}{l}{\textit{Diagonal Block $C_e$}} \\
    Gauss & 0.109 & 0.940 & 0.029 &       &       &       & 6     & 101197 & 202444 \\
    CT    & 0.118 & 0.971 & 0.042 & 6.878 & 6.877 & 6.879 & 8     & 97045 & 194157 \\
    HT    & 0.122 & 0.954 & 0.057 & 5.299 & 4.158 & 6.721 & 24    & {96871} & {193943} \\
        \midrule
    \bottomrule
\end{tabularx}
\end{footnotesize}

\par\end{centering}
{\small Note: The estimation results for the dynamic idiosyncratic
correlation matrix $C_{e}$ in the large universe, obtained by maximizing
$\ell(e)$. For each of the nine sectors. We report the average values
of $\mu$, $\alpha$, $\beta$, and $\nu$, as well as the range of
the estimates of degrees of freedom. Additionally, the number of model
parameters $p$, the negative log-likelihood function $\ell(e_{\{j\}})$,
and the Bayesian Information Criterion (BIC) are included. The largest
log-likelihood and smallest BIC values for each structure of $C_{e}$
are highlighted in bold font.}{\small\par}
\end{table}

Figure \ref{fig:DynamicBlockCorrLarge} presents the time series of
various correlations derived from the dynamic idiosyncratic correlation
matrix $C_{e}$ in large universe, estimated with a sparse block correlation
matrix, $C_{e}$, and the HT distribution. Most non-trivial idiosyncratic
correlations are for subindustries in the same industry, see Figure
\ref{fig:BlockCorrLargeUniverse}. So, for nine industries (one from
each sector) we selected a pair of subindustries. Their time series
of within-subindustry and between-subindustries correlations are displayed
in Figure \ref{fig:DynamicBlockCorrLarge}. Many of these time series
have strong and persistent variation e can find a very evident time
varying pattern in each subfigures, which means it is very important
to consider the dynamics in idiosyncratic correlation matrix $C_{e}$. 

Table \ref{tab:OOSLarge} conduct the out-of-sample analysis. As the
same in small universe, we estimate all models using data from 2007
to 2018 and evaluate the estimated models using out-of-sample data
from 2019 to 2023. The out-of-sample results are largely consistent
with the in-sample results. The HT distribution with a sparse block
structure for $C_{e}$ provides the best performance in all sectors,
which the exception of the Energy sector where CT (which permits nonlinear
dependencies with subindustries) outperform the HT distribution. The
out-of-sample results show that the in-sample results in Table \ref{tab:SeqEstLarge1}
are not driven by overfitting.

\begin{figure}[ph]
\centering
\centering{}\includegraphics[width=1\textwidth]{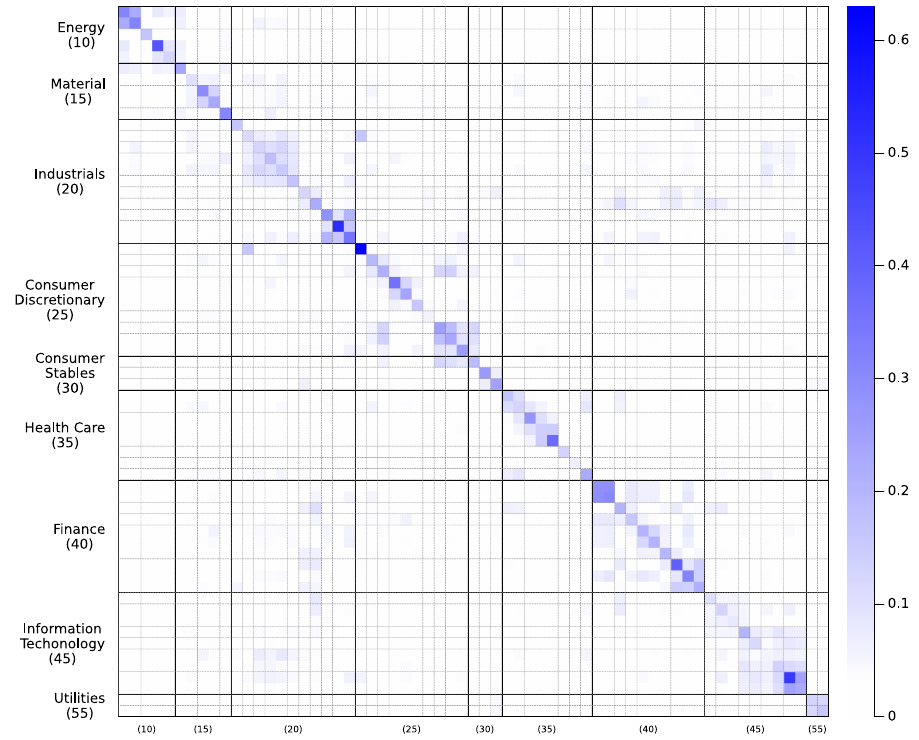}\caption{{\small The sample estimate of $C_{e}$ with a block structure defined
by subindustry. The nine sectors are indicated with the black lines
and industries are indicated with dashed lines. The block correlation
matrix is estimated from the residuals, $\hat{e}_{t}$, using the
method of \citet{ArchakovHansen:CanonicalBlockMatrix}. For better
visualization, correlations smaller than 0.05 in absolute value are
truncated to zero.\label{fig:BlockCorrLargeUniverse}}}
\end{figure}

\begin{figure}[ph]
\centering
\centering{}\includegraphics[width=1\textwidth]{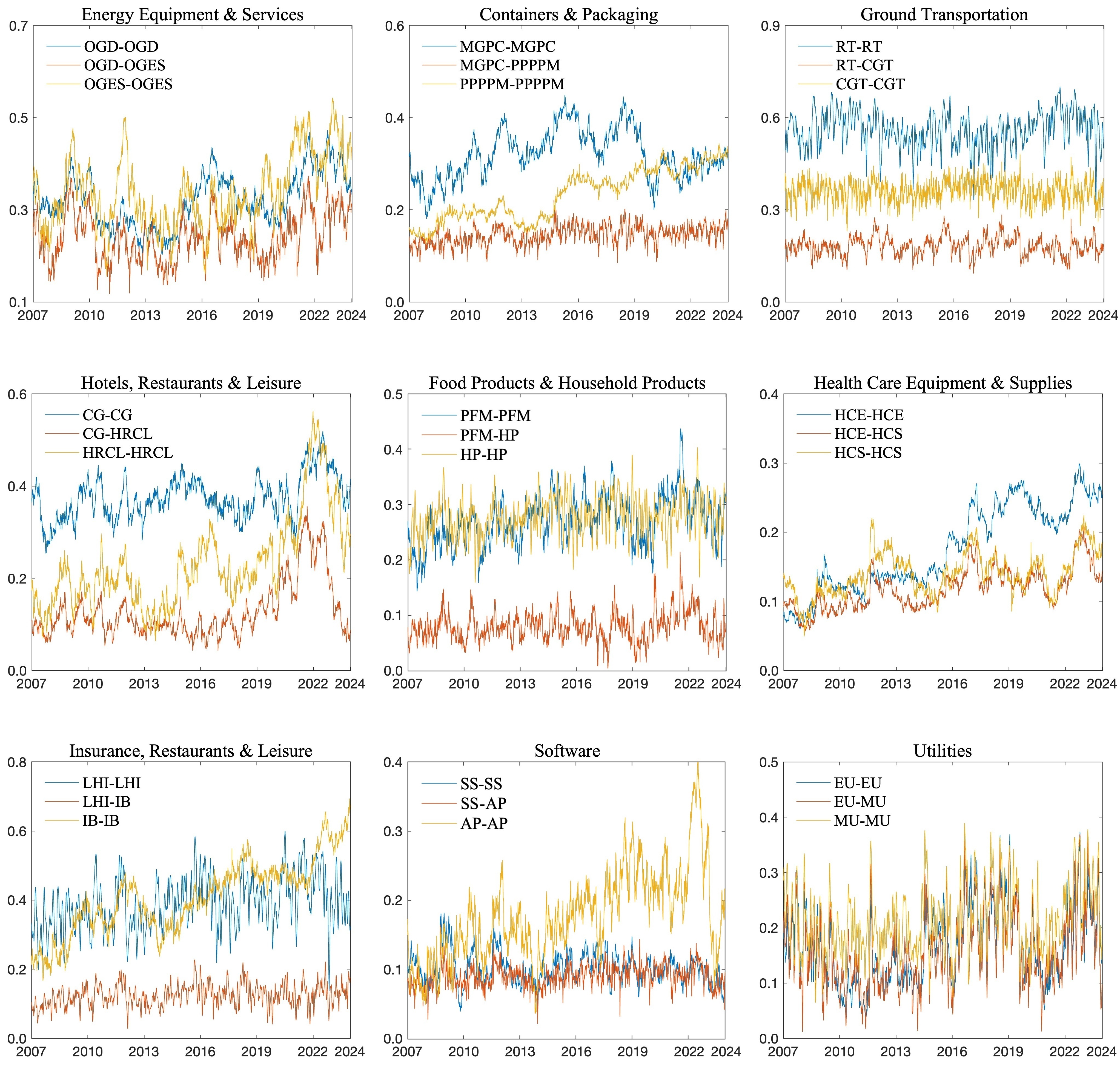}\caption{{\small Some time series for the idiosyncratic correlations in the
large universe, estimated with a sparse block structure and the HT
distribution. For nine industries (four digits) we selected a pair
of subindustries, and present their within- and between-subindustry
correlations. Subindustries are labeled with the first letters in
their names (e.g. ``OGD'' refers to the ``Oil \& Gas Drilling''
subindustry. \label{fig:DynamicBlockCorrLarge}}}
\end{figure}

\begin{table}
\caption{Out-of-sample Results in Large Universe}

\begin{centering}
\begin{footnotesize}
\begin{tabularx}{\textwidth}{p{3.8cm}YYYp{0cm}YYYYYY}
\toprule
\midrule 
          & \multicolumn{3}{c}{Sparse Block $C_e$} &       & \multicolumn{3}{c}{Diagonal Block $C_e$} \\
\cmidrule{2-4}\cmidrule{6-8}    Sector & \multicolumn{1}{c}{Gauss} & \multicolumn{1}{c}{CT} & \multicolumn{1}{c}{HT} &       & \multicolumn{1}{c}{Gauss} & \multicolumn{1}{c}{CT} & \multicolumn{1}{c}{HT} \\
    \midrule
    \\[-0.2cm]          
    Energy & -28378  & \textbf{-26851} & -26884 &       & -28584 & -27103 & -27132 \\
    \\[-0.4cm]      
    Materials & -43471  & -39734 & \textbf{-39498} &       & -43584 & -39887 & -39655 \\
    \\[-0.4cm]     
     Industrials & -90031  & -81666 & \textbf{-81301} &       & -90841 & -82591 & -82314 \\
    \\[-0.4cm]     
    Consumer Discretionary & -77730  & -69497 & \textbf{-69267} &       & -78353 & -70282 & -70070 \\
    \\[-0.4cm]     
    Consumer Staples & -28190  & -24983 & \textbf{-24798} &       & -28247 & -25032 & -24863 \\
        \\[-0.4cm] 
    Health Care & -81218  & -71819 & \textbf{-71037} &       & -81477 & -72216 & -71450 \\
            \\[-0.4cm] 
    Financial & -93339  & -83056 & \textbf{-82745} &       & -94066 & -84246 & -84125 \\
                \\[-0.4cm] 
     Information Technology & -76309  & -66810 & \textbf{-65709} &       & -77112 & -67784 & -66846 \\
                    \\[-0.4cm] 
    Utilities & -29539  & -28252 & \textbf{-28200} &       & -29771 & -28540 & -28498 \\
    \\[-0.2cm]  
\midrule
\bottomrule
\end{tabularx}
\end{footnotesize}

\par\end{centering}
{\small Note: Out-of-sample log-likelihoods, $\ell(e)$, for for six
model-specifications (three distributions and two structures for $C_{e}$)
for each of the nine sectors. In-sample estimation is based on data
from 2007 to 2018 and the out-of-sample period spans the five years,
from 2019 to 2023. The largest log-likelihood in each row is highlighted
in bold font.\label{tab:OOSLarge}}{\small\par}
\end{table}

\section{Conclusion\label{sec:Conclusion}}

In this paper, we introduced a novel dynamic factor correlation model
with a variation-free parametrization of factor loadings. To the best
of our knowledge, this is the first multivariate GARCH model that
simultaneously can accommodate dynamic factor structures, time-varying,
heavy-tailed distributions, and dependent idiosyncratic shocks which
time-varying correlations. The proposed framework incorporates a score-driven
approach to jointly estimate dynamic factor loadings and the idiosyncratic
correlation matrix. A decoupled estimation strategy make it possible
to scale the model to high dimensions, as illustrated in our empirical
application. Our novel parametrization of factor loadings facilitates
simple imposition of sparsity constraints and simplifies the dynamic
modeling of these.

The empirical applications demonstrate the flexibility and effectiveness
of the proposed model. The analysis of a small universe with 12 assets
allowed for comparisons across different specifications and estimation
methods, revealing the advantages of a convolution-$t$ distribution
with a block structure in the idiosyncratic correlation matrix. Our
large universe application involve 323 assets from the S\&P 500 and
it demonstrate the scalability of the model. The empirical results
underscored the importance of permitting some dependence in the idiosyncratic
correlation matrix. We found that the sparse block structure -- partitioned
by sectors and subindustries -- provides the best fit according to
likelihood-based criteria. Similarly, we found that Convolution-$t$
distributions outperform the standard specifications with Gaussian
and multivariate $t$-distributions. There was strong evidence for
time variation in the factor loadings, which can explain dynamic shifts
in asset return dependencies. The empirical evidence for nontrivial
correlation between idiosyncratic shocks is strong, but the nontrivial
correlations are primarily found between stocks in the same sector.
These observations a consistent with the structure of a sparse block
correlation matrix, which was also the structure that had the best
out-of-sample performance in all comparisons. 

\bibliographystyle{apalike}
\bibliography{prh}

\clearpage{}

\appendix

\section{Appendix of Proofs}

\subsection*{A Novel Parametrization of Factor Loadings\label{sec:NewParaCorr}}

\noindent \textbf{Proof of Theorem \ref{thm:CorrTau}.} Let $C^{\star}={\rm cor}\left((Z_{i},U^{\prime})^{\prime}\right)$,
then $C^{\star}=I+H$ where 
\[
H=\left[\begin{array}{cc}
0 & \rho_{i}^{\prime}\\
\rho_{i} & 0
\end{array}\right]\in\mathbb{R}^{(r+1)\times(r+1)}.
\]
Since the largest eigenvalue of $H^{\prime}H$ is $\rho_{i}^{\prime}\rho_{i}$,
we have the following expression for the matrix logarithm
\begin{align*}
\log(C^{\star}) & =\sum_{k=1}^{\infty}(-1)^{k+1}\frac{H^{k}}{k}=-\sum_{k=1}^{\infty}\frac{H^{2k}}{2k}+\sum_{k=1}^{\infty}\frac{H^{2k-1}}{2k-1},
\end{align*}
where 
\[
H^{2k}=\left[\begin{array}{cc}
(\rho_{i}^{\prime}\rho_{i})^{k} & 0\\
0 & (\rho_{i}\rho_{i}^{\prime})^{k}
\end{array}\right]\qquad\text{and}\qquad H^{2k-1}=\left[\begin{array}{cc}
0 & (\rho_{i}^{\prime}\rho_{i})^{k-1}\rho_{i}^{\prime}\\
\rho_{i}(\rho_{i}^{\prime}\rho_{i})^{k-1} & 0
\end{array}\right].
\]
So, the diagonal blocks of $\log(C^{\star})$ are given by $-\sum_{k=1}^{\infty}\frac{1}{2k}(\rho_{i}^{\prime}\rho_{i})^{k}=\frac{1}{2}\log\left(1-\rho_{i}^{\prime}\rho_{i}\right)\in\mathbb{R}$
and 
\[
-\sum_{k=1}^{\infty}\frac{1}{2k}\left(\rho_{i}\rho_{i}^{\prime}\right)^{k}=-\frac{1}{2\rho_{i}^{\prime}\rho_{i}}\left[\sum_{k=1}^{\infty}\frac{1}{k}\left(\rho_{i}\rho_{i}^{\prime}\right)^{k}\right]\left(\rho\rho^{\prime}\right)=\frac{\log\left(1-\rho_{i}^{\prime}\rho_{i}\right)}{2\rho_{i}^{\prime}\rho_{i}}\left(\rho_{i}\rho_{i}^{\prime}\right)\in\mathbb{R}^{r\times r}.
\]
The off-diagonal blocks are zero if $\rho_{i}^{\prime}\rho_{i}=0$.
Otherwise, if $\rho_{i}^{\prime}\rho_{i}\neq0$, we have 
\begin{align*}
\tau_{i}^{\prime}=\sum_{k=1}^{\infty}\frac{\left(\rho_{i}^{\prime}\rho_{i}\right)^{k-1}}{2k-1}\rho_{i}^{\prime} & =\tfrac{1}{\sqrt{\rho_{i}^{\prime}\rho_{i}}}\left[\sum_{k=1}^{\infty}\frac{\left(\sqrt{\rho_{i}^{\prime}\rho_{i}}\right)^{2k-1}}{2k-1}\right]\rho_{i}^{\prime}\\
 & =\tfrac{1}{\sqrt{\rho_{i}^{\prime}\rho_{i}}}\left[\tfrac{1}{2}\ensuremath{\log\left(\tfrac{1+\sqrt{\rho_{i}^{\prime}\rho_{i}}}{1-\sqrt{\rho_{i}^{\prime}\rho_{i}}}\right)}\right]\rho_{i}^{\prime}=\tfrac{1}{\sqrt{\rho_{i}^{\prime}\rho_{i}}}\left[\operatorname{artanh}(\sqrt{\rho_{i}^{\prime}\rho_{i}})\right]\rho_{i}^{\prime},
\end{align*}
which proves the expression for $\tau_{i}$ (\ref{eq:Tau_rho}). 

Next, we observed that $\tau_{i}^{\prime}\tau_{i}=\operatorname{artanh}^{2}(\sqrt{\rho_{i}^{\prime}\rho_{i}})$
such that $\sqrt{\tau_{i}^{\prime}\tau_{i}}=\operatorname{artanh}(\sqrt{\rho_{i}^{\prime}\rho_{i}})$
is the Fisher transformation of $\sqrt{\rho_{i}^{\prime}\rho_{i}}$,
which has inverse
\[
\sqrt{\rho_{i}^{\prime}\rho_{i}}=\operatorname{tanh}(\sqrt{\tau_{i}^{\prime}\tau_{i}})=\tfrac{e^{2\sqrt{\tau_{i}^{\prime}\tau_{i}}}-1}{e^{2\sqrt{\tau_{i}^{\prime}\tau_{i}}}+1}.
\]
So, we can rewrite the expression $\tau_{i}=\tfrac{1}{\sqrt{\rho_{i}^{\prime}\rho_{i}}}\sqrt{\tau_{i}^{\prime}\tau_{i}}\rho_{i}$.
which proves the expression for the inverse mapping, 
\[
\rho_{i}=\sqrt{\frac{\rho_{i}^{\prime}\rho_{i}}{\tau_{i}^{\prime}\tau_{i}}}\times\tau_{i}=\frac{\tanh(\sqrt{\tau_{i}^{\prime}\tau_{i}})}{\sqrt{\tau_{i}^{\prime}\tau_{i}}}\times\tau_{i}.
\]
The corresponding Jacobian is given by
\begin{eqnarray*}
\frac{\partial\rho_{i}}{\partial\tau_{i}^{\prime}} & = & \sqrt{\frac{\rho_{i}^{\prime}\rho_{i}}{\tau_{i}^{\prime}\tau_{i}}}I_{r}+\frac{(1-\rho_{i}^{\prime}\rho_{i})\sqrt{\tau_{i}^{\prime}\tau_{i}}-\sqrt{\rho_{i}^{\prime}\rho_{i}}}{\tau_{i}^{\prime}\tau_{i}}(\tau_{i}^{\prime}\tau_{i})^{-1/2}\tau_{i}\tau_{i}^{\prime}\\
 & = & \sqrt{\frac{\rho_{i}^{\prime}\rho_{i}}{\tau_{i}^{\prime}\tau_{i}}}(I_{r}-P_{\tau_{i}})+(1-\rho_{i}^{\prime}\rho_{i})P_{\tau_{i}},
\end{eqnarray*}
where we used that $\partial\sqrt{\rho_{i}^{\prime}\rho_{i}}/\partial\sqrt{\tau_{i}^{\prime}\tau_{i}}=\partial\tanh(\sqrt{\tau_{i}^{\prime}\tau_{i}})/\partial\sqrt{\tau_{i}^{\prime}\tau_{i}}=1-\tanh^{2}(\sqrt{\tau_{i}^{\prime}\tau_{i}})=1-\rho_{i}^{\prime}\rho_{i}$.\hfill{}$\square$

\subsection*{Proof of Theorem \ref{thm:JointScore}\label{sec:Proof-of-Theorem1}}

We know that $Z|U\sim\mathrm{CT}_{\boldsymbol{m},\boldsymbol{\nu}}^{{\rm std}}(\mu,\Xi)$
where $\mu=\boldsymbol{\rho}^{\prime}U$ and $\Xi=\Lambda_{\omega}C_{e}^{1/2}$.
We have the form of the score $\nabla_{\mu}$ and $\nabla_{\Xi}=\nabla_{{\rm vec}(\Xi)}$,
such that
\[
\nabla_{\xi}=\Theta^{\prime}\left(\begin{array}{c}
\begin{array}{c}
\nabla_{\mu}\\
\nabla_{\Xi}
\end{array}\end{array}\right),\quad\mathcal{I}_{\xi}=\Theta^{\prime}\left(\begin{array}{cc}
\mathcal{I}_{\mu} & 0\\
0 & \mathcal{I}_{\Xi}
\end{array}\right)\Theta,
\]
where $\Theta=\frac{\partial\theta}{\partial\xi^{\prime}}$ with $\theta=\left[\mu^{\prime},{\rm vec}\left(\Xi\right)^{\prime}\right]^{\prime}$.
And by defining
\[
\tilde{\xi}\equiv\left(\begin{array}{c}
\mu\\
\sigma\\
\eta
\end{array}\right)=\left[\begin{array}{cc}
K_{2n} & \bm{0}\\
\bm{0} & I_{n(n-1)/2}
\end{array}\right]\xi\quad{\rm where}\quad\xi=\left(\begin{array}{c}
{\rm vec}\left(\mu,\sigma\right)^{\prime}\\
\eta
\end{array}\right),
\]
where $K_{2n}$ is the communication matrix, and we now have
\begin{align*}
\Theta & =\frac{\partial\theta}{\partial\xi^{\prime}}=\frac{\partial\theta}{\partial\tilde{\xi}^{\prime}}\frac{\partial\tilde{\xi}}{\partial\xi^{\prime}}=\left[\begin{array}{ccc}
I_{n} & \bm{0} & \bm{0}\\
\bm{0} & \frac{\partial{\rm vec}\left(\Xi\right)}{\partial\sigma^{\prime}} & \frac{\partial{\rm vec}\left(\Xi\right)}{\partial\eta^{\prime}}
\end{array}\right]\left[\begin{array}{cc}
K_{2n} & \bm{0}\\
\bm{0} & I_{n(n-1)/2}
\end{array}\right],
\end{align*}
where
\begin{align*}
\frac{\partial{\rm vec}\left(\Xi\right)}{\partial\sigma^{\prime}} & =\frac{\partial{\rm vec}\left(\Lambda_{\omega}C_{e}^{1/2}\right)}{\partial{\rm vec}\left(\Lambda_{\omega}\right)^{\prime}}\frac{\partial{\rm vec}\left(\Lambda_{\omega}\right)}{\partial{\rm diag}\left(\Lambda_{\omega}\right)^{\prime}}=\left(C_{e}^{1/2}\otimes I_{n}\right)E_{d}^{\prime},
\end{align*}
and
\[
\frac{\partial{\rm vec}\left(\Xi\right)}{\partial\eta^{\prime}}=\frac{\partial{\rm vec}(\Lambda_{\rho}C_{e}^{1/2})}{\partial{\rm vec}(C_{e}^{1/2})^{\prime}}\frac{\partial{\rm vec}(C_{e}^{1/2})}{\partial{\rm vec}(C_{e})^{\prime}}\frac{\partial{\rm vec}(C_{e}^{1/2})}{\partial\gamma^{\prime}}B=\left(I_{n}\otimes\Lambda_{\omega}\right)\left(C_{e}^{1/2}\oplus I_{n}\right)^{-1}\frac{\partial{\rm vec}(C_{e}^{1/2})}{\partial\gamma^{\prime}}B,
\]
and from \citet[Proposition 3]{ArchakovHansen:Correlation}, we have
\[
\frac{\partial{\rm vec}(C_{e}^{1/2})}{\partial\gamma^{\prime}}=\left(E_{l}+E_{u}\right)^{\prime}E_{l}\left(I-\Gamma E_{d}^{\prime}\left(E_{d}\Gamma E_{d}^{\prime}\right)^{-1}E_{d}\right)\Gamma\left(E_{l}+E_{u}\right)^{\prime},
\]
which uses the fact that $\partial\operatorname{vec}(C)/\partial\operatorname{vecl}(C)=E_{l}+E_{u}$,
where $E_{l},E_{u},E_{d}$ are elimination matrices, and the expression
$\Gamma=\partial\mathrm{vec}(C)/\partial\mathrm{vec}(\log C)^{\prime}$
is given in \citet{LintonMcCrorie:1995}.

As for the matrix $M$, we have
\[
M=\frac{\partial{\rm vec}\left(\mu,\sigma\right)^{\prime}}{\partial\left[{\rm vec}\left(\mu,\sigma\right)\right]^{\prime}}\frac{\partial{\rm vec}\left(\mu,\sigma\right)}{\partial\tau^{\prime}}=K_{n2}\left[\begin{array}{c}
\frac{\partial\mu}{\partial\tau^{\prime}}\\
\frac{\partial\sigma}{\partial\tau^{\prime}}
\end{array}\right]
\]
with $\bm{\tau}=\left(\tau_{1},\tau_{2},\ldots,\tau_{n}\right)\in\mathbb{R}^{r\times n}$
and $\tau={\rm vec}\left(\bm{\tau}\right)$, where
\begin{align*}
\frac{\partial\mu}{\partial\tau^{\prime}} & =\frac{\partial\mu}{\partial{\rm vec}(\bm{\rho}^{\prime})^{\prime}}\frac{\partial{\rm vec}(\bm{\rho}^{\prime})}{\partial{\rm vec}(\bm{\tau}^{\prime})^{\prime}}\frac{\partial{\rm vec}(\bm{\tau}^{\prime})}{\partial\tau^{\prime}},\quad\\
\frac{\partial\mu}{\partial{\rm vec}(\bm{\rho}^{\prime})^{\prime}} & =\frac{\partial{\rm vec}(\bm{\rho}^{\prime}U)}{\partial{\rm vec}(\bm{\rho}^{\prime})^{\prime}}=\left(U^{\prime}\otimes I_{n}\right)\\
\frac{\partial{\rm vec}(\bm{\rho}^{\prime})}{\partial{\rm vec}(\bm{\tau}^{\prime})^{\prime}} & =\sum_{i=1}^{n}\frac{\partial\left(I_{d}\otimes P_{i}\right)\rho_{i}}{\partial\left[\left(I_{d}\otimes P_{i}\right)\tau_{i}\right]^{\prime}}=\sum_{i=1}^{n}\left(I_{d}\otimes P_{i}\right)J_{i}\left(I_{d}\otimes P_{i}^{\prime}\right),\quad J_{i}=\frac{\partial\rho_{i}}{\partial\tau_{i}^{\prime}}\\
\frac{\partial{\rm vec}(\bm{\rho}^{\prime})}{\partial\tau^{\prime}} & =\sum_{i=1}^{n}\left(I_{d}\otimes P_{i}\right)J_{i}\left(I_{d}\otimes P_{i}^{\prime}\right)K_{dn}=\sum_{i=1}^{n}\left(I_{d}\otimes P_{i}\right)J_{i}\left(P_{i}^{\prime}\otimes I_{d}\right)\\
\frac{\partial\mu}{\partial\tau^{\prime}} & =\left(U^{\prime}\otimes I_{n}\right)\sum_{i=1}^{n}\left(I_{d}\otimes P_{i}\right)J_{i}\left(P_{i}^{\prime}\otimes I_{d}\right)
\end{align*}
where $P_{i}$ is the $i$-th column of identity matrix $I_{n}$.
Now using
\begin{align*}
\frac{\partial{\rm {\rm diag}}(\Lambda_{\omega})}{\partial{\rm vec}(\bm{\rho}^{\prime})^{\prime}} & =\frac{\partial{\rm {\rm diag}}(\Lambda_{\omega})}{\partial{\rm {\rm diag}}(\Lambda_{\omega}^{2})^{\prime}}\frac{\partial{\rm {\rm diag}}(\Lambda_{\omega}^{2})}{\partial{\rm {\rm diag}}(\bm{\rho}^{\prime}\bm{\rho}^ {})^{\prime}}\frac{\partial{\rm {\rm diag}}(\bm{\rho}^{\prime}\bm{\rho})}{\partial{\rm vec}(\bm{\rho}^{\prime}\bm{\rho}^ {})^{\prime}}\frac{\partial{\rm vec}(\bm{\rho}^{\prime}\bm{\rho})}{\partial{\rm vec}(\bm{\rho}^{\prime})^{\prime}}=-\tfrac{1}{2}\text{\ensuremath{\Lambda_{\omega}^{-1}}}E_{d}\ensuremath{\left(I_{n^{2}}+K_{n}\right)\left(\bm{\rho}^{\prime}\otimes I_{n}\right),}
\end{align*}
where ${\rm {\rm diag}}(\Lambda_{\omega}^{2})=\iota_{n}-{\rm diag}(\bm{\rho}^{\prime}\bm{\rho})$,
we find
\begin{align*}
\frac{\partial\sigma}{\partial\tau^{\prime}} & =\frac{\partial{\rm diag}(\Lambda_{\omega})}{\partial\tau^{\prime}}=-\tfrac{1}{2}\text{\ensuremath{\Lambda_{\rho}^{-1}}}E_{d}\ensuremath{\left(I_{n^{2}}+K_{n}\right)\left(\bm{\rho}^{\prime}\otimes I_{n}\right)}\frac{\partial{\rm vec}(\bm{\rho}^{\prime})}{\partial\tau^{\prime}}\\
 & =-\tfrac{1}{2}\text{\ensuremath{\Lambda_{\omega}^{-1}}}E_{d}\ensuremath{\left(I_{n^{2}}+K_{n}\right)\left(\bm{\rho}^{\prime}\otimes I_{n}\right)}\sum_{i=1}^{n}\left(I_{d}\otimes P_{i}\right)J_{i}\left(P_{i}^{\prime}\otimes I_{d}\right).
\end{align*}
This completes the proof of Theorem \ref{thm:JointScore}.

\subsection*{Proof of Theorem \ref{thm:SeqDynamicLoadings}\label{sec:Proof-of-Theorem2}}

Define $\mu_{i}=\rho_{i}^{\prime}U$, $\sigma_{i}=\sqrt{1-\rho_{i}^{\prime}\rho_{i}}$,
and the log-likelihood function is given by
\begin{align*}
\ell(Z_{i}|U) & =c_{v}-\log\left(\omega_{i}\right)-\frac{\nu_{i}^{\star}+1}{2}\log\left(1+\frac{e_{i}^{2}}{\nu_{i}^{\star}-2}\right)
\end{align*}
where the interested parameter is $\tau_{i}$. We have
\[
\nabla_{\mu_{i}}=\frac{\partial(Z_{i}|U)}{\partial\mu_{i}},\quad\nabla_{\omega_{i}}=\frac{\partial(Z_{i}|U)}{\partial\omega_{i}},\quad\mathbb{E}\left[\nabla_{\mu_{i}\omega_{i}}\right]=0,
\]
and by defining $W_{i}=\left(\nu_{i}^{\star}+1\right)/(\nu_{i}^{\star}-2+e_{i}^{2})$,
we have 
\begin{align*}
\nabla_{\mu_{i}} & =W_{i}\frac{e_{i}}{\sigma_{i}}\\
\nabla_{\omega_{i}} & =W_{i}\frac{e_{i}^{2}}{\omega_{i}}-\frac{1}{\omega_{i}}\\
\mathcal{I}_{\mu_{i}} & =\mathbb{E}\left(W_{i}^{2}e_{i}^{2}/\omega_{i}^{2}\right)=\tfrac{\left(\nu_{i}^{\star}+1\right)\nu_{i}^{\star}}{\left(\nu_{i}^{\star}+3\right)\left(\nu_{i}^{\star}-2\right)}\frac{1}{\omega_{i}^{2}}\\
\mathcal{I}_{\omega_{i}} & =\frac{1}{\omega_{i}^{2}}\mathbb{E}\left[\left(W_{i}e_{i}^{2}-1\right)^{2}\right]=\tfrac{2\nu_{i}^{\star}}{\nu_{i}^{\star}+3}\frac{1}{\omega_{i}^{2}}.
\end{align*}
We now have
\[
M_{i}\equiv\left[\begin{array}{c}
\frac{\partial\mu_{i}}{\partial\tau{}_{i}^{\prime}}\\
\frac{\partial\omega_{i}}{\partial\tau{}_{i}^{\prime}}
\end{array}\right]=\left[\begin{array}{c}
U^{\prime}J_{i}\\
-\frac{1}{\omega_{i}}\rho_{i}^{\prime}J_{i}
\end{array}\right]=\left[\begin{array}{c}
U^{\prime}\\
-\frac{1}{\omega_{i}}\rho_{i}^{\prime}
\end{array}\right]J_{i},
\]
where we used that
\begin{align*}
\frac{\partial\mu}{\partial\tau{}_{i}^{\prime}} & =\frac{\partial\mu}{\partial\rho_{i}^{\prime}}\frac{\partial\rho_{i}}{\partial\tau{}_{i}^{\prime}}=U^{\prime}J_{i}\\
\frac{\partial\sigma_{i}}{\partial\tau{}_{i}^{\prime}} & =\frac{\partial\omega_{i}^ {}}{\partial\omega_{i}^{2}}\frac{\partial\omega_{i}^{2}}{\partial\left(\rho_{i}^{\prime}\rho_{i}\right)}\frac{\partial\left(\rho_{i}^{\prime}\rho_{i}\right)}{\partial\rho_{i}^{\prime}}\frac{\partial\rho_{i}^{\prime}}{\partial\tau{}_{i}^{\prime}}=-\frac{1}{2\omega_{i}}\left(2\rho_{i}^{\prime}J_{i}\right)=-\frac{1}{\omega_{i}}\rho_{i}^{\prime}J_{i}=-\frac{\rho_{i}^{\prime}J_{i}}{\sqrt{1-\rho_{i}^{\prime}\rho_{i}}},
\end{align*}
and we arrive at the expressions
\[
\nabla_{\tau_{i}}=M_{i}^{\prime}\left(\begin{array}{c}
\nabla_{\mu_{i}}\\
\nabla_{\omega_{i}}
\end{array}\right)=J_{i}\left[\frac{W_{i}e_{i}}{\omega_{i}}U-\frac{W_{i}e_{i}^{2}-1}{\omega_{i}^{2}}\rho_{i}\right],
\]
and
\[
\mathcal{I}_{\tau_{i}}=M_{i}^{\prime}\left(\begin{array}{cc}
\mathcal{I}_{\mu_{i}} & 0\\
0 & \mathcal{I}_{\omega_{i}}
\end{array}\right)M_{i}=J_{i}\left[\tfrac{\left(\nu_{i}^{\star}+1\right)\nu_{i}^{\star}}{\left(\nu_{i}^{\star}+3\right)\left(\nu_{i}^{\star}-2\right)}\frac{UU^{\prime}}{\omega_{i}^{2}}+\tfrac{2\nu_{i}^{\star}}{\nu_{i}^{\star}+3}\frac{1-\omega_{i}^{2}}{\omega_{i}^{4}}\right]J_{i}.
\]
We also have
\begin{align*}
M_{i}^{+} & =\frac{1}{U^{\prime}U\rho_{i}^{\prime}\rho_{i}-U^{\prime}\rho_{i}\rho_{i}^{\prime}U}J_{i}^{-1}\left[\begin{array}{c}
U^{\prime}\\
-\frac{1}{\omega_{i}}\rho_{i}^{\prime}
\end{array}\right]^{\prime}\left[\begin{array}{cc}
\rho_{i}^{\prime}\rho_{i} & \omega_{i}U^{\prime}\rho_{i}\\
\omega_{i}U^{\prime}\rho_{i} & \omega_{i}^{2}U^{\prime}U
\end{array}\right]\\
 & =\frac{1}{U^{\prime}U\rho_{i}^{\prime}\rho_{i}-U^{\prime}\rho_{i}\rho_{i}^{\prime}U}J_{i}^{-1}\left[\begin{array}{cc}
U & -\frac{1}{\omega_{i}}\rho_{i}\end{array}\right]\left[\begin{array}{cc}
\rho_{i}^{\prime}\rho_{i} & \omega_{i}U^{\prime}\rho_{i}\\
\omega_{i}U^{\prime}\rho_{i} & \omega_{i}^{2}U^{\prime}U
\end{array}\right]\\
 & =\frac{1}{U^{\prime}U\rho_{i}^{\prime}\rho_{i}-U^{\prime}\rho_{i}\rho_{i}^{\prime}U}J_{i}^{-1}\left[\begin{array}{cc}
U\rho_{i}^{\prime}\rho_{i}-\rho_{i}U^{\prime}\rho_{i} & \omega_{i}\left(UU^{\prime}\rho_{i}-\rho_{i}U^{\prime}U\right)\end{array}\right]\\
 & =\frac{1}{U^{\prime}U\rho_{i}^{\prime}\rho_{i}-U^{\prime}\rho_{i}\rho_{i}^{\prime}U}J_{i}^{-1}\left[\begin{array}{c}
\rho_{i}^{\prime}\rho_{i}U^{\prime}-\rho_{i}^{\prime}U\rho_{i}^{\prime}\\
\omega_{i}\left(\rho^{\prime}UU^{\prime}-U^{\prime}U\rho_{i}^{\prime}\right)
\end{array}\right]^{\prime}.
\end{align*}
This completes the proof of Theorem \ref{thm:SeqDynamicLoadings}.

\subsection*{Proof of Theorem \ref{thm:TquiCorr}}

For $C$ a equicorrelation matrix with coefficient $\varrho$, we
have $C=(1-\varrho)I_{n}+\varrho\iota_{n}\iota_{n}^{\prime}$, 
\[
C^{-1}=\frac{1}{1-\varrho}\left(I_{n}-\frac{\varrho}{1+(n-1)\varrho}\iota_{n}\iota_{n}^{\prime}\right),
\]
and $|C|=\left[1+(n-1)\varrho\right]\left(1-\varrho\right)^{n-1}$.
It follows that
\[
X^{\prime}C^{-1}X=\frac{X^{\prime}X}{1-\varrho}-\frac{\varrho X^{\prime}\iota_{n}\iota_{n}^{\prime}X}{\left(1-\varrho\right)\left[1+(n-1)\varrho\right]},
\]
such that 
\begin{align*}
\ell & \left(X\right)=c_{\nu}-\frac{1}{2}\left[\log\left(1+\left(n-1\right)\varrho\right)+\left(n-1\right)\log\left(1-\varrho\right)\right]\\
 & \quad-\frac{\nu+n}{2}\log\left(1+\frac{1}{\nu-2}\left(\frac{X^{\prime}X}{1-\varrho}-\frac{\varrho X^{\prime}\iota_{n}\iota_{n}^{\prime}X}{\left(1-\varrho\right)\left[1+(n-1)\varrho\right]}\right)\right).
\end{align*}
Furthermore, 
\[
\frac{\partial\ell\left(X\right)}{\partial\varrho}=-\frac{1}{2}\left[\frac{\left(n-1\right)}{1+\left(n-1\right)\varrho}-\frac{\left(n-1\right)}{1-\varrho}\right]-\frac{1}{2}W\left(\frac{X^{\prime}X}{\left(1-\varrho\right)^{2}}-\frac{\left[1+\left(n-1\right)\varrho^{2}\right]X^{\prime}\iota_{n}\iota_{n}^{\prime}X}{\left(1-\varrho\right)^{2}\left[1+(n-1)\varrho\right]^{2}}\right),
\]
and from \citet[Theorem 3]{TongHansenArchakov:2024}, we have
\begin{align*}
\mathcal{I}_{A} & =\tfrac{1}{4}\frac{\left(3\phi-1\right)}{\left(1+\left(n-1\right)\varrho\right)^{2}}+\tfrac{\phi}{2}\frac{1}{\left(1-\varrho\right)^{2}\left(n-1\right)}+\tfrac{1-\phi}{4}\left[\frac{2}{\left(1+\left(n-1\right)\varrho\right)\left(1-\varrho\right)}-\frac{1}{\left(1-\varrho\right)^{2}}\right]\\
 & =\frac{1}{4}\frac{\left(3\phi-1\right)}{\left(1+\left(n-1\right)\varrho\right)^{2}}+\frac{\phi}{2}\frac{1}{\left(1-\varrho\right)^{2}\left(n-1\right)}+\frac{1-\phi}{4}\left[\frac{1-\left(n+1\right)\varrho}{\left(1+\left(n-1\right)\varrho\right)\left(1-\varrho\right)^{2}}\right],
\end{align*}
such that
\[
\mathcal{I}_{\rho}=\mathcal{I}_{A}\left(n-1\right)^{2}=\frac{1}{4}\frac{\left(3\phi-1\right)\left(n-1\right)^{2}}{\left(1+\left(n-1\right)\varrho\right)^{2}}+\frac{\phi}{2}\frac{\left(n-1\right)}{\left(1-\varrho\right)^{2}}+\frac{1-\phi}{4}\left[\frac{\left(1-\left(n+1\right)\varrho\right)\left(n-1\right)^{2}}{\left(1+\left(n-1\right)\varrho\right)\left(1-\varrho\right)^{2}}\right],
\]
$\nabla_{\eta}=J\nabla_{\varrho}$ and $\mathcal{I}_{\gamma}=J^{2}\mathcal{I}_{\varrho}$,
where $J=\frac{\partial\varrho}{\partial\gamma}=\frac{1}{(1-\varrho)(1+(n-1)\varrho)}$.
This completes the proof of Theorem \ref{thm:TquiCorr}.

\subsection*{Proof of Theorem \ref{thm:HTquiCorr}}

For $e\sim\mathrm{CT}_{n,\nu}^{\mathrm{std}}\left(0,C^{1/2}\right)$,
where $C$ is an equicorrelation matrix, we have
\begin{align*}
\ell(X) & =-\frac{1}{2}\log|C|+\sum_{i=1}^{n}c_{i}-\frac{\nu_{i}+1}{2}\log\left(1+\frac{1}{\nu_{i}-2}V_{i}^{2}\right)\\
 & =-\frac{1}{2}\left[\log\left(1+\left(n-1\right)\varrho\right)+\left(n-1\right)\log\left(1-\varrho\right)\right]\\
 & \quad-\sum_{i=1}^{n}\frac{\nu_{i}+1}{2}\log\left(1+\frac{1}{\nu_{i}-2}V_{i}^{2}\right),
\end{align*}
where $V=C^{-1/2}e$. We also have
\[
C^{-1/2}=\frac{1}{\sqrt{1-\varrho}}I_{n}+\left(\frac{1}{\sqrt{1+(n-1)\rho}}-\frac{1}{\sqrt{1-\rho}}\right)\frac{1}{n}\iota_{n}\iota_{n}^{\prime},
\]
such that
\begin{align*}
V_{i} & =\frac{1}{\sqrt{1-\varrho}}e_{i}+\left(\frac{1}{\sqrt{1+(n-1)\varrho}}-\frac{1}{\sqrt{1-\varrho}}\right)\bar{e}\\
 & =\frac{1}{\sqrt{1-\varrho}}\left(e_{i}-\bar{e}\right)+\frac{1}{\sqrt{1+(n-1)\varrho}}\bar{e},\\
\frac{\partial V_{i}}{\partial\varrho} & =\frac{1}{2}\frac{1}{\left(1-\varrho\right)^{3/2}}\left(e_{i}-\bar{e}\right)-\frac{1}{2}\frac{\left(n-1\right)}{\left(1+(n-1)\varrho\right)^{3/2}}\bar{e}.
\end{align*}
The score is therefore given by 
\[
-2\frac{\partial\ell(X)}{\partial\varrho}=\ensuremath{\left[\frac{(n-1)}{1+(n-1)\varrho}-\frac{(n-1)}{1-\varrho}\right]}+\sum_{i=1}^{n}W_{i}V_{i}\left[\frac{\left(e_{i}-\bar{e}\right)}{\left(1-\rho\right)^{3/2}}-\frac{\left(n-1\right)\bar{e}}{\left(1+(n-1)\varrho\right)^{3/2}}\right].
\]
We can use \citet[Theorem 5]{TongHansenArchakov:2024} to obtain the
following expression for the information matrix,
\begin{align*}
\mathcal{I}_{A}= & \Omega_{A}\left(K_{K}+\Upsilon_{K}^{e}\right)\Omega_{A}+\ensuremath{\tfrac{1}{4}E_{d}^{\prime}\Xi E_{d}}+\ensuremath{\tfrac{1}{2}E_{d}^{\prime}\Theta\Omega_{A}}+\ensuremath{\tfrac{1}{2}\Omega_{A}\Theta^{\prime}E_{d}}\\
= & \tfrac{1}{4}\ensuremath{A^{-2}}\left(1+C_{1}\right)+\tfrac{1}{4}C_{2}+\tfrac{1}{2}C_{3}\ensuremath{A^{-1}},
\end{align*}
where $A=1+\left(n-1\right)\rho$, and
\begin{align*}
C_{1} & =n_{}^{-1}\left(3\bar{\phi}-2-\bar{\psi}_{}\right)+\bar{\psi},\\
C_{2} & =\delta^{-2}n_{}^{-1}\left[3\bar{\phi}_{}-1+\left(\bar{\psi}+1\right)\left(n_{}-1\right)^{-1}\right],\\
C_{3} & =\delta^{-1}n^{-1}\left(\bar{\psi}_{}+2-3\bar{\phi}_{}\right).
\end{align*}
The results now follows from $\mathcal{I}_{\varrho}=\left(n-1\right)^{2}\mathcal{I}_{A}$,
which completes the proof.

\setcounter{table}{0}
\global\long\def\thetable{B.\arabic{table}}%
\newpage{}

\section{Supplementary Material}

\begin{table}[H]
\caption{Estimation of Marginal Volatility Models for Factors and Individual
Stocks}

\begin{centering}
\begin{footnotesize}
\begin{tabularx}{\textwidth}{p{1cm}YYYYYYYYYYY}
\toprule 
\midrule
    \multicolumn{7}{c}{Panel A: Factors} \\
          & $a_0 \ (\times 10^{-4})$    & $a_1$    & $b_0$    & $b_1$    & $b_2$    & $b_3$ \\
        \midrule
    \\[-0.2cm]         
    MKT   & 4.554 & -0.053 & -0.219 & 0.968 & -0.069 & 0.092 \\
    SMB   & -0.556 & -0.016 & -0.151 & 0.983 & -0.026 & 0.082 \\
    HML   & -2.740 & 0.021 & -0.120 & 0.992 & -0.012 & 0.101 \\
    RMW   & 0.968 & 0.011 & -0.084 & 0.993 & 0.014 & 0.056 \\
    CMA   & -0.243 & 0.038 & -0.087 & 0.993 & 0.006 & 0.060 \\
    UMD   & 2.764 & 0.083 & -0.131 & 0.991 & 0.011 & 0.112 \\
    XLB   & 2.255 & -0.015 & -0.112 & 0.986 & -0.049 & 0.063 \\
    XLE   & 2.482 & -0.020 & -0.116 & 0.987 & -0.037 & 0.080 \\
    XLF   & 5.168 & -0.038 & -0.176 & 0.978 & -0.055 & 0.106 \\
    XLI   & 3.381 & -0.016 & -0.156 & 0.978 & -0.056 & 0.073 \\
    XLK   & 5.958 & -0.052 & -0.210 & 0.969 & -0.056 & 0.093 \\
    XLP   & 3.478 & -0.062 & -0.236 & 0.966 & -0.059 & 0.091 \\
    XLU   & 2.902 & -0.026 & -0.179 & 0.976 & -0.027 & 0.091 \\
    XLV   & 4.105 & -0.031 & -0.244 & 0.961 & -0.063 & 0.081 \\
    XLY   & 5.186 & -0.022 & -0.168 & 0.978 & -0.046 & 0.090 \\
    \\[-0.2cm]     
    \midrule
    \multicolumn{7}{c}{Panel B: Individual Stocks} \\
          & $a_0 \ (\times 10^{-4})$    & $a_1$    & $b_0$    & $b_1$    & $b_2$    & $b_3$ \\
    \midrule
    \\[-0.2cm]     
    Mean  & 4.227 & -0.022 & -0.124 & 0.980 & -0.031 & 0.062 \\
    $Q_{\text{5}}$    & 0.084 & -0.063 & -0.287 & 0.950 & -0.050 & 0.022 \\
    $Q_{\text{25}}$   & 2.712 & -0.039 & -0.150 & 0.976 & -0.038 & 0.045 \\
    $Q_{\text{50}}$ & 3.996 & -0.022 & -0.110 & 0.984 & -0.030 & 0.061 \\
    $Q_{\text{75}}$   & 5.625 & -0.004 & -0.075 & 0.989 & -0.023 & 0.078 \\
    $Q_{\text{95}}$   & 8.896 & 0.022 & -0.040 & 0.995 & -0.011 & 0.102 \\
    \\[-0.2cm]  
\midrule
\bottomrule
\end{tabularx}
\end{footnotesize}

\par\end{centering}
{\small Note: Panel A presents the estimation results of the AR(1)-EGARCH
model for factors. Panel B shows the means and selected quantiles
of the cross-sectional estimations for individual stocks in a large
universe.\label{tab:EGARCHest}}{\small\par}
\end{table}

\begin{table}
\caption{Stocks in Empirical Analysis}

\begin{centering}
\begin{footnotesize}
\begin{tabularx}{\textwidth}{p{8cm}p{8cm}}
\toprule 
\midrule
    \multicolumn{1}{c}{Subindustry} & \multicolumn{1}{c}{Symbols} \\
    \midrule
    \multicolumn{2}{c}{\textbf{Energy Sector (10)}} \\
\\[-0.2cm]     
    Oil \& Gas Drilling & HP,NBR,RIG \\
    Oil \& Gas Equipment \& Services & HAL,NOV,SLB \\
    Integrated Oil \& Gas & CVX,STO,XOM \\
    Oil \& Gas Exploration \& Production & CNX,COG,EQT,RRC,SWN \\
    Oil \& Gas Storage \& Transportation & FRO,OKE,WMB \\
\\[-0.2cm] 
    \midrule
    \multicolumn{2}{c}{\textbf{Materials Sector (15)}} \\
\\[-0.2cm] 
    Fertilizers \& Agricultural Chemicals & CF,FMC,MOS \\
    Specialty Chemicals & ALB,ASH,CE,ECL,EMN,IFF,PPG,SHW \\
    Metal, Glass \& Plastic Containers & BLL,CCK,OI \\
    Paper \& Plastic Packaging Products \& Materials & AVY,IP,MWV,PKG,SEE \\
    Steel & ATI,CLF,NUE,STLD,TX,X \\
\\[-0.2cm] 
    \midrule
    \multicolumn{2}{c}{\textbf{Industrials Sector (20)}} \\
\\[-0.2cm] 
    Aerospace \& Defense & BA,GD,HRS,LMT,NOC,TDG,TXT,UTX \\
    Building Products & AOS,BLDR,IR,MAS \\
    Electrical Components \& Equipment & AME,AYI,EMR,ENS,ETN,ROK \\
    Construction Machinery \& Heavy Trans. Equip. & CAT,CMI,MTW,PCAR,TEX,WAB \\
    Industrial Machinery \& Supplies \& Components & DOV,FLS,IEX,ITT,ITW,PH,SNA,SWK,TKR \\
    Trading Companies \& Distributors & AIT,FAST,GWW,URI \\
    Environmental \& Facilities Services & ROL,RSG,SRCL \\
    Human Resource \& Employment Services & ADP,PAYX,RHI \\
    Air Freight \& Logistics & CHRW,EXPD,FDX,UPS \\
    Rail Transportation & CSX,NSC,UNP \\
    Cargo Ground Transportation & JBHT,ODFL,R \\
\\[-0.2cm]
    \midrule
    \multicolumn{2}{c}{\textbf{Consumer Discretionary Sector (25)}} \\
\\[-0.2cm] 
    Homebuilding & DHI,KBH,LEN,LEN,NVR,PHM \\
    Leisure Products & BC,HAS,MAT \\
    Apparel, Accessories \& Luxury Goods & FOSL,HBI,PVH,RL,UA,VFC \\
    Casinos \& Gaming & LVS,MGM,PENN,WYNN \\
    Hotels, Resorts \& Cruise Lines & CCL,EXPE,MAR,RCL \\
    Restaurants & CMG,DPZ,DRI,MCD,SBUX,YUM \\
    Specialized Consumer Services & HRB,SCI,WTW \\
    Broadline Retail & BIG,DDS,JWN,KSS \\
    Apparel Retail & ANF,FL,GPS,ROST,TJX,URBN \\
    Automotive Retail & AAP,AN,AZO,KMX,ORLY \\
\\[-0.2cm]  
    \midrule
    \multicolumn{2}{c}{\textbf{Consumer Staples Sector (30)}} \\
\\[-0.2cm]  
    Consumer Staples Merchandise Retail & COST,DLTR,TGT,WMT \\
    Packaged Foods \& Meats & CAG,CPB,GIS,K,MKC,MKC,SJM \\
    Household Products & CHD,CL,CLX,KMB,PG, \\
\\[-0.2cm]  
\midrule
\end{tabularx}
\end{footnotesize}

\par\end{centering}
{\small Note: Table continues on next page.\label{tab:CompanyName1}}{\small\par}
\end{table}

\addtocounter{table}{-1} 
\begin{table}
\caption{(cont.)}

\begin{centering}
\begin{footnotesize} 
\begin{tabularx}{\textwidth}{p{7cm}p{8cm}}
\toprule
\midrule
    \multicolumn{1}{c}{Subindustry} & \multicolumn{1}{c}{Symbols} \\
    \midrule
    \multicolumn{2}{c}{\textbf{Health Care Sector (35)}} \\
\\[-0.2cm]   
    Health Care Equipment & ABT,BAX,BDX,BSX,EW,MDT,SYK,TFX \\
    Health Care Supplies & ALGN,COO,XRAY \\
    Health Care Distributors & ABC,CAH,HSIC,MCK,PDCO \\
    Health Care Services & CI,CVS,DGX,DVA,LH \\
    Managed Health Care & CNC,HUM,MOH,UNH,WLP \\
    Biotechnology & AMGN,BIIB,GILD,INCY,REGN,VRTX \\
    Pharmaceuticals & BMY,LLY,MRK,NKTR,PFE,PRGO \\
    Life Sciences Tools \& Services & A,BIO,BIO,DHR,MTD,PKI,TECH,TMO,WAT \\
\\[-0.2cm] 
    \midrule    
    \multicolumn{2}{c}{\textbf{Financial Sector (40)}} \\
\\[-0.2cm]     
    Diversified Banks & CMA,FITB,KEY,PNC,USB,WFC \\
    Regional Banks & BBT,FHN,HBAN,MTB,RF,SNV,ZION \\
    Transaction \& Payment Processing Services & FIS,FISV,GPN,JKHY,MA \\
    Consumer Finance & ADS,AXP,COF,SLM \\
    Asset Management \& Custody Banks & AMG,BEN,BK,BLK,FII,NTRS,STT,TROW \\
    Investment Banking \& Brokerage & GS,MS,RJF,SCHW \\
    Financial Exchanges \& Data & CME,FDS,ICE,MKTX,NDAQ \\
    Insurance Brokers & AJG,BRO,MMC \\
    Life \& Health Insurance & AFL,LNC,MET,PFG,PRU,TMK,UNM \\
    Property \& Casualty Insurance & ACGL,AIZ,ALL,CINF,HIG,PGR \\
\\[-0.2cm] 
    \midrule
    \multicolumn{2}{c}{\textbf{Information Technology Sector (45)}} \\
\\[-0.2cm] 
    IT Consulting \& Other Services & ACN,CTSH,IBM,IT,UIS \\
    Application Software & ADBE,ADSK,ANSS,CDNS,CRM,INTU,SNPS,TYL \\
    Systems Software & MSFT,ORCL,SYMC \\
    Communications Equipment & CIEN,CSCO,FFIV,JNPR \\
    Technology Hardware, Storage \& Peripherals & AAPL,HPQ,NTAP,STX,WDC,XRX \\
    Electronic Equipment \& Instruments & CR,TDY,TRMB,ZBRA \\
    Electronic Manufacturing Services & IPGP,JBL,SANM \\
    Semiconductor Materials \& Equipment & AMAT,KLAC,LRCX,TER \\
    Semiconductors & ADI,AMD,INTC,MCHP,MPWR,MU,NVDA,SWKS,TXN \\
\\[-0.2cm] 
    \midrule
    \multicolumn{2}{c}{\textbf{Utilities Sector (55)}} \\
\\[-0.2cm] 
    Electric Utilities & AEP,DUK,ETR,FE,LNT,PNW,PPL,SO,XEL \\
    Multi-Utilities & AEE,CMS,CNP,DTE,ED,NI,SRE,WEC \\
\\[-0.2cm] 
\midrule
\bottomrule
\end{tabularx}
\end{footnotesize}

\par\end{centering}
{\small Note: This table presents the symbols of the companies used
in our empirical analysis, along with the names of their respective
sectors and subindustries.\label{tab:CompanyName2}}{\small\par}
\end{table}

\end{document}